\documentclass[11pt,reqno,tikz]{amsart}
\usepackage{amsmath,amsthm,amssymb,bm,bbm,dsfont,braket}
\usepackage[mathscr]{eucal}
\usepackage{amsaddr}
\usepackage{upgreek}
\usepackage[left=2cm,right=2cm,top=2cm,bottom=2cm]{geometry}
\usepackage{mathtools}\mathtoolsset{centercolon}
\usepackage{nicefrac}

\usepackage[shortlabels]{enumitem}
\usepackage{setspace}
\usepackage{graphicx}
\usepackage{verbatim}
\usepackage{float}
\usepackage{placeins}
\usepackage{array}
\usepackage{booktabs}
\usepackage{multicol,multirow,tabularx}
\usepackage{threeparttable}
\usepackage[update,prepend]{epstopdf}
\usepackage{amsfonts,amssymb,dsfont,xfrac,comment}
\usepackage[abs]{overpic}
\makeatletter
\def\adl@drawiv#1#2#3{%
	\hskip0
	\tabcolsep
	\xleaders#3{#2 0\@tempdimb #1{1}#2 0.5\@tempdimb}%
	#2\z@ plus1fil minus1fil\relax
	\hskip0\tabcolsep}
\newcommand{\cdashlinelr}[1]{%
	\noalign{\vskip\aboverulesep
		\global\let\@dashdrawstore\adl@draw
		\global\let\adl@draw\adl@drawiv}
	\cdashline{#1}
	\noalign{\global\let\adl@draw\@dashdrawstore
		\vskip\belowrulesep}}
\makeatother

\usepackage[dvipsnames]{xcolor}
\usepackage[hidelinks,linktocpage=true]{hyperref}
\hypersetup{
	unicode=false,          % non-Latin characters in Acrobat’s bookmarks
	pdftoolbar=true,        % show Acrobat’s toolbar?
	pdfmenubar=true,        % show Acrobat’s menu?
	pdffitwindow=false,     % window fit to page when opened
	pdfstartview={FitH},    % fits the width of the page to the window
	pdftitle={My title},    % title
	pdfauthor={Author},     % author
	pdfsubject={Subject},   % subject of the document
	pdfcreator={Creator},   % creator of the document
	pdfproducer={Producer}, % producer of the document
	pdfkeywords={keyword1} {key2} {key3}, % list of keywords
	pdfnewwindow=true,      % links in new PDF window
	colorlinks=true,        % false: boxed links; true: colored links
	linkcolor=Red,          % color of internal links (change box color with linkbordercolor)
	citecolor=ForestGreen,  % color of links to bibliography
	filecolor=Magenta,      % color of file links
	urlcolor=BlueViolet,    % color of external links
}
\usepackage{doi}
\usepackage{url}
\usepackage{caption, subcaption}
\usepackage{enumitem}
\usepackage{shadethm}
\usepackage{tikz,pgf}
\usepackage{pgfplots}
\pgfplotsset{compat=1.14}
\usepackage{witharrows}
\usepackage{tikzsymbols} % can be removed
\usepackage[abs]{overpic}

\usepackage{tcolorbox}

\usepackage{cleveref}

\makeatletter
\ifx\@NODS\undefined%

\let\mathbb=\mathds
\else%
\fi
\makeatother

\def\d{{\text {\rm d}}}

\DeclareMathOperator*{\argmax}{arg\,max}
\DeclareFontFamily{U}{mathx}{}
\DeclareFontShape{U}{mathx}{m}{n}{<-> mathx10}{}
\DeclareSymbolFont{mathx}{U}{mathx}{m}{n}
\DeclareMathAccent{\widehat}{0}{mathx}{"70}
\DeclareMathAccent{\widecheck}{0}{mathx}{"71}

\DeclareMathOperator{\Tr}{Tr}

\DeclareMathOperator{\spec}{spec}

\DeclareMathOperator{\e}{\mathrm{e}}
\def\I{\mathbf{1}}

\newcommand{\proj}[1]{\left\{#1\right\}} % projection \{ \}

\newcommand{\be}{{\mathbf e}}

\newcommand{\tr}{\operatorname{Tr}}

\newcommand{\Renyi}{{}R\'{e}nyi{ }}

\def\cF{{\mathcal{F}}}

\def\0{{\mathbf{0}}}
\def\1{{\mathbf{1}}}
\def\2{{\mathbf{2}}}
\def\3{{\mathbf{3}}}
\def\4{{\mathbf{4}}}
\def\5{{\mathbf{5}}}
\def\6{{\mathbf{6}}}

\def\7{{\mathbf{7}}}
\def\8{{\mathbf{8}}}
\def\9{{\mathbf{9}}}

\def\bbR{\mathbb{R}}

\def\be{\begin{equation}}
	\def\ee{\end{equation}}
\def\bea{\begin{eqnarray}}
	\def\eea{\end{eqnarray}}

\theoremstyle{plain}
\newshadetheorem{theo}{Theorem}[section]
\newshadetheorem{prop}[theo]{Proposition}%[section]
\newshadetheorem{lemm}[theo]{Lemma}%[section]
\newshadetheorem{coro}[theo]{Corollary}%[section]
\theoremstyle{definition}
\newtheorem{defn}{Definition} %[section]

\theoremstyle{remark}
\newtheorem{remark}{Remark}%[section]

\makeatletter
\newcommand{\opnorm}{\@ifstar\@opnorms\@opnorm}
\newcommand{\@opnorms}[1]{%
	$\left|\mkern-1.5mu\left|\mkern-1.5mu\left|
	#1
	\right|\mkern-1.5mu\right|\mkern-1.5mu\right|$
}
\newcommand{\@opnorm}[2][]{%
	\mathopen{#1|\mkern-1.5mu#1|\mkern-1.5mu#1|}
	#2
	\mathclose{#1|\mkern-1.5mu#1|\mkern-1.5mu#1|}
}
\makeatother

\usepackage{tikz}
\usetikzlibrary{arrows.meta} % to control arrow size
\tikzset{>={Latex[length=4,width=4]}} % for LaTeX arrow head
\usetikzlibrary{calc}

\colorlet{mylightblue}{blue!5!white}
\colorlet{mydarkblue}{blue!30!black}
\colorlet{myblue}{blue!50!black}
\colorlet{myred}{red!50!black}
\colorlet{mydarkred}{red!30!black}
\colorlet{mydarkgreen}{green!30!black}

\newcommand{\sh}{\kern-0.08em$^\textbf{\#}$\hspace{-3pt}}
\renewcommand{\b}{\kern-0.06em$\flat$}

\begin{document}
	
	%%%%% How to disable amsart.cls to capitalize the article title?
	\let\origmaketitle\maketitle
	\def\maketitle{
		\begingroup
		\def\uppercasenonmath##1{} % this disables uppercasing title
		\let\MakeUppercase\relax % this disables uppercasing authors
		\origmaketitle
		\endgroup
	}
	%%%%%%%%%%%%%%%%%%%%%%%%%%%%%%%%
	
	\title{\bfseries \Huge{ 
			Layer Cake Representations for Quantum Divergences
	}}
	
	\author{ \normalsize 
		{Po-Chieh Liu$^\ast$},
		{Christoph Hirche$^\dagger$}, and
		{Hao-Chung Cheng$^\ast$}
	}
	\address{\small  	
		$^\ast$ Department of Electrical Engineering and Graduate Institute of Communication Engineering,\\ National Taiwan University, Taipei 106, Taiwan (R.O.C.) \\
		$^\dagger$ Institute for Information Processing (tnt/L3S), Leibniz Universit\"at Hannover, Germany
	}
	
	% \email{\href{mailto:haochung.ch@gmail.com}{haochung.ch@gmail.com}}
	
	\date{\today}
	\begin{abstract}
		Defining suitable quantum extensions of classical divergences often poses a challenge due to the non-commutative nature of quantum information. In this work, we propose a new approach via what we call the layer cake representation. The resulting quantum \Renyi and $f$-divergences are then proven to be equivalent to those recently defined via integral representations. Nevertheless, the approach can provide several insights. We give an alternative proof of the integral representation of the relative entropy by Frenkel and prove a conjecture regarding a trace expression for the \Renyi divergence. Additionally, we give applications to error exponents in hypothesis testing, a new Riemann-Stieltjes type integral representation and a variational representation. 
	\end{abstract}

	\maketitle
	%\vspace{-2.5em}
	\tableofcontents
	
	\section{Introduction} \label{sec:introduction}
	
	%% Motivation
	% Quantum entropies and divergences play a substantial role in characterizing operationally relevant quantities.
	% Most of the known quantum divergences can be regarded as a kind of \emph{quantum generalizations} of certain statistical metrics.
	
	Quantum divergences are a kind of distance measure that quantifies how distinguishable two quantum systems are.
	When the states describing the quantum systems are mutually commutative, the quantum divergences reduce to some statistical measure of the associated eigenvalues of the states. 
	Even though there have been a wealth of statistical measures in information theory, how to define operationally relevant quantum divergences that inherit desirable properties (such as the data-processing inequality) is nonetheless highly nontrivial.
	
	Some of the known quantum divergences are defined in an axiomatic way~\cite{AD15}, \cite[\S 4]{Tom16}, following the analog approach for the well known classical \Renyi divergence~\cite{Ren61}.
	Some are introduced in an operational way.
	For example, the measured \Renyi divergence $\widecheck{D}_{\alpha}(\rho\Vert\sigma)$ is the maximal \Renyi divergence between the output distributions among all quantum measurements (i.e., a quantum-classical channel) performed on the states $(\rho,\sigma)$
	\cite{Don86, HP91}.
	On the other hand, the maximal \Renyi divergence $\widehat{D}_{\alpha}(\rho\Vert\sigma)$ is the minimal \Renyi divergence between all input distributions that can prepare the states $(\rho,\sigma)$ via a classical-quantum channel \cite{Matsumoto} (see also \cite{Hia19}).
	
	To see how one can generalize a statistical metric to its quantum version in an algebraic way, let us first recall the definitions of the two most fundamental classical divergences for probability distributions $P$ and $Q$---the \Renyi divergence $D_{\alpha} (P\Vert Q)$, $\alpha \in (0,1)\cup(1,\infty)$ \cite{Ren61} and the $f$-divergence $D_{f} (P\Vert Q)$ for a convex function $f$ with $f(1) = 0$:
	\begin{align*}
		D_{\alpha} (P\Vert Q) &:= \frac{1}{\alpha-1} \log Q_{\alpha} (P\Vert Q),
		\quad Q_{\alpha} (P\Vert Q) := \mathds{E}_{Q}\left[ \left( \frac{\d P}{\d Q} \right)^{\alpha} \right];
		\\
		D_{f} (P\Vert Q) &:= \mathds{E}_{Q}\left[ f \left( \frac{\d P}{\d Q} \right) \right],
	\end{align*}
	where $P$ is required to be absolutely continuous with respect to $Q$ (denoted as $P\ll Q$) unless $\alpha \in (0,1)$, and $Q_{\alpha}(P\Vert Q)$ is called the quasi \Renyi divergence.
	The term $\nicefrac{\d P}{\d Q}$ is the Radon--Nikodym derivative for $P\ll Q$, and it is also known as \emph{likelihood ratio} in statistics and information theory.
	Thereby, the essential question is:
	\[
	\textit{How to properly generalize a function of a likelihood ratio to the quantum scenario?}
	\]
	
	However, there is no unique way to define the Radon--Nikodym derivative for noncommutative quantum states $\rho$ and  $\sigma$.
	One of the attempts is Araki's \emph{relative modular operator} 
	\begin{align}
		\Delta_{\rho,\sigma}: X \mapsto \sum_{\lambda \in \texttt{spec}(\rho)}  \sum_{\mu \in \texttt{spec}(\mu)}  
		\frac{\lambda}{\mu}P_{\lambda} X Q_{\mu} 
	\end{align} given the spectral decompositions $\rho = \sum_{\lambda \in \texttt{spec}(\rho)} \lambda P_{\lambda}$ and $\sigma = \sum_{\mu \in \texttt{spec}(\sigma)} \mu Q_{\mu} $ \cite{Ara76, Ara77}.
	Petz then introduced the so-called (quasi) Petz--\Renyi divergence and the corresponding $f$-divergence: 
	\cite{Pet85, Pet86, Petz_book_1993, HMP+11, Hia18}
	\begin{align*}
		\bar{Q}_{\alpha}(\rho\Vert\sigma)
		&:= \Tr\left[\sigma^{1/2} \cdot (\Delta_{\rho,\sigma})^{\alpha} \sigma^{1/2} \right] = \Tr[\rho^{\alpha} \sigma^{1-\alpha}],
		\\
		\bar{D}_f(\rho\Vert\sigma)
		&:= \Tr\left[\sigma^{1/2} \cdot f(\Delta_{\rho,\sigma}) \sigma^{1/2} \right],
		\\
		f(\Delta_{\rho,\sigma}) &\colon X \mapsto  \sum_{\lambda \in \texttt{spec}(\rho)} \sum_{\mu \in \texttt{spec}(\mu)} f\left( \nicefrac{\lambda}{\mu}\right) P_{\lambda} X Q_{\mu}.
	\end{align*}
	
	Another way is to consider the tracial version of Pusz and Woronowicz's matrix geometric mean \cite{PW75} or, more generally, Kubo and Ando's matrix mean \cite{KA80} as follows,
	\begin{align*}
		\widehat{Q}_{\alpha}(\rho\Vert\sigma) 
		&= \Tr\left[ \sigma^{1/2} \left(\sigma^{-1/2} \rho \sigma^{-1/2}\right)^\alpha \sigma^{1/2} \right],
		\\
		\widehat{D}_{f}(\rho\Vert\sigma) 
		&= \Tr\left[ \sigma^{1/2} f\left(\sigma^{-1/2} \rho \sigma^{-1/2}\right) \sigma^{1/2} \right].
	\end{align*}
	Here, $\sigma^{-1/2} \rho \sigma^{-1/2}$ may be considered as a kind of noncommutative quotient.
	Interestingly, for $\alpha \in (0,2]$, the induced \Renyi divergence coincides with the maximal \Renyi divergence $\widehat{D}_{\alpha}(\rho\Vert\sigma) $ \cite{Matsumoto, BS82}.
	
	The sandwiched \Renyi divergence \cite{MDS+13, WWY14} is another example with operational meaning, which can also be expressed in terms of the noncommutative quotient \cite{CG22, CG22b}:\footnote{We note that there are alternative ways of writing the sandwiched \Renyi divergence in terms of other different noncommutative $L_p$ norms \cite{Jen18, BTS18} (see also \cite[\S 3.3]{Hia21}).}
	\begin{align*}
		\widetilde{Q}_{\alpha} (\rho\Vert\sigma)
		&:= \left\| \sigma^{-1/2} \rho \sigma^{-1/2} \right\|_{\alpha, \sigma}^{\alpha},
	\end{align*}
	where $\left\| X \right\|_{\alpha, \sigma} := (\Tr[ ( \sigma^{\nicefrac{1}{2\alpha}} X \sigma^{\nicefrac{1}{2\alpha}} )^{\alpha} ]  )^{1/\alpha}$ is the $\sigma$-weighted norm for self-adjoint operators (See e.g., \cite[\S 3]{Kos84}).
	
	A completely different approach was recently taken in~\cite{hirche2023quantum}, where instead an integral representation for $f$-divergences from~\cite{sason2016f} was generalized to the quantum setting. This entirely circumvented the need of giving a quantum generalization of the Radon-Nikodym derivative. 
	The resulting quasi \Renyi divergence is defined as follows, 
	\begin{align}\label{Eq:int-quasi-petz}
		Q_{\alpha}(\rho\Vert\sigma)
		= 1+\alpha (\alpha-1) \int_1^\infty \gamma^{\alpha-2} E_{\gamma}(\rho \Vert \sigma) + \gamma^{-\alpha -1} E_{\gamma}(\sigma \Vert \rho) \, \d \gamma,
	\end{align}
	where $E_{\gamma}(A\Vert B) := \Tr\left[ (A-\gamma B)_+\right]$ is the quantum Hockey-Stick divergence.
	More generally, the corresponding $f$-divergence was defined by, 
	\begin{align}\label{Eq:int-f-div}
		D_{f} (\rho\Vert \sigma)
		&= \int_1^\infty  f''(\gamma) E_{\gamma}(\rho\Vert\sigma) 
		+ \gamma^{-3} f''(\gamma^{-1}) E_{\gamma}(\sigma\Vert\rho) \, \d \gamma
		= \int_0^\infty  f''(\gamma) E_{\gamma}(\rho\Vert\sigma)\, \d \gamma- \int_0^1  f''(\gamma)(1-\gamma) \d\gamma.
	\end{align}

	\bigskip
	
	%% Our work
	In this work, we introduce a new approach to generalizing functions of the likelihood ratio to the quantum setting.
	This will result in seemingly new \Renyi and $f$-divergences. 
	We subsequently establish their equivalence to the definitions given in Equations~\eqref{Eq:int-quasi-petz} and~\eqref{Eq:int-f-div}.

	Our starting point is the so-called \emph{layer cake representation} of any
	nonnegative measurable function $g$
	\cite[Theorem 1.13]{LM01}:
	\begin{align*}
		g(x) &= \int_0^\infty \proj{ g(x) > \gamma}  \d \gamma,
		\\
		\{ \text{Event} \} &:= \begin{cases} 1, & \text{{Event} is true,} \\
			0, & \text{otherwise}.
		\end{cases}
	\end{align*}
	The layer cake representation of the powered $g$ also holds,
	\begin{align*}
		g(x)^{\alpha} &= \alpha \int_0^\infty \gamma^{\alpha-1} \proj{g(x) > \gamma} \, \d \gamma, \quad \alpha > 0.
	\end{align*}
	Viewing the measurable function $g$ as the Radon--Nikodym derivative (i.e., $g = \frac{\d P}{\d Q}$), the classical quasi \Renyi divergence admits the form:
	\begin{align*}
		Q_{\alpha}(P\Vert Q) = \mathds{E}_{Q}\left[ \left( \frac{\d P}{\d Q} \right)^{\alpha} \right]
		= \alpha \int_0^\infty \gamma^{\alpha-1} \Pr\nolimits_{Q}\proj{\frac{\d P}{\d Q} > \gamma}  \d \gamma.
	\end{align*}
	This then leads us to proposing the \emph{quantum layer cake \Renyi divergence}:
	\begin{tcolorbox}[size = small, colback=orange!1.5!white, colframe=orange, halign=center, center, boxrule=0.5pt]
		\begin{align} \label{eq:defn:Renyi}
			D_{\alpha} (\rho\Vert\sigma)
			:= \frac{1}{\alpha-1} \log Q_{\alpha}(\rho\Vert\sigma),
			\quad
			Q_{\alpha}(\rho\Vert\sigma)
			:= \alpha \int_0^\infty \gamma^{\alpha-1} \Tr\left[ \sigma  \proj{ \rho > \gamma \sigma }\right] \d \gamma, \quad \alpha > 0,
		\end{align}
	\end{tcolorbox}
	\noindent where $\proj{ \rho > \gamma \sigma }$ means the orthogonal projection onto the positive part of $\rho - \gamma \sigma$.
	
	Similarly, via the layer cake representation with a convex function $f$, i.e.,
	\begin{align*}
		f(g(x)) &= f(0) + \int_0^\infty \proj{ g(x) > \gamma} \, \d f(\gamma)
	\end{align*}
	(provided it exists),
	we define the corresponding quantum $f$-divergence for states $\rho\ll\sigma$ (i.e., the support of $\rho$ is contained in that of $\sigma$):
	\begin{tcolorbox}[size = small, colback=orange!1.5!white, colframe=orange, halign=center, center, boxrule=0.5pt]
		\begin{align} \label{eq:defn:f}
			D_f(\rho\Vert\sigma)
			:= f(0) + \int_0^\infty \Tr\left[ \sigma \proj{ \rho > \gamma \sigma}\right] \, \d f(\gamma)
			= f(0) + \int_0^\infty f'(\gamma) \Tr\left[ \sigma \proj{ \rho > \gamma \sigma}\right] \, \d \gamma,
		\end{align}
	\end{tcolorbox}
	\noindent Here, the definition is in terms of the Riemann–Stieltjes integral;
	the last equality gives a Riemann integral if $f$ is differentiable and the integrand is Riemann integrable.
	A special case for $f(x)=x\log x$ then yields a quantum relative entropy:
	\begin{align}
		D_{x\log x}(\rho\Vert\sigma)
		= 1 + \int_0^\infty \log \gamma \Tr\left[ \sigma \proj{\rho > \gamma \sigma} \right] \, \d \gamma.
	\end{align}
	
	It is not immediately clear yet if the proposals in \eqref{eq:defn:Renyi} and \eqref{eq:defn:f} satisfy the data-processing inequality.
	Our first result is to show that the (quasi) \Renyi divergence coincides with the integral representation given in \cite{hirche2023quantum}, as given in Equations~\eqref{Eq:int-quasi-petz}, and the quantum $f$-divergence in~\eqref{Eq:int-f-div} when $f$ is twice differentiable. This justifies using the same notation and will be formally proven in Section~\ref{sec:relation}.
	Hence, the definitions given in \eqref{eq:defn:Renyi} and \eqref{eq:defn:f} satisfy the data-processing inequality and several other desirable properties proven in~\cite{hirche2023quantum,beigi2025some}.
	
	For $\alpha \in (0,1)$, Ref.~\cite{beigi2025some} also provides an alternative formula for the \Renyi divergence,
	\begin{align}
		Q_{\alpha}(\rho\Vert\sigma) = \alpha \int_0^\infty \Tr\left[ \left( \frac{\I}{\sqrt{\rho+t\I}} \sigma \frac{\I}{\sqrt{\rho+t\I}} \right)^{1-\alpha} \right] - (1+t)^{\alpha-1} \, \d t + 1, 
		\quad \alpha \in (0,1),
	\end{align}
	and employs the Araki--Lieb--Thirring inequality to show that
	\begin{align} \label{eq:relation_Petz}
		Q_{\alpha}(\rho\Vert\sigma)
		\leq \bar{Q}_{\alpha} (\rho\Vert\sigma), \quad \forall\,\alpha \in (0,1).
	\end{align}
	Furthermore, the authors conjectured a trace expression for $\alpha > 1$ and provided a proof for all integers $\alpha \geq 2$ and the full range when $\rho$ is a pure state.
	In this paper, we resolve the conjecture and show that, 
	\begin{align}
		Q_{\alpha}(\rho\Vert\sigma) = (\alpha-1) \int_0^\infty \Tr\left[ \left( \frac{\I}{\sqrt{\sigma+t\I}} \rho \frac{\I}{\sqrt{\sigma+t\I}} \right)^{\alpha} \right]  \, \d t, 
		\quad \forall, \alpha >1,
	\end{align}
	which then complements the relation to the sandwiched quasi divergence:
	\begin{align}
		Q_{\alpha}(\rho\Vert\sigma)
		\leq \widetilde{Q}_{\alpha} (\rho\Vert\sigma), \quad \forall\, \alpha >1 .
	\end{align}
	Additionally, for $\alpha \to 1$, we employ the recently established operator layer cake theorem~\cite{preparation} to show that
	\begin{align}
		\lim_{\alpha \to 1} D_{\alpha}(\rho\Vert\sigma)
		= \int_1^\infty \frac{1}{\gamma} \Tr\left[ \rho \left( \proj{\rho\geq \gamma \sigma} - \proj{\sigma > \gamma \rho} \right) \right] \,\d\gamma
		= D(\rho\Vert\sigma),
	\end{align}
	where the last equality gives the Umegaki relative entropy $ D(\rho\Vert\sigma)
	\coloneq \Tr\left[  \rho \left( \log \rho - \log \sigma \right) \right]$~\cite{Ume62} and hence provides an alternative proof to Frenkel's integral formula~\cite{frenkel2022integral}.
	
	\bigskip
	
	We then use the new representations in Equations~\eqref{eq:defn:Renyi} and~\eqref{eq:defn:f} to establish some interesting properties. Notably, the proposed layer cake representation in \eqref{eq:defn:f} naturally leads to a \emph{quantum Markov inequality}: for nondecreasing $f$ satysfying $f(0)\geq 0$, % and any scalar $c>0$,
	\begin{align}
		\Tr\left[ \sigma \proj{ \rho > c \sigma} \right] \leq 
		\frac{D_f(\rho\Vert\sigma)}{f(c)}, \quad c > 0.
	\end{align}
	In terms of the \Renyi divergence in \eqref{eq:defn:Renyi}, we obtain the following error estimates for quantum hypothesis testing using the quantum Neyman--Pearson tests: for $a \in \mathds{R}$,
	\begin{align}
		\begin{cases}
			\Tr\left[ \sigma \proj{\rho > \e^a \sigma} \right] 
			\leq \e^{- a - (\alpha-1)\left[ a - D_\alpha(\rho\|\sigma)\right]}, & \alpha\in(0,1)\cup(1,\infty) 
			\\
			\Tr[\rho\proj{\rho > \e^a \sigma} ] 
			\leq  \e^{-(\alpha-1)\left[a-D_\alpha(\rho\|\sigma)\right]}, & \alpha \in (0, 1)
			\\
			\tr[\rho\proj{\rho > \e^a \sigma }] 
			\geq 1 - \e^{-(\alpha-1)\left[a - D_\alpha(\rho\|\sigma)\right]}, & \alpha > 1
		\end{cases}.
	\end{align}
	By appropriately choosing the threshold $a$, the above bounds imply (and sharpen) the quantum Chernoff bound \cite{ACM+07, ANS+08} in view of the relation \eqref{eq:relation_Petz}.
	
	Lastly, we show that the proposed divergences in \eqref{eq:defn:Renyi} and \eqref{eq:defn:f} admit a  \textit{Riemann--Stieltjes integral representation}:
	\begin{align}
		D(\rho\Vert\sigma)
		&= \int_0^\infty \log \gamma \, \d P(\gamma)
		= \int_0^\infty \gamma \log \gamma \, \d Q(\gamma),
		\\
		Q_{\alpha}(\rho\Vert\sigma)
		&= \int_0^\infty \gamma^{\alpha} \, \d Q(\gamma),
		\\
		D_{f}(\rho\Vert\sigma)
		&= \int_0^\infty f(\gamma) \, \d Q(\gamma),
	\end{align}
	where we call $P$ and $Q$ the \emph{Riemann--Stieltjes cumulative distributions} induced from the states $(\rho,\sigma)$ satisfying $\rho\ll\sigma$:
	\begin{align}
		P(\gamma) 
		= \Tr\left[ \rho \proj{\rho \leq \gamma \sigma}\right],
		\quad
		Q(\gamma) 
		= \Tr\left[ \sigma \proj{\rho \leq \gamma \sigma}\right].
	\end{align}
	We consider the above formulas the most intuitive version in terms of a quantum generalization; we refer the readers to Figures~\ref{figure:RS-integral_quantum_increasing} and \ref{figure:increasing}
	for interpretations.
	With this representation, we immediately obtain the \emph{quantum $f$-divergence duality}:
	\begin{tcolorbox}[size = small, colback=orange!1.5!white, colframe=orange, halign=center, center, boxrule=0.5pt, width = 0.6\linewidth]
		\begin{align*}
			D_f(\rho \Vert \sigma) = \sup_{g: [0,\infty)\to \mathds{R}}  \big\{
			\mathds{E}_{P} [g] +  \mathds{E}_{Q}[f^\star \circ g] \big\},
		\end{align*}
	\end{tcolorbox}
	\noindent where $f^\star$ is the convex conjugate of $f$.
	This then leads to new variational formulas for important quantities such as the trace distance and the chi-square divergence.
	
	%% Structure
	This paper is organized as follows.
	Section~\ref{sec:preliminaries} presents some preliminaries and important properties regarding projections.
	In Section~\ref{sec:relation}, we formally define the quantum $f$-divergence and \Renyi divergence, and prove the relations to the known quantities.
	In Section~\ref{sec:exponents}, we show how the proposed quantum \Renyi divergence applies to the error exponents using quantum Neyman--Pearson tests.
	In Section~\ref{sec:RS-integral}, we show that our proposals admit certain Riemann–Stieltjes integral representations.
	In Section~\ref{sec:f-duality}, we prove the quantum $f$-divergence duality.
	We conclude the paper in Section~\ref{sec:conclusions}.

	%%%%%%%%%%%%%%%%%%%%%%%%%%%%%%%%%%%%%%%%%%%%%%%%%%%%%%%%%%%
	%% Summary Big Table
	\newpage
	\begin{table*}[ht]
		% Prevent stretching above displays inside summary box
		\setlength{\abovedisplayskip}{1.5ex}
		\setlength{\belowdisplayskip}{\abovedisplayskip}
		% Set-up box around summary, height is calibrated manually
		\fbox{%
			\begin{minipage}[t][49.3\baselineskip][t]{\textwidth} %% Maximum hight that fits on a single page in this format.
				\centerline{{\large\textbf{Summary}}}
				\vspace{1.5\baselineskip}
				
				{ \scriptsize %\footnotesize %\small % fontsize
					\begin{multicols}{2}

						\begin{samepage}
							\textbf{Hockey-Stick divergence}:
							%For Hermitian $A$ and $B$,
							$E_\gamma(A \Vert B):= \Tr\left[ \left( A - \gamma B \right)_+\right]$.
							% \vspace{-2em}
							% \textit{Properties}:
							\begin{enumerate}[1., leftmargin=2em]
								\item 
								Data-processing inequality:
								for any positive and trace-preserving map $\mathscr{N}$,
								$E_\gamma(\mathscr{N}(A) \Vert \mathscr{N}(B)) \leq E_\gamma(A \Vert B)$ .
								
								\item Positivity: $E_\gamma(A\| B)\geq 0$. 
								
								\item 
								Convexity: 
								$\gamma \mapsto E_\gamma(A \Vert B)$ is convex nonincreasing.
								
								\item 
								$E_\gamma(A \Vert B) = \Tr[A - A\wedge \gamma B] = \frac12\left[ \left\|A-B\right\|_1 +\Tr[A-B] \right]$.
								
								\item 
								Semi-derivatives:
								$E_{\gamma^+}'(A \Vert B) =  -\Tr\left[ B \proj{\gamma B < A }\right]$, 
								\\
								% {\color{white}spac}
								(Lem.~\ref{Lem:HS-derivative})
								{\color{white}space}\,
								$E_{\gamma^-}'(A \Vert B) =  -\Tr\left[ B \proj{\gamma B \leq A }\right]$,
								\begin{align*}
									% E_{\gamma^+}'(A \Vert B) &=  -\Tr\left[ B \proj{\gamma B < A }\right],
									% \\
									% E_{\gamma^-}'(A \Vert B) &=  -\Tr\left[ B \proj{\gamma B \leq A }\right],
									% \\
									-\int_0^\infty E_{\gamma}'(A\Vert B) \, \d \gamma = \Tr[A], \quad B>0.
								\end{align*}
							\end{enumerate}
						\end{samepage}
						
						\begin{samepage}
							\textbf{Quantum relative entropy}:
							For non-zero $A,B\geq0$, $A\ll B$,
							\begin{align*}
								&D(A \Vert B) 
								\coloneq \Tr\left[ A (\log A - \log B) + B-A\right] 
								\\
								&\,\,\,\overset{\text{\cite{frenkel2022integral}}}{=}\int_{-\infty}^\infty \frac{\d t}{|t|(t-1)^2} \tr\left[ ((1-t) A+t B)_- + A-B\right]
								\\
								&\!\!\!\!\overset{\text{Prop.~\ref{prop:algernative_proof_Frenkel2}}}{=} \int_1^\infty \frac{1}{\gamma}  E_\gamma (\rho \Vert \sigma) + \frac{1}{\gamma^2} E_\gamma (\sigma \Vert \rho) \, \d \gamma
								\\
								&\!\!\!\!\overset{\text{Prop.~\ref{prop:algernative_proof_Frenkel}}}{=}  \int_1^\infty \frac{1}{\gamma}  \Tr\left[ A \left( \left\{ A > \gamma B\right\} - \left\{ B > \gamma A \right\} \right) \right] \d \gamma + \Tr[B-A]
								\\
								&\,\,\,\,= \Tr[A] + \int_0^\infty \log \gamma \Tr\left[ B  \left\{  A > \gamma B \right\} \right]\d \gamma
								\\
								&\!\!\!\!\overset{\text{Prop.~\ref{prop:relative_entropy_formula2}}}{=} 
								\!\!\!\!\!- \!\int_0^\infty \!\!\! \log \gamma \, \d \tr\left[A\{A>\gamma B\}\right]
								= \int_0^\infty \! \log \gamma \, \d \tr\left[A\{A\leq\gamma B\}\right].
							\end{align*}
						\end{samepage}
						
						% \begin{samepage}
							\textbf{\Renyi divergence}:
							Let $A,B\geq 0$.
							Define \begin{align*}
								D_{\alpha}(A\Vert B) &\coloneq \frac{1}{\alpha-1}  \log Q_{\alpha}(A \Vert B), \\
								Q_{\alpha}(A \Vert B) &\coloneq 1 + (\alpha-1)H_{\alpha}(A \Vert B).
							\end{align*}

							\underline{For $\alpha \in (0,1)$, or for $\alpha \in (1,\infty)$ and $A\ll B$},
							\begin{align*}
								&Q_{\alpha}(A \Vert B) 
								=
								\alpha \int_{0}^{\infty} \gamma^{\alpha-1}  \Tr\left[ B  \left\{  A > \gamma B  \right\} \right] \d\gamma
								\\  &\!\!\!\overset{\text{Prop.~\ref{prop:f-divergence_RS}}}{=} \!\!\!\!-  \int_{0}^{\infty} \gamma^{\alpha} \, \d \Tr\left[ B \proj{ A > \gamma B} \right]
								= \int_{0}^{\infty} \gamma^{\alpha} \, \d \Tr\left[ B \proj{ A \leq \gamma B} \right]
								\\
								&\!\!\!\overset{\text{Lem.~\ref{lemm:change-of-measure}}}{=} \!\!\!\!-  \int_{0}^{\infty} \!\! \gamma^{\alpha-1} \d \Tr\left[ A \proj{ A > \gamma B} \right]
								=\! \int_{0}^{\infty} \!\!\gamma^{\alpha-1}  \d \Tr\left[ A \proj{ A \leq \gamma B} \right]
								\\
								&\!\!\!\overset{\text{Thm.~\ref{theorem:f_divergence_layer-cake}}}{=} \alpha (\alpha-1) \int_1^\infty \gamma^{\alpha-2} E_{\gamma}(A\Vert B) + \gamma^{-\alpha -1} E_{\gamma}(B\Vert A) \, \d \gamma + 1.
							\end{align*}
							
							\underline{For $0<\alpha<1$},
							\begin{align*}
								&Q_\alpha(A\Vert B) 
								= Q_{1-\alpha}(B\Vert A)
								\\
								&\overset{\text{Lem.~\ref{lemm:alpha-representation}}}{=} (1-\alpha) \int_0^\infty \gamma^{\alpha-2} \tr[A\{A <\gamma B \}] \d\gamma 
								% &= (1-\alpha) \int_0^\infty \gamma^{-\alpha} \tr[A\{\gamma A \leq B \}] \d\gamma
								\\
								&\overset{\text{\cite{beigi2025some}}}{=} -\alpha \int_0^\infty \Tr\left[ \left( \frac{\I}{\sqrt{A+t\I}} B \frac{\I}{\sqrt{A+t\I}} \right)^{1-\alpha} \right] \!- (1+t)^{\alpha-1} \, \d t +\!1.
								% &&\text{}
							\end{align*}
							
							\underline{For $\alpha>1$ and $A\ll B$},
							\begin{align*}
								Q_\alpha(A\Vert B) &\overset{\text{Lem.~\ref{lemm:alpha-representation}}}{=} (\alpha-1) \int_0^\infty \gamma^{\alpha-2} \tr[A\{A>\gamma B\}] \d\gamma 
								\\ &\overset{\text{Lem.~\ref{lemm:alpha-representation}}}{=} (\alpha-1) \int_0^\infty \gamma^{-\alpha} \tr[A\{\gamma A>B\}] \d\gamma
								\\
								&\overset{\text{Thm.~\ref{Theo:trace-renyi-PO}}}{=} (\alpha-1) \int_0^\infty \Tr\left[ \left( \frac{\I}{\sqrt{B+t\I}} A \frac{\I}{\sqrt{B+t\I}} \right)^{\alpha} \right]\d t.
							\end{align*}
							% \end{samepage}
						
						\begin{samepage}
							\textbf{Relations to other (quasi) divergences}:
							\begin{align*}
								\begin{cases}
									Q_{\alpha}(A \Vert B ) \leq \Tr\left[ A^{\alpha} B^{1-\alpha} \right] & \alpha \in (0,1) \quad \text{\cite{beigi2025some}}
									\\
									Q_{\alpha}(A \Vert B ) \leq \Tr\left[ \left( B^{\frac{1-\alpha}{2\alpha}}A  B^{\frac{1-\alpha}{2\alpha}} \right)^{\alpha} \right]  
									& \alpha >1 \quad \text{ \scriptsize (Prop.~\ref{prop:relation_sandwiched})}
								\end{cases}.
							\end{align*}
						\end{samepage}
						
						%%%%%%%%%%%%%% Second Column %%%%%%%%%%%%%%
						
						\begin{samepage}
							\textbf{Quantum $f$-divergence}: For $B>0$ and convex $f$ with $f(1)=0$,
							\begin{align*}
								D_f(A\Vert B)
								&\coloneq f(0)\Tr[B] + \int_0^\infty \Tr\left[ B \proj{A>\gamma B} \right] \, \d f(\gamma)
								\\
								&\!\!\!\!\!\!\!\!\!\!\!\!\!\!\!\!\!\!\!\!\!\!\!\!\overset{\text{Prop.~\ref{prop:genearl_f}}}{=} f(0) \Tr[B]+ \int_0^\infty \Tr\left[ B \frac{1}{B+t\I} \rho \frac{1}{B+t\I} f'\left( A \frac{1}{Bt\I} \right) \right] \, \d t
								\\
								&\!\!\!\!\!\!\!\!\!\!\!\!\!\!\!\!\!\!\!\!\!\!\!\!\overset{\text{Prop.~\ref{prop:f-divergence_RS}}}{=} -\int_0^\infty \!\!\!\! f(\gamma) \, \d \Tr\left[ B \proj{A>\gamma B} \right] = \int_0^\infty \!\!\!\!f(\gamma) \, \d \Tr\left[ B \proj{A\leq\gamma B} \right]
								\\
								&= f(0)\Tr[B] + \int_0^\infty f'(\gamma) \Tr\left[ B \proj{A>\gamma B} \right] \, \d \gamma
								\\
								&\!\!\!\!\!\!\!\!\overset{\text{Thm.~\ref{theorem:f_divergence_layer-cake}}}{=} \int_1^\infty  f''(\gamma) E_{\gamma}(\rho\Vert\sigma) 
								+ \gamma^{-3} f''(\gamma^{-1}) E_{\gamma}(\sigma\Vert\rho) \, \d \gamma.
								% \\
								% &= \int_0^\infty  f''(\gamma) E_{\gamma}(\rho\Vert\sigma)\, \d \gamma.
							\end{align*}
						\end{samepage}
						
						\begin{samepage}
							\textbf{Quantum $f$-divergence duality}:
							Let $f$ be a convex differentiable function and $f(0)=0$, and $f^\star(y) \coloneq \sup_{x \in \texttt{dom}(f)} \left\{ y x - f(x) \right\}$ be its convex conjugate.
							\textit{Riemann--Stieltjes cumulative distribution functions} for $\rho\!\ll\!\sigma$:
							\begin{align*}
								P(\gamma) \coloneq \Tr[\rho\proj{\rho\leq \gamma \sigma}],
								\;\;
								Q(\gamma) \coloneq \Tr[\sigma\proj{\rho\leq \gamma \sigma}],
								\;\;\gamma \in [0,\infty).
							\end{align*}
							\begin{equation*}
								D_f(\rho \Vert \sigma) = \sup_{g: [0,\infty)\to \mathds{R}}  \big\{
								\mathds{E}_{P} [g] +  \mathds{E}_{Q}[f^\star \circ g] \big\}
								\quad \text{(Thm.~\ref{theorem:f-duality})}
							\end{equation*}
						\end{samepage}

						\textbf{Noncommutative minimum}:
						For $A,B\geq 0$.
						\begin{align*}
							A \wedge B
							&\coloneq \argmax\limits_{H=H^\dagger} \left\{ \Tr{H} : H\leq A, H\leq B \right\}
							\\
							&=A \{A\leq B \} + B \{B < A  \}
							= \frac12\left[ A + B - |A-B| \right],
							\\
							\!\!\!\Tr\left[ A \wedge B \right]
							&= \inf_{0\leq T\leq \I} \!\Tr[A(\I\!-\!T)] + \Tr[B T]
							=\frac{\Tr[A \!+\! B] - \left\| A \!-\! B\right\|_1}{2}
							\\
							&\overset{\text{\cite{preparation}}}{=} \int_{0}^1 \Tr\left[ A \left\{ u A < B\right\} \right] \d u
							\geq \Tr\left[ A \cdot \mathrm{D} \log [A+B](B)] \right]
							\\
							&\geq \Tr\left[ A (A+B)^{-\frac12} B (A+B)^{-\frac12}\right].
						\end{align*}
						
						% \textbf{Quantum Neyman--Pearson test}:
						% Let $\!A,B\!$ be Hermitian.
						% Define $\!\proj{\!A \!>\! \gamma B}\!$ as the projection onto positive part of $A\!-\!\gamma B$.
						% \begin{itemize}[leftmargin=2em]
							% \item
							% The map $\gamma \mapsto \proj{A > \gamma B}$ (resp.~$\gamma \mapsto \proj{A \geq \gamma B}$) is piecewise analytic except in which with singular $A-\gamma B$, right (resp.~left) continuous, and right (resp.~left) differentiable with a bounded right (resp.~left) derivative.
							
							% \item 
							% $\Tr[A\proj{A< \gamma B}]$,
							% $\Tr[A\proj{A \leq \gamma B}]$
							% are nondecreasing in $\gamma$.
							% \end{itemize}

						\begin{samepage}
							\textbf{Logarithm}: For positive definite operators $A,B>0$,
							\begin{align*}
								\log A &- \log B
								= \int_0^\infty (B+t\I)^{-1}(A-B)(A+t\I)^{-1} \d t
								\\
								&\overset{\text{Prop.~\ref{prop:algernative_proof_Frenkel}}}{=} \int_1^\infty \left( \left\{ A > \gamma B\right\} - \left\{ B > \gamma A \right\} \right)\frac{1}{\gamma} \, \d\gamma
								\\
								&\overset{\text{Prop.~\ref{prop:relative_entropy_formula2}}}{=} - \int_0^\infty \log \gamma \, \d \left\{ A > \gamma B \right\} 
								= \int_0^\infty \log \gamma \, \d \left\{ A \leq \gamma B \right\}.
							\end{align*}
						\end{samepage}
						
						\begin{samepage}
							\textbf{Logarithmic derivative}: For $A>0$ and Hermitian $B$,
							\begin{align*}
								\mathrm{D} \log [A](B)
								&\overset{\text{\cite{Lie73}}}{=}  \int_0^1 (A+t\I)^{-1} B (A+t\I)^{-1}
								\\
								&\overset{\text{\cite{SBT16}}}{=} \int_0^\infty A^{-\frac12 - \frac{\mathrm{i}t}{2}} B A^{-\frac12 + \frac{\mathrm{i}t}{2}} \, \d \beta_0(t)
								\\
								&\overset{\text{\cite{preparation}}}{=} \int_0^\infty \{ u A < B \} \, \d u - \int_{-\infty}^0 \{ uA > B \} \, \d u.
							\end{align*}
						\end{samepage}
						
						For $A\geq0$, $B>0$, and Lebesgue-integrable function $h$,
						\begin{align*} 
							\int_0^\infty \proj{ A > \gamma B } h(\gamma) \, \d \gamma
							% &= \int_0^\infty (B+t\I)^{-1} A (B+t\I)^{-1} h\left( A (B+t\I)^{-1} \right) \d t.
							&\overset{\text{\cite{preparation}}}{=} \int_0^\infty \frac{1}{B+t\I} A \frac{1}{B+t\I} h\left( A \frac{1}{B+t\I} \right) \d t.
						\end{align*} 
						
					\end{multicols}
				} % end of font size
			\end{minipage}
		}
	\end{table*}
	\newpage

	\begin{figure}[ht]
		\centering
		\begin{overpic}[width=\linewidth]{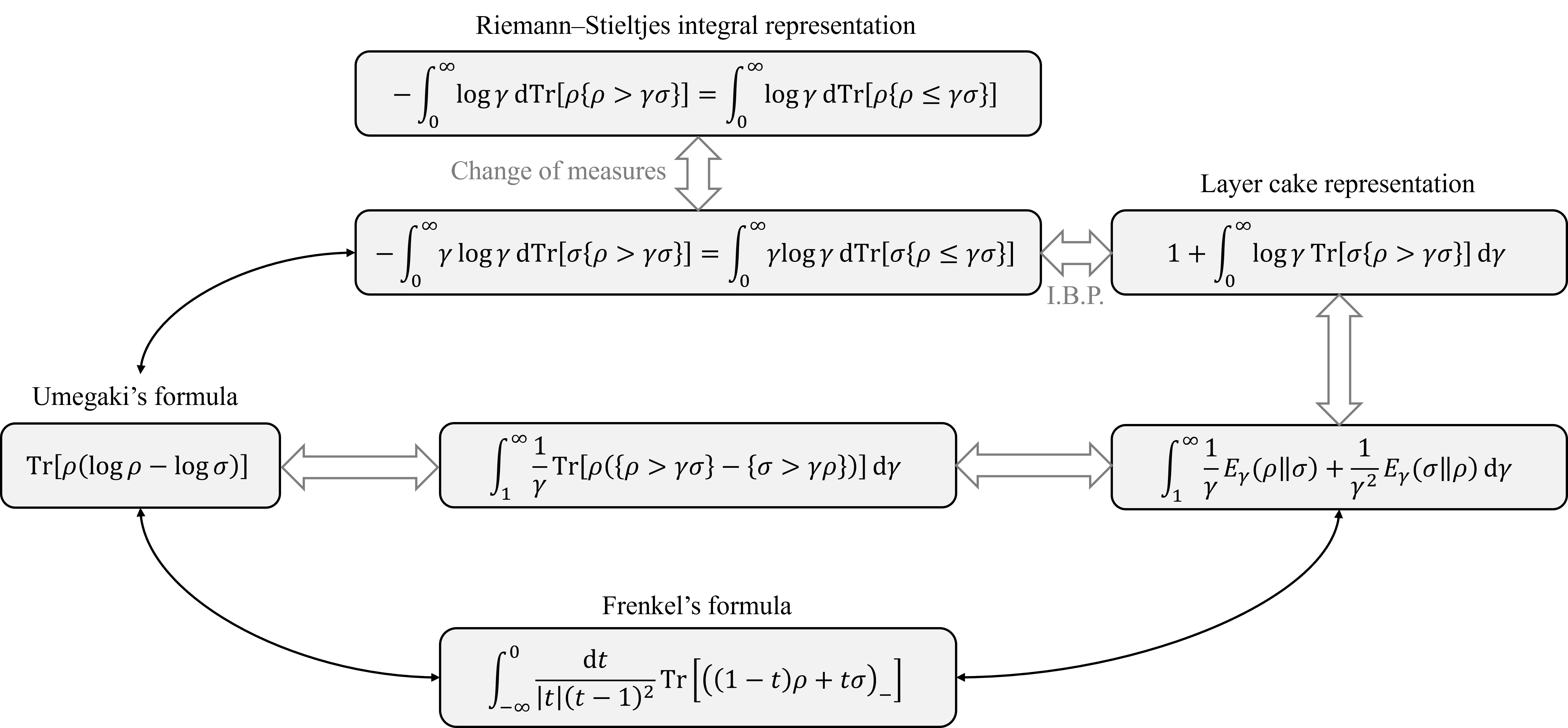}
			\put(79,104){\footnotesize \cite{Ume62}}
			\put(79,38){\footnotesize \cite{frenkel2022integral}}
			\put(255,38){\footnotesize \cite{frenkel2022integral}}
			\put(375,38){\footnotesize \cite{hirche2023quantum}}
			
			\put(101,91){\footnotesize Prop.~\ref{prop:algernative_proof_Frenkel}}
			
			\put(70,155){\footnotesize Prop.~\ref{prop:relative_entropy_formula2}}
			
			\put(434,115){\footnotesize Thm.~\ref{theorem:f_divergence_layer-cake}}
			
			% \put(165,175){\footnotesize (Change of measures)}
			\put(235,175){\footnotesize Lemma~\ref{lemm:change-of-measure}}
		\end{overpic}
		\caption{Schematic relations between the known formulas of the quantum relative entropy $D(\rho\Vert\sigma)$
		} \label{figure:relative_entropy}
	\end{figure} 
	
	\begin{figure}[ht]
		\centering
		\begin{overpic}[width=\linewidth]{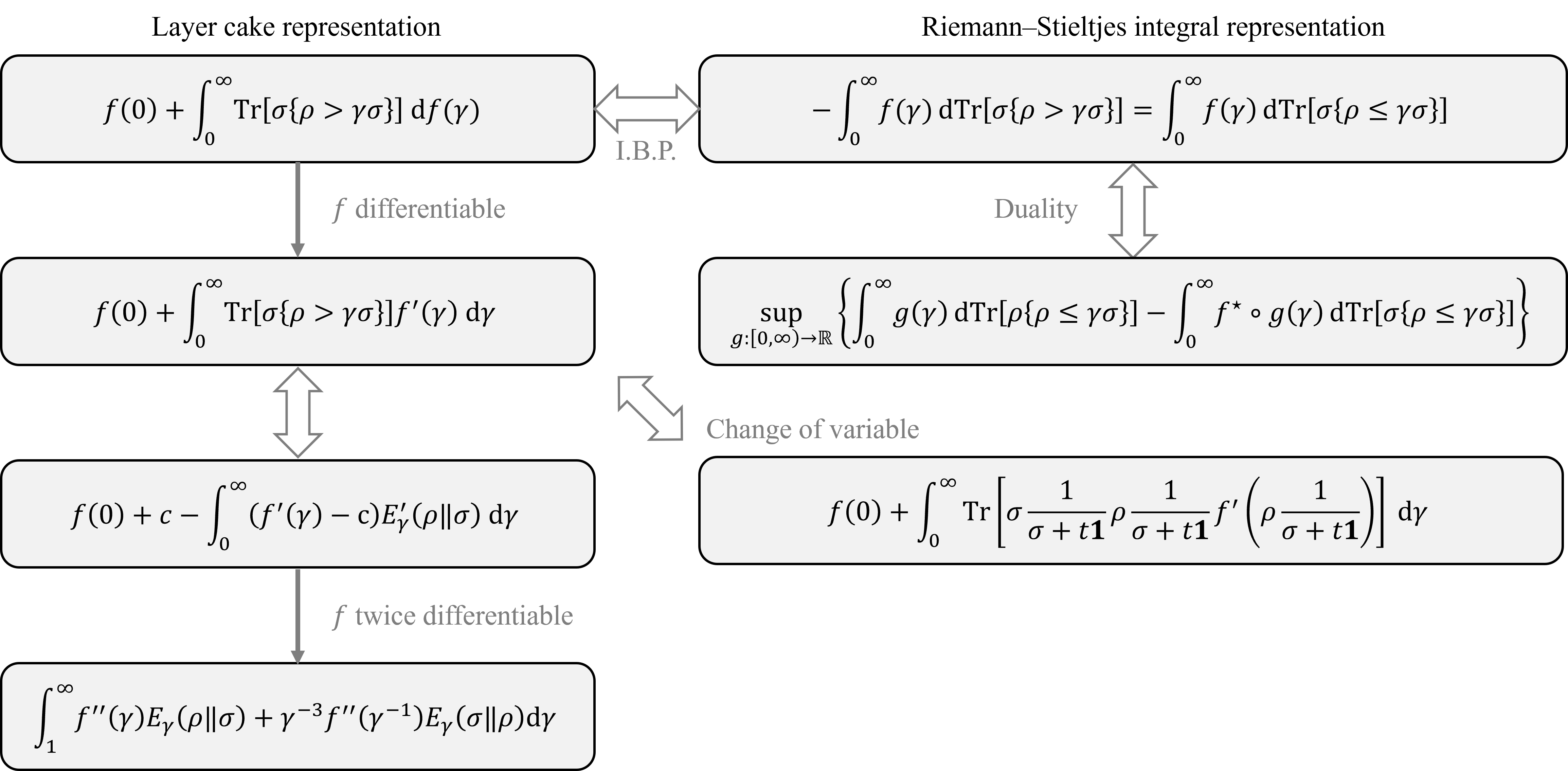}
			% \put(200,27){\footnotesize \cite{hirche2023quantum}}
			
			\put(370,177){\footnotesize Thm.~\ref{theorem:f-duality}}
			
		\end{overpic}
		\caption{Schematic relations between the known formulas of the quantum $f$-divergence $D_{f}(\rho\Vert\sigma)$
		} \label{figure:f-divergence}
	\end{figure} 
	
	%%%%%%%%%%%%%%%%%%%%%%%%%%%%%%%%%%%%%%%%%%%%%%%%%%%%%%%%%%%
	
	\section{Preliminaries and basic properties} \label{sec:preliminaries}
	
	In this section, we give some of the necessary definitions and their properties. A central quantity will be the quantum Hockey Stick divergence~\cite{sharma2012strong, hirche2023QDP}, 
	\begin{align}
		E_\gamma(A\|B) = \tr(A-\gamma B)_+.
	\end{align}
	For an extensive list of its properties we refer to~\cite[Section 2.2]{hirche2023quantum}\footnote{In~\cite{hirche2023quantum} and occasionally elsewhere in the literature, the Hockey Stick divergence is defined as $E_\gamma(A\|B) \coloneq \tr(A-\gamma B)_+ - (\tr(A-\gamma B))_+$. For normalized states and $\gamma\geq1$ this coincides with our definition which originated in~\cite{sharma2012strong, hirche2023QDP}.}. In particular, the Hockey Stick divergence satisfies data processing, is monotonically decreasing in~$\gamma$, convex in the operators and recovers the trace distance for $\gamma=1$. The divergence has several applications in quantum information theory, i.e. as a strong converse bound~\cite{sharma2012strong}, for differential privacy~\cite{hirche2023QDP, cheng2024sample} and as the underlying basis of the aforementioned $f$-divergences~\cite{hirche2023quantum}. 
	
	In the following, we give additional properties, in particular with regard to its derivative. We begin however with some preliminary observations regarding the projector that achieves the Hockey Stick divergence. 
	\begin{lemm} \label{lemm:differentiability_projection}
		Let $A$ and $B$ be Hermitian.
		% and we denote the spectrum $\spec\left\{B^{-1/2} A B^{-1/2} \right\}$ by $\mathcal{S}$.
		Then, the map $\gamma \mapsto \proj{A > \gamma B}$ is piecewise analytic except in which with singular $A-\gamma B$, right continuous, and right differentiable with a bounded right derivative.
		
		On the other hand, the map $\gamma \mapsto \proj{A \geq \gamma B}$ is piecewise analytic except when singular $A-\gamma B$, left continuous, and left differentiable with a bounded left derivative.
	\end{lemm}
	\begin{proof}
		First, since the map $\gamma \mapsto A-\gamma B$ is analytic on the whole real axis, there exist analytic eigen-pairs $\{\lambda_i(\gamma), |v_i(\gamma)\rangle\}_i$ such that, the spectral decomposition 
		\begin{align} \label{eq:analytic-eigen}
			A - \gamma B
			&= \sum_i \lambda_i (\gamma) |v_i(\gamma)\rangle \langle v_i(\gamma)|,
		\end{align}
		holds for each $\gamma \in \mathds{R}$ \cite[Theorem 1.8, p.~70]{Kat95}.
		Note here that the eigenvalues $\{\lambda_i(\gamma)\}_i$ for a fixed $\gamma$ may not be in increasing or decreasing order.
		
		Second, we claim that each eigenvalue $\lambda_i(\gamma)$ is nonincreasing in $\gamma$.
		Indeed,
		\begin{align}
			\frac{\d}{\d \gamma} \lambda_i(\gamma)
			&= \frac{\d}{\d \gamma} \langle v_i(\gamma) | (A - \gamma B) | v_i(\gamma) \rangle
			\\
			&=\left( \frac{\d}{\d \gamma} \langle v_i(\gamma) | \right) \cdot (A - \gamma B) | v_i(\gamma) \rangle
			+  \langle v_i(\gamma) | (- B) | v_i(\gamma) \rangle
			+ \langle v_i(\gamma) | (A - \gamma B) \cdot  \left( \frac{\d}{\d \gamma} | v_i(\gamma) \rangle \right)
			\\
			&= \lambda_i(\gamma) \left( \frac{\d}{\d \gamma} \langle v_i(\gamma) | \right) \cdot  | v_i(\gamma) \rangle
			+ \langle v_i(\gamma) | (- B) | v_i(\gamma) \rangle
			+ \lambda_i(\gamma)  \langle v_i(\gamma) | \cdot  \left( \frac{\d}{\d \gamma} | v_i(\gamma) \rangle \right)
		\end{align}
		Note that the first and the third term add up to 
		\begin{align}
			\lambda_i(\gamma) \cdot \frac{\d}{\d \gamma} \left( \langle v_i(\gamma) | v_i(\gamma) \rangle \right)
			&= \lambda_i(\gamma) \cdot \frac{\d}{\d \gamma}(1) = 0.
		\end{align}
		Hence, each eigenvalue is nonincreasing, i.e., $\frac{\d}{\d \gamma} \lambda_i(\gamma) = -\langle v_i(\gamma) | B | v_i(\gamma) \rangle \leq 0$.
		
		Then, we have
		\begin{align}
			\proj{A>\gamma B} &= \sum_i \mathbf{1}_{\{ \lambda_i(\gamma) > 0 \}} |v_i(\gamma)\rangle \langle v_i(\gamma)|, \label{eq:analytic_projection}
			\\
			\proj{A \geq \gamma B} &= \sum_i \mathbf{1}_{\{ \lambda_i(\gamma) \geq 0 \}} |v_i(\gamma)\rangle \langle v_i(\gamma)|.
		\end{align}
		Since each $\lambda_i(\gamma)$ is nonincreasing, the indicator function $\mathbf{1}_{\{ \lambda_i(\gamma) > 0 \}}$ is right continuous in $\gamma$, which in turn implies that the map $\gamma \mapsto \proj{A>\gamma B}$ is right continuous.
		The left continuity of $\gamma \mapsto \proj{A\geq \gamma B}$ holds by the similar argument.
		
		Finally, note that each eigenprojection $|v_i(\gamma)\rangle \langle v_i(\gamma)|$ is analytic on the whole real axis, and thus it has a bounded derivative on any finite closed interval.
		From \eqref{eq:analytic_projection}, we see that $\proj{A>\gamma B}$ is a finite sum of analytic eigenprojections multiplied by an indicator function $\mathbf{1}_{\{ \lambda_i(\gamma) > 0 \}}$, and hence $\proj{A>\gamma B}$ is piecewise analytic; see Figure~\ref{figure:RS-integral_quantum} for an example.
		Together with the right continuity, the right differentiability with a bounded right derivative is then evident.
		The same argument applies to the piecewise analyticity and the left differentiability of $\proj{A\geq \gamma B}$.
	\end{proof}

	\begin{figure}[ht!]
		\centering
		\includegraphics[width=0.5\linewidth]{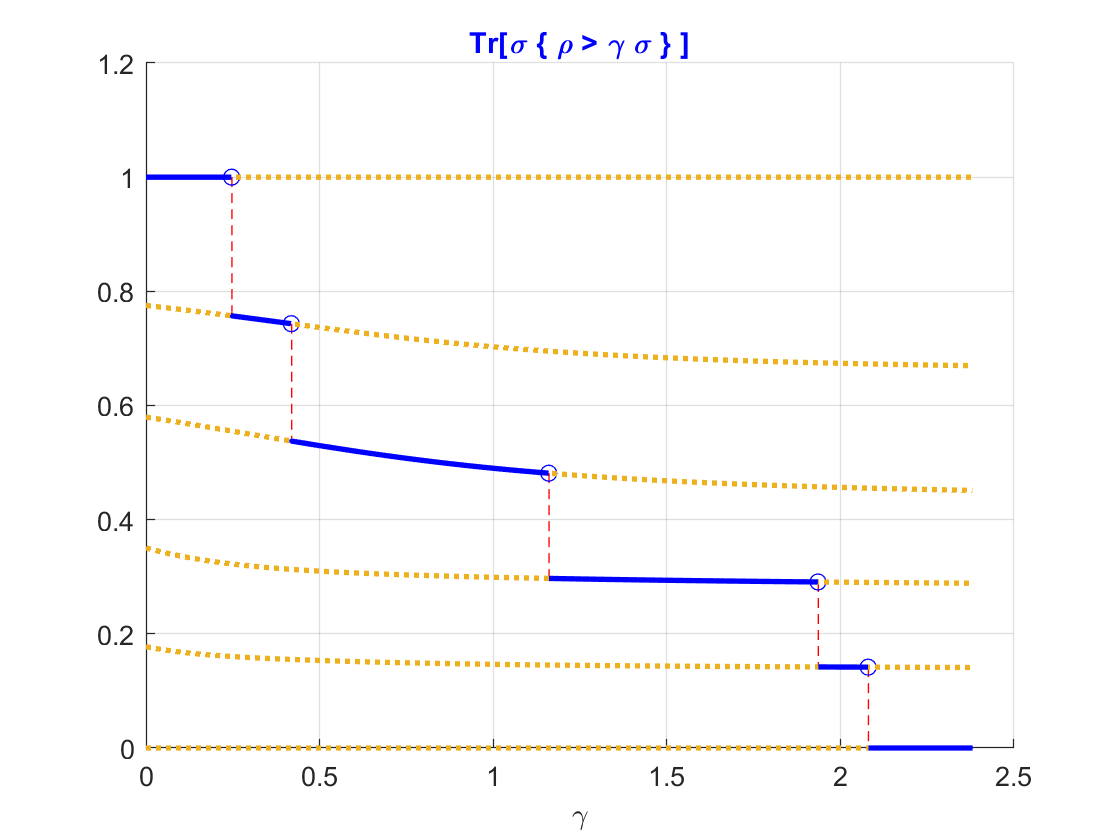}
		\caption{A numerical demonstration of the piecewise analytic map $\gamma \mapsto \Tr\left[ \sigma \{ \rho > \gamma \sigma\} \right]$ in {\color{blue}blue} color.
			Each dashed vertical line in {\color{red}red} corresponds to a jump at some $\gamma$ such that $\rho - \gamma \sigma$ is singular.
			Let $\{(\lambda_i(\gamma), |v_i(\gamma)\rangle)\}_i$ be analytic eigen-pairs for $\rho-\gamma \sigma$; see \eqref{eq:analytic-eigen}.
			Each dotted line in {\color{Orange}orange} color is an analytic map $\gamma \mapsto \Tr\left[ \sigma \sum_{i\in\mathcal{I}} |v_i(\gamma)\rangle \langle v_i(\gamma) | \right]$ for some $\gamma$-independent index set $\mathcal{I}$ (including the empty set), where the caldinality $|\mathcal{I}|$ plays the role of the rank of the projection.
			Each dotted {\color{Orange}orange} line covers one portion of the {\color{blue}blue} piecewise curve $\Tr[\sigma \proj{\rho>\gamma \sigma}] = \Tr[ \sigma \sum_i \mathbf{1}_{\{ \lambda_i(\gamma) > 0 \}} |v_i(\gamma)\rangle \langle v_i(\gamma)| ] $ for a certain range of $\gamma$.
		} \label{figure:RS-integral_quantum}
	\end{figure} 
	
	\begin{figure}[ht!]
		\centering
		\includegraphics[width=0.5\linewidth]{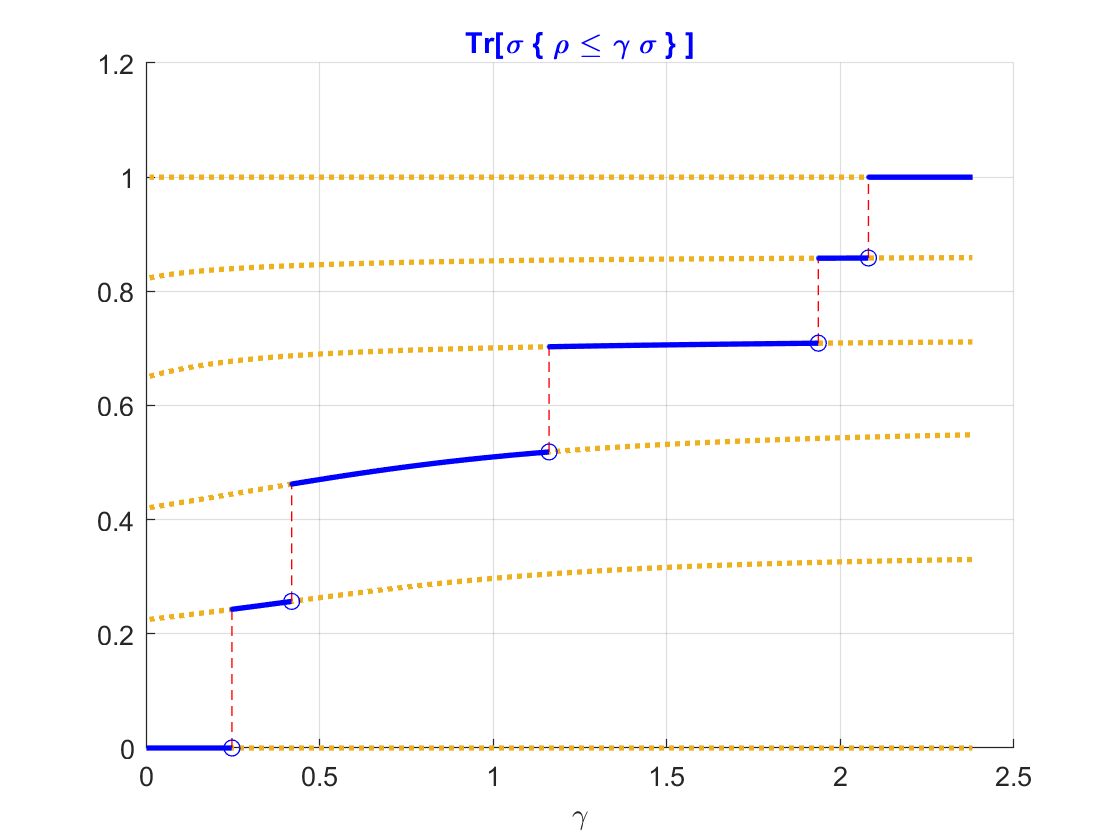}
		\caption{A numerical demonstration of the piecewise analytic map $\gamma \mapsto \Tr\left[ \sigma \{ \rho \leq \gamma \sigma\} \right]$ in {\color{blue}blue} color.
			The interpretation of the plot is similar to that of Figure~\ref{figure:RS-integral_quantum}.
			The map $\gamma \mapsto \Tr\left[ \sigma \{ \rho \leq \gamma \sigma\} \right]$ is a right-continuous monotone increasing map from $0$ to $1$.
		} \label{figure:RS-integral_quantum_increasing}
	\end{figure} 
	
	With the above we can discuss properties of the derivative of the Hockey-Stick divergence. 
	\begin{lemm}\label{Lem:HS-derivative}
		Let $A$ and $B$ be Hermitian. Then, the map $\gamma \mapsto E_{\gamma}(A\Vert B)$ is piecewise analytic except for $\gamma$ such that $A-\gamma B$ is singular.
		The right and left derivatives are
		\begin{align}
			\partial_+ E_{\gamma}'(A \Vert B) =  -\Tr\left[ B \proj{\gamma B < A }\right] \\
			\partial_- E_{\gamma}'(A \Vert B) =  -\Tr\left[ B \proj{\gamma B \leq A }\right].
		\end{align}
		
		Moreover, if $A\geq 0$ and $B>0$, then we have the Lebesgue integral
		\begin{align}
			-\int_0^\infty \frac{\d}{\d\gamma} E_{\gamma}(A\Vert B) \, \d \gamma = \Tr[A], \quad B>0. \label{eq:integral_HS}
		\end{align}
	\end{lemm}
	\begin{proof}
		The piecewise analyticity directly follows from Lemma~\ref{lemm:differentiability_projection} and the definition $E_{\gamma}(A\Vert B) = \Tr\left[\left( A - \gamma B \right)_+ \right]$.
		
		For $\delta$ such that $A-\delta B$ is not singular, we have
		\begin{DispWithArrows}[displaystyle]
			\left.\frac{\d}{\d \gamma} \Tr\left[ ( A - \gamma B )_+ \right]\right|_{\gamma = \delta}
			= - \Tr\left[ B \left\{ A  > \delta B \right\} \right]
			= - \Tr\left[ B \left\{ A  \geq \delta B \right\} \right],
		\end{DispWithArrows}
		by Lemma~\ref{lemm:HP14} below with $f(x) = x_+$ and 
		\begin{align}
			f'(x) = \begin{cases} 
				1 & x>0
				\\
				0 & x< 0
			\end{cases}.
		\end{align}
		
		For $\delta$ such that $A-\delta B$ is not singular, 
		\begin{DispWithArrows}[displaystyle]
			\left.\frac{\d}{\d \gamma} \Tr\left[ ( A - \gamma B )_+ \right]\right|_{\gamma = \delta^+}
			&= - \Tr\left[ B \left\{ A  > \delta^+ B \right\} \right]
			\notag
			\Arrow{\textnormal{$\delta \mapsto \left\{ A  > \delta B \right\}$ is right} \\ \textnormal{continuous by Lemma~\ref{lemm:differentiability_projection}}}
			\notag
			\\
			&= - \Tr\left[ B \left\{ A  > \delta B \right\} \right],
		\end{DispWithArrows}
		and 
		\begin{DispWithArrows}[displaystyle]
			\left.\frac{\d}{\d \gamma} \Tr\left[ ( A - \gamma B )_+ \right]\right|_{\gamma = \delta^-}
			&= - \Tr\left[ B \left\{ A  \geq \delta^- B \right\} \right]
			\Arrow{\textnormal{$\delta \mapsto \left\{ A  \geq \delta B \right\}$ is left} \\ \textnormal{continuous by Lemma~\ref{lemm:differentiability_projection}}}
			\notag
			\\
			&= - \Tr\left[ B \left\{ A  \geq \delta B \right\} \right].
		\end{DispWithArrows}
		
		Lastly, by Theorem~\ref{theo:Dlog_formula} below, we have
		\begin{align*}
			-\int_0^\infty E_{\gamma}'(A\Vert B) \, \d \gamma
			&= \int_0^\infty \Tr\left[ B \left\{ \gamma B < A \right\} \right] \, \d \gamma
			\\
			&= \Tr\left[ B \cdot \mathrm{D} \log [B] (A) \right]
			\\
			&= \Tr\left[ \mathrm{D} \log [B] (B) \cdot A \right]
			\\
			&= \Tr[A].
		\end{align*}
	\end{proof}

	\begin{lemm}\label{lemm:monotone_in_gamma}
		For any $A, B \geq 0$, the maps
		\begin{align}
			\gamma \mapsto \Tr\left[ A \left\{ A < \gamma B \right\} \right],
			\quad
			\gamma \mapsto \Tr\left[ A \left\{ A \leq \gamma B \right\} \right]
		\end{align}
		are nondecreasing on $\gamma >0$.
	\end{lemm}
	\begin{proof}
		Since the Hockey-stick divergence $\Tr\left[ ( A - \gamma B )_+ \right]$ is convex nonincreasing in $\gamma >0$ (see Lemma~\ref{lemm:Hockey-stick_convex_nonincreasing} below), both the left and right derivatives are nondecreasing.
		Invoking Lemma~\ref{Lem:HS-derivative} and  switching the roles of $A$ and $B$, and substituting $\delta \leftarrow 1/\gamma$, the proof is concluded.
	\end{proof}
	
	\begin{lemm}[\hspace{-0.3pt}{\cite[Theorem 3.23]{HP14}}\footnote{Ref.~{\cite[Theorem 3.23]{HP14}} states for continuously differentiable function $f$. The same argument holds for one-sided differentiable functions as well.}] \label{lemm:HP14}
		Let $A$ and $B$ be Hermitian, and let $f$ be a continuously one-sided differentiable function defined around the eigenvalues of $A+t_0 B$.
		Then,
		\begin{align}
			\left.\frac{\d}{\d t} \Tr\left[ f(A+tB) \right]\right|_{t=t_0}
			= \Tr\left[ B \cdot  f'(A+t_0 B)\right].
		\end{align}
	\end{lemm}
	
	\begin{lemm}\label{lemm:Hockey-stick_convex_nonincreasing}
		Let $A,B\geq 0$.
		The Hockey-Stick divergence $E_{\gamma}(\rho\Vert\sigma)$ is convex nonincreasing in $\gamma \geq 0$.
	\end{lemm}
	\begin{proof}
		Let $0\leq \gamma \leq \gamma'$. We immediately have $A-\gamma B \geq A - \gamma' B$.
		Note that $(A-\gamma B)_+$ and $(A-\gamma' B)_+$ might not be comparable, but by the functional calculus (or the eigenvalue mapping theorem), the ordered eigenvalue are pointwise comparable, i.e.,~$\lambda^{\downarrow}_i(A-\gamma B) \geq \lambda^{\downarrow}_i(A-\gamma' B)$.
		After adding the positive eigenvalues, we have that $\gamma \mapsto E_{\gamma}(\rho\Vert\sigma)$ is nonincreasing.
		
		From the identity
		\begin{align}
			\Tr\left[ (A-\gamma B)_+ \right] 
			= \frac12 \left\| A - \gamma B \right\|_1 + \Tr\left[ A - \gamma B \right],
		\end{align}
		the convexity of the Hockey-stick divergence in $\gamma$ follows from the convexity of Schatten $1$-norm and the linearity of trace.
	\end{proof}
	
	We will need the following two theorems from the recent work~\cite{preparation}. These will crucially connect the integral representations with more traditional expressions. 
	\begin{theo}[Operator layer cake {\cite[Theorem B.1]{preparation}}] \label{theo:Dlog_formula}
		For any positive definite operator $A$ and any self-adjoint operator $B$ on a finite-dimensional Hilbert space,
		the following representation holds:
		\begin{align} \label{eq:Dlog_formula}
			\mathrm{D} \log[A](B) = \int_{0}^{\infty} \{ uA < B \} \d u\,-\,\int_{-\infty}^{0} \{ uA > B \} \d u,
		\end{align}
		where $ \mathrm{D} \log[A](B)=\frac{\d}{\d t}\log(A+tB)|_{t=0} $ is the directional derivative of the operator logarithm,
		and $ \{ uA < B \} \equiv \{ B - uA > 0 \} $ denotes the projection onto the positive part of $B-uA$.
	\end{theo}
	
	\begin{theo}[Operator change of variables{\cite[Theorem C.1]{preparation}}] \label{theo:change-of-variable}
		Let $A$ and $B$ be finite-dimensional positive semi-definite operators satisfying $r := \|A B^{-1} \|_{\infty} < \infty$.
		Then, for any Lebesgue-integrable function $h$ on $[0,r]$,
		\begin{align} \label{eq:change-of-variable}
			\int_0^\infty \proj{ A > \gamma B } h(\gamma) \, \d \gamma
			&= \int_0^\infty \frac{1}{B+t\I} A \frac{1}{\sqrt{B+t\I}} h\left(\frac{1}{\sqrt{B+t\I}} A \frac{1}{\sqrt{B+t\I}} \right)\frac{1}{\sqrt{B+t\I}} \, \d t
			\\
			&= \int_0^\infty \frac{1}{B+t\I} A \frac{1}{B+t\I} h\left( A \frac{1}{B+t\I} \right) \d t.
		\end{align}
	\end{theo}
	With these technical tools in hand, we can now start investigating their implications for quantum $f$-divergences in the following sections.  
	
	%%%%%%%%%%%%%%%%%%%%%%%%%%%%%%%%%%%%%%%%%%%%%%%%%%%%%%%%%%%
	
	\section{Layer cake representation of \Renyi and $f$-divergences} \label{sec:relation}
	
	We start by stating formal definitions of the quantum \Renyi and $f$-divergences that constitute the starting point of this work.  
	\begin{defn}[Quantum $f$-divergence]
		For convex function $f$ defined on $[0,\infty)$ satisfying $f(1) = 0$ and states $\rho$ and $\sigma$ such that $\rho\ll\sigma$,
		\begin{align}
			D^{\ast}_f(\rho\Vert\sigma) :=
			f(0) + \int_0^\infty \Tr\left[ \sigma \proj{ \rho > \gamma \sigma}\right] \, \d f(\gamma).
		\end{align}
		If $f$ is differentiable,
		\begin{align}\label{Eq:Def-1-dif}
			D^{\ast}_f(\rho\Vert\sigma) =
			f(0) + \int_0^\infty f'(\gamma) \Tr\left[ \sigma \proj{ \rho > \gamma \sigma}\right] \, \d \gamma.
		\end{align}
	\end{defn}
	We also define the corresponding \Renyi divergence. 
	\begin{defn}[Quantum \Renyi divergence]
		Let $\alpha \in (0,1)\cup(1,\infty)$.
		For states $\rho$ and $\sigma$ satisfying $\rho\ll\sigma$ unless $\alpha \in (0,1)$,
		\begin{align}
			D^{\ast}_{\alpha} (\rho\Vert\sigma)
			&:= \frac{1}{\alpha-1} \log Q^{\ast}_{\alpha}(\rho\Vert\sigma),
			\\
			Q^{\ast}_{\alpha}(\rho\Vert\sigma)
			&:= \alpha \int_0^\infty \gamma^{\alpha-1} \Tr\left[ \sigma  \proj{ \rho > \gamma \sigma }\right] \d \gamma, \quad \alpha > 0. \label{Eq:Q-layerCake}
		\end{align}
	\end{defn}
	Note that \Renyi and $f$-divergences are related via the Hellinger divergence,
	\begin{align}
		H_{\alpha}(\rho\Vert\sigma)
		:= D_{f_\alpha}(\rho\Vert\sigma),
	\end{align} 
	where $f_{\alpha}(x) = \frac{x^{\alpha}-1}{\alpha-1}$, $\alpha>0$, 
	by 
	\begin{align}\label{Eq:quasi-and-Hellinger}
		Q_{\alpha}(\rho\Vert\sigma)
		&= (\alpha-1) H_{\alpha} (\rho\Vert\sigma) + 1.
	\end{align}
	Also, by Lemma~\ref{Lem:HS-derivative}, we can write Equation~\eqref{Eq:Def-1-dif} as, 
	\begin{align}\label{Eq:Def-1-dif-HS}
		D^{\ast}_f(\rho\Vert\sigma) =
		f(0) - \int_0^\infty f'(\gamma) \frac{\d E_\gamma(\rho\|\sigma)}{\d\gamma} \, \d \gamma.
	\end{align}
	This is suggestive towards a connection to $f$-divergences based on Hockey Stick divergences. 
	In his PhD thesis~\cite{polyanskiy2010channel}, Polyanskiy gave the following integral representation for the relative entropy,
	\begin{align}
		D(P\Vert Q)
		&=1 - \int_0^\infty \log \gamma \frac{\d E_{\gamma}(P\Vert Q)}{\d \gamma} \d \gamma.
	\end{align}
	Our definition can be seen as an extension to general $f$-divergences that is well defined also in the quantum setting\footnote{The relative entropy special case follows from Equation~\eqref{Eq:Def-1-dif-HS} with $f(x)=x\log x$ and applying Equation~\eqref{Eq:shiftByC}, which is shown a bit later, with $c=1$.}. 
	Our first main result is that the above definitions are equivalent to those in Equations~\eqref{Eq:int-quasi-petz} and~\eqref{Eq:int-f-div} previously proposed in~\cite{hirche2023quantum}. 
	
	Let $\cF$ be the set of functions $f:(0,\infty)\rightarrow\bbR$ with $f(1)=0$ that are convex and twice differentiable. 
	\begin{theo} \label{theorem:f_divergence_layer-cake}
		For any $f\in\cF$ or $\alpha\in(0,1)\cup(1,\infty)$, we have
		\begin{align}
			D^{\ast}_f(\rho\|\sigma) &=   D_f(\rho\|\sigma)  \\
			D^{\ast}_\alpha(\rho\|\sigma) &=   D_\alpha(\rho\|\sigma)
		\end{align}
	\end{theo}
	\begin{proof}
		By the relation in~\eqref{Eq:quasi-and-Hellinger} it is sufficient to prove the first equality. Set, $\beta_2=e^{-D_{\max}(\sigma\|\rho)}$ and $\beta_1^{-1}=e^{D_{\max}(\rho\|\sigma)}$. In~\cite[Equation~(4.17)]{beigi2025some} the following expression was shown, 
		\begin{align}
			D_f(\rho\|\sigma) &=   f({\beta_2}) - f(1) - \sum_{i=0}^N \left(  \int_{a_i}^{a_{i+1}} f'(\gamma)  E'_\gamma(\rho\|\sigma) \d\gamma \right),  
		\end{align}
		where the $a_i$ are a finite number of points for which $E'_\gamma(\rho\|\sigma)$ is not continuous. Because at any such point the left and right derivatives exist and are finite, we can extend this to 
		\begin{align}
			D_f(\rho\|\sigma) &=   f({\beta_2}) - f(1) -   \int_{\beta_2}^{\beta_1^{-1}} f'(\gamma)  E'_\gamma(\rho\|\sigma) \d\gamma.  
		\end{align}
		Note that, 
		\begin{align}
			\int_{0}^{\beta_2} f'(\gamma)  E'_\gamma(\rho\|\sigma) \d\gamma = -\int_{0}^{\beta_2} f'(\gamma)   \d\gamma = f(0) - f(\beta_2). 
		\end{align}
		Hence, 
		\begin{align}
			D_f(\rho\|\sigma) &=   f(0) - f(1) -   \int_{0}^{\beta_1^{-1}} f'(\gamma)  E'_\gamma(\rho\|\sigma) \d\gamma.   
		\end{align}
		Finally, the result follows because the Hockey-Stick divergence is constant and equal to $0$ for all $\gamma\geq\beta_1^{-1}$ and by using Lemma~\ref{Lem:HS-derivative}. 
	\end{proof}
	
	%%%%%%%%%%%%%%%%%%%%%%%%%%%%%%%%%%%%%%%

	Note that, 
	\begin{align}
		-\int_0^\infty \frac{\d E_{\gamma}(\rho\Vert\sigma)}{\d \gamma} \d \gamma = E_0(\rho\|\sigma) - E_\infty(\rho\|\sigma) = 1. 
	\end{align}
	Hence, we can modify the representation with arbitrary constants, 
	\begin{align}\label{Eq:shiftByC}
		D_f(\rho\|\sigma) &=   f(0) + c -   \int_{0}^{\infty} (f'(\gamma)-c)  E'_\gamma(\rho\|\sigma) \d\gamma 
	\end{align}
	For the particular case of the Hellinger divergence, we can give further integral representations. Interestingly, here we have to consider different values of $\alpha$ separately. 
	\begin{lemm} \label{lemm:alpha-representation}
		We have, for $0<\alpha<1$, 
		\begin{align}
			(\alpha-1)H_\alpha(\rho\|\sigma) &= (1-\alpha) \int_0^\infty \gamma^{\alpha-2} \tr[\rho\{\rho\leq\gamma\sigma\}] \d\gamma - 1 \label{eq:layer-cake_alpha<1-0}\\
			&= (1-\alpha) \int_0^\infty \gamma^{-\alpha} \tr[\rho\{\gamma\rho\leq\sigma\}] \d\gamma - 1
			\label{eq:layer-cake_alpha<1}
		\end{align}
		and, for $\alpha>1$, 
		\begin{align}
			(\alpha-1)H_\alpha(\rho\|\sigma) &= (\alpha-1) \int_0^\infty \gamma^{\alpha-2} \tr[\rho\{\rho>\gamma\sigma\}] \d\gamma  -1  \label{eq:layer-cake_alpha>1-0}\\
			&= (\alpha-1) \int_0^\infty \gamma^{-\alpha} \tr[\rho\{\gamma\rho>\sigma\}] \d\gamma  -1 .
		\end{align}
	\end{lemm}

	\begin{proof}[Proof of Lemma~\ref{lemm:alpha-representation}]
		Here, we use the integral representation found in~\cite{hirche2023quantum}, which can be written as, 
		\begin{align}
			H_\alpha(\rho\|\sigma) = \alpha \int_0^\infty \gamma^{\alpha-2} E_\gamma(\rho\|\sigma) \d\gamma - \alpha\int_0^1 \gamma^{\alpha-2}(1-\gamma) \d\gamma.
		\end{align}
		Recall that, 
		\begin{align}
			E_\gamma(\rho\|\sigma) = \tr(\rho-\gamma\sigma)_+ = \tr[\rho\{\rho>\gamma\sigma\}] - \gamma\tr[\sigma\{\rho>\gamma\sigma\}]. 
		\end{align}
		With, this we start by proving the case of $0<\alpha<1$, 
		\begin{align}
			H_\alpha(\rho\|\sigma) &= \alpha \int_0^\infty \gamma^{\alpha-2} \left(\tr[\rho\{\rho>\gamma\sigma\}] - \gamma\tr[\sigma\{\rho>\gamma\sigma\}] \right)\d\gamma - \alpha\int_0^1 \gamma^{\alpha-2}(1-\gamma) \d\gamma \\
			\label{Eq:Ha-step2}
			&= \alpha \int_0^\infty \gamma^{\alpha-2} \tr[\rho\{\rho>\gamma\sigma\}] \d\gamma  -1 -(\alpha-1) H_\alpha(\rho\|\sigma) - \alpha\int_0^1 \gamma^{\alpha-2}(1-\gamma) \d\gamma \\
			&= \alpha \int_0^\infty \gamma^{\alpha-2} \left(1-\tr[\rho\{\rho\leq\gamma\sigma\}]\right) \d\gamma  -1 -(\alpha-1) H_\alpha(\rho\|\sigma) - \alpha\int_0^1 \gamma^{\alpha-2}(1-\gamma) \d\gamma \\
			&= -\alpha \int_0^\infty \gamma^{\alpha-2} \tr[\rho\{\rho\leq\gamma\sigma\}] \d\gamma  -1 -(\alpha-1) H_\alpha(\rho\|\sigma) + \alpha\int_1^\infty \gamma^{\alpha-2} \d\gamma + \alpha\int_0^1 \gamma^{\alpha-1} \d\gamma \\
			&= -\alpha \int_0^\infty \gamma^{\alpha-2} \tr[\rho\{\rho\leq\gamma\sigma\}] \d\gamma   -(\alpha-1) H_\alpha(\rho\|\sigma) -\frac{\alpha}{\alpha-1}\\
			&= -\alpha \int_0^\infty \gamma^{-\alpha} \tr[\rho\{\gamma\rho\leq\sigma\}] \d\gamma   -(\alpha-1) H_\alpha(\rho\|\sigma) -\frac{\alpha}{\alpha-1}, 
		\end{align}
		where for the second equality we used Equation~\eqref{Eq:Q-layerCake} and the final equality is substituting $\gamma'=\gamma^{-1}$ in the integral. The claim then follows be rearranging the terms. 
		
		We continue by proving the claim for $\alpha>1$, starting from Equation~\eqref{Eq:Ha-step2},
		\begin{align}
			H_\alpha(\rho\|\sigma) &= \alpha \int_0^\infty \gamma^{\alpha-2} \tr[\rho\{\rho>\gamma\sigma\}] \d\gamma  -1 -(\alpha-1) H_\alpha(\rho\|\sigma) - \alpha\int_0^1 \gamma^{\alpha-2}(1-\gamma) \d\gamma  \\
			&= \alpha \int_0^\infty \gamma^{\alpha-2} \tr[\rho\{\rho>\gamma\sigma\}] \d\gamma  -1 -(\alpha-1) H_\alpha(\rho\|\sigma) - \frac{1}{\alpha-1} \\
			&= \alpha \int_0^\infty \gamma^{-\alpha} \tr[\rho\{\gamma\rho>\sigma\}] \d\gamma  -1 -(\alpha-1) H_\alpha(\rho\|\sigma) - \frac{1}{\alpha-1}.
		\end{align}
		The claim follows again after some rearranging. 
	\end{proof}    
	This new representation gives us immediately the following corollary. 
	\begin{coro}
		For $0<\alpha<1$,
		\begin{align}
			D_{\alpha}(\rho \Vert \sigma)
			= \frac{\alpha}{1-\alpha} D_{1-\alpha}(\sigma \Vert \rho).
		\end{align}
	\end{coro}
	\begin{proof}
		It is equivalent to proving 
		\begin{align}
			(\alpha-1) H_{\alpha}(\rho \Vert \sigma)
			= - \alpha H_{1-\alpha}(\sigma \Vert \rho).
		\end{align}
		We start with \eqref{eq:layer-cake_alpha<1-0}:
		\begin{DispWithArrows}[displaystyle]
			- \alpha H_{1-\alpha}(\sigma \Vert \rho) -1
			&= \alpha \int_0^\infty \gamma^{-\alpha-1} \tr[\sigma\{\sigma < \gamma\rho\}] \d\gamma 
			\Arrow{$\gamma \leftarrow \frac{1}{\gamma}$} -1
			\\
			&= \alpha \int_\infty^0 \gamma^{\alpha+1} \tr[\sigma\{\gamma \sigma < \rho\}] \left( - \frac{1}{\gamma^2} \right)\d\gamma -1
			\\
			&=\alpha \int_0^\infty \gamma^{\alpha-1} \tr[\sigma\{\gamma \sigma < \rho\}]\d\gamma -1
			\Arrow{\eqref{Eq:quasi-and-Hellinger}}
			\\
			&= (\alpha-1) H_{\alpha}(\rho \Vert \sigma).
		\end{DispWithArrows}
	\end{proof}
	We next turn our attention to implications of the new divergences.

	%%%%%%%%%%%%%%%%%%%%%%%%%%%%%%%%%%%%%%%%%%%%%%%%%%%%%%%%%%%
	
	\section{Trace representations for \(f\)-divergences}
	
	Earlier work in~\cite{beigi2025some} coined the term \textit{trace representations} for the expressions that resemble the most usual representations of divergences that are given by the trace of some operator function.  Of course the best known example is given by the typical expression for the Umegaki relative entropy~\cite{Ume62}, 
	\begin{align} \label{eq:defn:Umegaki}
		D(\rho\Vert\sigma) := \Tr\left[ \rho \left( \log \rho - \log \sigma \right) \right].
	\end{align}
	The relative entropy is a special case of the studied $f$-divergence for $f(x) = x \log x$.
	This observation from~\cite{hirche2023quantum} was made using an integral representation by Frenkel~\cite{frenkel2022integral}. 
	
	In the following, we first give an alternative proof of this statement based on the layer cake representation. 
	After that we proof a conjecture regarding a trace representation for the \Renyi divergence for $\alpha>1$ from~\cite{beigi2025some}. Finally, we briefly discuss the case of general $f$-divergences. 
	
	\subsection{The relative entropy (an alternative proof for Frenkel's formula)} \label{sec:alpha=1}

	In this section we provide an alternative proof of the integral formula for the relative entropy by Frenkel. We start by giving an alternative representation. 
	\begin{prop} \label{prop:algernative_proof_Frenkel}
		For any positive definite operators $A$ and $B$, we have
		\begin{align}
			\log A - \log B 
			&= \int_1^{\infty} \left\{ A > \alpha B \right\} \frac{1}{\alpha} \d \alpha - \int_1^{\infty} \left\{ B > \alpha A \right\} \frac{1}{\alpha} \d \alpha,
			\label{eq:log-difference1}
			\\
			D(A\Vert B)
			&= \int_1^\infty \frac{1}{\gamma} \Tr\left[ A \left( \left\{ A > \gamma B \right\} - \left\{ B > \gamma A \right\} \right) \right] \d \gamma.
			\label{eq:D-projection}
		\end{align}
	\end{prop}
	\begin{proof}
		Using the fundamental theorem of calculus and Theorem~\ref{theo:Dlog_formula}, we write
		\begin{align}
			\begin{split} \label{eq:FTC}
				&\log A - \log B 
				= \int_0^1 \mathrm{D} \log \left[ B + \beta (A-B) \right](A-B) \, \d \beta
				\\
				&=
				\underbrace{
					\int_0^1 \int_0^\infty \left\{ A-B > \left( B + \beta(A-B)\right)u \right\} \d u \,\d \beta
				}_{\text{(I) }}
				-
				\underbrace{
					\int_{0}^1 \int_{-\infty}^0 \left\{ A-B < \left( B + \beta(A-B)\right)u \right\} \d u \,\d \beta
				}_{\text{(II)}}.
			\end{split}
		\end{align}
		We continue by evaluating the individual terms, starting with the first.
		\begin{DispWithArrows}[displaystyle, fleqn, mathindent = 75pt]
			\text{(I)}
			&= \int_0^1 \int_0^\infty \left\{ (1-u\beta) (A-B) > u B \right\} \d u \,\d \beta
			\Arrow{for $0\leq \beta \leq 1$ and $u\beta \geq 1$,\\
				\quad $(1-u\beta) (A-B) \leq u B$ } \notag
			\\
			&= \int_0^1 \int_0^{\frac{1}{\beta}} \left\{ (1-u\beta) (A-B) > u B\right\} \d u \,\d \beta \notag
			\\
			&= \int_0^1 \int_0^{\frac{1}{\beta}} \left\{ A - B > \frac{u}{1-u\beta} B \right\} \d u\, \d\beta
			\Arrow{$\gamma = \frac{u}{1-u\beta}$, $\d u = \frac{1}{(1+\beta\gamma)^2}\d \gamma$} \notag
			\\ 
			&= \int_0^1 \int_0^{\infty} \left\{ A - B > \gamma B \right\} \frac{1}{(1+\beta \gamma)^2} \d \gamma \, \d \beta 
			\Arrow{$\because$ Fubini's theorem and \\ dominated convergence theorem}
			\notag
			\\ 
			&= \int_0^{\infty} \left\{ A - B > \gamma B \right\} \int_0^1 \frac{1}{(1+\beta \gamma)^2} \d \beta \, \d \gamma \notag
			\\
			&=\int_0^{\infty} \left\{ A - B > \gamma B \right\} \frac{1}{\gamma + 1} \d \gamma 
			\Arrow{$\alpha  = \gamma + 1$} \notag
			\\
			&= \int_1^{\infty} \left\{ A > \alpha B \right\} \frac{1}{\alpha} \, \d \alpha. 
			\label{eq:(I)}
		\end{DispWithArrows}
		For the second term we get the following.
		\begin{DispWithArrows}[displaystyle, fleqn, mathindent = 65pt] 
			\text{(II)}
			&= \int_0^1 \int_{-\infty}^0 \left\{ (1-u\beta) (A-B) < u B \right\} \d u \,\d \beta 
			\notag
			\\
			&= \int_0^1 \int_{-\infty}^0 \left\{ (1-u\beta+u) (A-B) < u A \right\} \d u \,\d \beta
			\Arrow{for $0\leq \beta \leq 1$ and $u(1-\beta)\leq -1$
				\\ \quad $(1-u\beta+ u) (A-B) \geq u A$ } \notag
			\\
			&= \int_0^1 \int_{\frac{1}{\beta-1}}^0 \left\{ A - B < \frac{u}{1-u\beta+u} A \right\} \d u \, \d \beta 
			\Arrow{$\gamma = \frac{u}{1+(1-\beta)u}$, $\d u = \frac{\d \gamma}{\left(1+(\beta-1)\gamma\right)^2}$}  \notag
			\\
			&= \int_0^1 \int_{-\infty}^0 \left\{ A - B < \gamma A \right\} \frac{1}{\left( 1 + (\beta-1)\gamma\right)^2} \d \gamma \, \d \beta 
			\Arrow{$\because$ Fubini's theorem and\\ \, dominated convergence theorem}
			\notag
			\\
			&=
			\int_{-\infty}^0 \left\{ A - B < \gamma A \right\} \int_0^1 \frac{1}{\left( 1 + (\beta-1)\gamma\right)^2} \d \beta \, \d \gamma \notag
			\\
			&= \int_{-\infty}^0 \left\{ A - B < \gamma A \right\} \frac{1}{1-\gamma} \d \gamma
			\Arrow{$\alpha = 1-\gamma$} \notag
			\\
			&= \int_1^{\infty} \left\{ B > \alpha A \right\} \frac{1}{\alpha} \d \alpha.
			\label{eq:(II)}
		\end{DispWithArrows}
		Combining \eqref{eq:FTC}, \eqref{eq:(I)}, and \eqref{eq:(II)}, we have
		\begin{align}
			\log A - \log B 
			&= \int_1^{\infty} \left\{ A > \alpha B \right\} \frac{1}{\alpha} \d \alpha - \int_1^{\infty} \left\{ B > \alpha A \right\} \frac{1}{\alpha} \d \alpha,
			\notag
			\\
			D(A\Vert B)
			&= \int_1^\infty \frac{1}{\gamma} \Tr\left[ A \left( \left\{ A > \gamma B \right\} - \left\{ B > \gamma A \right\} \right) \right] \d \gamma.
			\notag\qedhere
		\end{align}
	\end{proof}
	
	The following is a restatement of the result by Frenkel~\cite{frenkel2022integral}. However, we follow an alternative path by building on our previous result. 
	\begin{prop} \label{prop:algernative_proof_Frenkel2}
		For any positive definite operators $A$ and $B$, we have
		\begin{align}
			D(A\Vert B )
			&= \int_1^\infty \Big(\frac{1}{\gamma} E_\gamma(A\|B) +  \frac{1}{\gamma^2} E_\gamma(B\|A)\Big)\, \d \gamma
			+ \Tr[A - B].
		\end{align}
	\end{prop}
	\begin{proof}
		Continuing from the previous proposition,
		\begin{DispWithArrows}[displaystyle] 
			&\int_1^\infty \frac{1}{\gamma} \Tr\left[ A \left( \left\{ A > \gamma B \right\} - \left\{ B > \gamma A \right\} \right) \right] \d \gamma
			\notag
			\\
			&= \int_1^{\infty} \Tr\left[ \frac{1}{\gamma} \left( A - \gamma B \right) \left\{ A > \gamma B \right\} + \frac{1}{\gamma^2} \left( B - \gamma A \right) \left\{ B > \gamma A \right\} \right] \d \gamma \notag
			\\
			&\quad + \int_1^{\infty} \Tr\left[ B \left\{ A > \gamma B \right\} \right] \d \gamma 
			- \int_1^{\infty} \frac{1}{\gamma^2} \Tr\left[  B \left\{ B > \gamma A \right\} \right] \d \gamma \notag
			\\
			&= \int_1^\infty \Big(\frac{1}{\gamma} E_\gamma(A\|B) +  \frac{1}{\gamma^2} E_\gamma(B\|A)\Big)\, \d \gamma \notag
			\\
			&\quad + \underbrace{\int_1^{\infty} \Tr\left[ B \left\{ A > \gamma B \right\} \right] \d \gamma}_{\text{(i)}} 
			- \underbrace{\int_1^{\infty} \frac{1}{\gamma^2} \Tr\left[  B \left\{ B > \gamma A \right\} \right] \d \gamma}_{\text{(ii)}}, \label{eq:temp}
		\end{DispWithArrows}
		and
		\begin{DispWithArrows}[displaystyle] 
			\text{(i)}
			&= \int_1^{\infty} \Tr\left[ B \left\{ A > \gamma B \right\} \right] \d \gamma
			\notag
			\\
			&=  -\int_0^1 \Tr\left[ B \left\{ A > \gamma B \right\}\right] \d \gamma 
			+ \int_0^{\infty} \Tr\left[ B \left\{ A > \gamma B\right\}\right] \d \gamma \notag
			\\
			&= -\int_0^1 \Tr\left[ B \left\{ A > \gamma B \right\} \right] \d \gamma 
			+  \Tr[A],
			\label{eq:(i)}
		\end{DispWithArrows}
		
		\begin{DispWithArrows}[displaystyle] 
			\text{(ii)}
			&= \int_1^{\infty} \frac{1}{\gamma^2} \Tr\left[  B \left\{ B > \gamma A \right\} \right] \d \gamma
			\Arrow{$\alpha = \frac{1}{\gamma}$}
			\notag
			\\
			&= \int_0^1 \Tr\left[ B \left\{ \alpha B > A \right\} \right] \d \alpha
			\notag
			\\
			&= \int_0^1 \Tr\left[ B \left( \I - \left\{ \alpha B \leq A \right\} \right) \right] \d \alpha
			\Arrow{ $\because \{\alpha\mid\{\alpha B = A \}\neq 0\}$\\ \phantom{\,$\because$\,}\text{has measure zero}}
			% \Arrow{ $\because \{\alpha B = A \}$ \text{has measure zero}}
			\notag
			\\
			&= \int_0^1 \Tr\left[ B \left( \I - \left\{ \alpha B < A \right\} \right) \right] \d \alpha
			\notag
			\\
			&= \Tr[B] - \int_0^1 \Tr\left[ B \left\{ A > \alpha B \right\} \right] \d \alpha.
			\label{eq:(ii)}
		\end{DispWithArrows}
		
		By \eqref{eq:D-projection}, \eqref{eq:temp}, \eqref{eq:(i)}, and \eqref{eq:(ii)}, we obtain
		\begin{align}
			D(A\Vert B )
			&= \int_1^\infty \frac{1}{\gamma} \Tr\left[ A \left( \left\{ A > \gamma B \right\} - \left\{ B > \gamma A \right\} \right) \right] \d \gamma
			\\
			&= \int_1^\infty \Big(\frac{1}{\gamma} E_\gamma(A\|B) +  \frac{1}{\gamma^2} E_\gamma(B\|A)\Big)\, \d \gamma
			+ \Tr[A - B].
		\end{align}
	\end{proof}
	
	\subsection{The \Renyi divergence for \(\alpha > 1\)} \label{sec:alpha>1}
	
	In this section, we prove the trace expression conjectured in~\cite{beigi2025some} for the \Renyi divergence with $\alpha>1$. For now, we focus on normalized quantum states. That is because we can give a particularly simple proof using the previously derived results. Afterwards, we comment on the case of more general positive operators. 
	
	\begin{theo}\label{Theo:trace-renyi-QS}
		For any quantum states $\rho$, $\sigma$ and any $\alpha>1$,
		\begin{align}
			H_{\alpha}(\rho\Vert \sigma)
			&= \frac{-1}{\alpha-1} + \int_0^\infty \Tr\left[ \left( \rho \frac{1}{\sigma+t\I} \right)^{\alpha} \right] \d t .
		\end{align}
	\end{theo}
	\begin{proof}
		Recall the layer cake representation in Equation~\eqref{Eq:Def-1-dif}, which speciallised to the Hellinger divergence gives,
		\begin{align}
			H_\alpha(\rho\Vert\sigma) =
			\frac{-1}{\alpha-1} + \frac{\alpha}{\alpha-1}\int_0^\infty \gamma^{\alpha-1} \Tr\left[ \sigma \proj{ \rho > \gamma \sigma}\right] \, \d \gamma.
		\end{align}
		Using Theorem~\ref{theo:Dlog_formula} then gives, 
		\begin{align}
			H_\alpha(\rho\Vert\sigma) &=
			\frac{-1}{\alpha-1} + \frac{\alpha}{\alpha-1}\int_0^\infty \Tr\left[ \sigma \frac{1}{\sigma+t\I} \rho \frac{1}{\sigma+t\I} \left( \rho \frac{1}{\sigma+t\I} \right)^{\alpha-1} \right] \, \d t \\
			&=
			\frac{-1}{\alpha-1} + \frac{\alpha}{\alpha-1}\int_0^\infty \Tr\left[ \sigma \frac{1}{\sigma+t\I}  \left( \rho \frac{1}{\sigma+t\I} \right)^{\alpha} \right] \, \d t. \label{Eq:Almost-tr-exp}
		\end{align}
		Finally, we employ Lemma~\ref{lemm:order_above_1} below, from which the result immediately follows.
	\end{proof}
	% Note that the above proof only works for $\alpha>1$, because of the corresponding requirement in Lemma~\ref{lemm:order_above_1}. The intermediate result in Equation~\eqref{Eq:Almost-tr-exp} also holds for $0<\alpha<1$ and can be compared to the trace representation given in~\cite{beigi2025some}. 
	% \HC{I see. Didn't notice that \ref{Eq:Almost-tr-exp} also holds for $0<\alpha<1$. Here, we mainly focused on $\alpha>1$ only. Maybe there's a way or the above way can be adapted to deal with $0<\alpha<1$ as well.}
	
	% \smallskip
	
	Once we move to positive operators, we can not as much rely on the previously derived results. We here state the formal result and provide the proof in the Appendix. 
	\begin{theo}\label{Theo:trace-renyi-PO}
		For positive semidefinite operators $A, B$ and any $\alpha>1$,
		\begin{align}
			H_{\alpha}(A\Vert B)
			&= \frac{-1}{\alpha-1}\Tr[A] + \int_0^\infty \Tr\left[ \left( A \frac{1}{B+t\I} \right)^{\alpha} \right] \d t .
		\end{align}
	\end{theo}
	The proof relies again strongly on Theorem~\ref{theo:Dlog_formula} and Lemma~\ref{lemm:order_above_1}, however it takes a more direct route from the integral representation rather than using the layer cake representation. For details we refer to the Appendix~\ref{sec:appendix_proof}. 
	
	Theorem~\ref{Theo:trace-renyi-PO} with the observation made in \cite{beigi2025some} allow us to relate the integral quasi \Renyi divergence to the quasi sandwiched \Renyi divergence.
	
	\begin{prop}[Relation to quasi sandwiched \Renyi divergence] \label{prop:relation_sandwiched}
		For any quantum states $\rho,\sigma $ satisfying $\rho\ll\sigma$, and any $\alpha > 1$,
		\begin{align}
			Q_{\alpha}(A \Vert B ) \leq \Tr\left[ \left( B^{\frac{1-\alpha}{2\alpha}}A  B^{\frac{1-\alpha}{2\alpha}} \right)^{\alpha} \right].
		\end{align}
	\end{prop}
	\begin{proof}
		The claim follows from Theorem~\ref{Theo:trace-renyi-PO} and 
		\cite[Proposition~3.10]{beigi2025some}  for $\alpha >1$.
	\end{proof}
	
	\smallskip
	Finally, we need to prove the lemma that makes all of the results in this section possible. 
	\begin{lemm} \label{lemm:order_above_1}
		Let $A$ and $B$ be positive semi-definite operators satisfying $ \|A B^{-1} \|_{\infty} < \infty$.
		Then, 
		\begin{align}
			\int_0^\infty \Tr\left[ B \frac{1}{B+t\I} \left( A \frac{1}{B+t\I} \right)^\alpha \right] \d t
			&= \frac{\alpha-1}{\alpha} \int_0^\infty \Tr\left[ \left( A \frac{1}{B + t\I} \right)^\alpha\right] \d t, \quad \forall\, \alpha>1.
		\end{align}
	\end{lemm}
	\begin{proof}
		The claim is equivalent to proving 
		\begin{align}
			\int_0^\infty \Tr\left[ \frac{t}{B+t\I} \left( A \frac{1}{B+t\I} \right)^\alpha \right] \d t
			&= \frac{1}{\alpha} \int_0^\infty \left( A \frac{1}{B + t\I} \right)^\alpha \d t.
		\end{align}
		
		Let $n = \lfloor \alpha \rfloor$ and $\tilde{\alpha} = \alpha - n$.
		We first calculate,
		\begin{align}
			\frac{\d}{\d t}\left[\left( A \frac{1}{B+t\I} \right)^{\alpha} \right]
			&= \frac{\d}{\d t}\bigg[ \underbrace{ A \frac{1}{B+t\I} \cdots A \frac{1}{B+t\I} }_{\text{$n$ times}} \left( A \frac{1}{B+t\I} \right)^{\tilde{\alpha}} \bigg]
			\notag
			\\
			\begin{split}
				\label{eq:order_above_1:derivative}
				&= A \frac{-1}{(B+t\I)^2} A \frac{1}{B+t\I} A \frac{1}{B+t\I} \cdots \left( A \frac{1}{B+t\I} \right)^{\tilde{\alpha}}
				\\
				&+A \frac{1}{B+t\I} \frac{-1}{(B+t\I)^2} A \frac{1}{B+t\I}  \cdots \left( A \frac{1}{B+t\I} \right)^{\tilde{\alpha}}
				\\
				&\quad\vdots
				\\
				&+ A \frac{1}{B+t\I} A \frac{1}{B+t\I} \cdots \frac{-1}{(B+t\I)^2} \cdot \left( A \frac{1}{B+t\I} \right)^{\tilde{\alpha}}
				\\
				&+ \left(A \frac{1}{B+t\I} \right)^n  \cdot \frac{\d}{\d t} \left( A \frac{1}{B+t\I} \right)^{\tilde{\alpha}}.
			\end{split}
		\end{align}
		Note that the first $n$ terms of \eqref{eq:order_above_1:derivative}
		have the same trace $\Tr\left[ \frac{-1}{B+t\I} \left( A \frac{1}{B+t\I} \right)^{\alpha} \right]$.
		We then calculate the trace of the last term in \eqref{eq:order_above_1:derivative}.
		By using the integral representation for the power function $0<x \mapsto x^{\tilde{\alpha}}$, i.e., $X^{\tilde{\alpha}} =  \int_0^\infty  \lambda^{\tilde{\alpha}-1} X \frac{1}{X+ \lambda\I} \d \lambda \cdot \frac{ \sin \alpha \pi }{\pi}$, we have
		\begin{align*}
			\frac{\d}{\d t} \left[ \left( A \frac{1}{B+t\I} \right)^{\tilde{\alpha}} \right]
			&= \frac{\d}{\d t} \int_0^\infty \lambda^{\tilde{\alpha}-1} A \frac{1}{B+t\I}  \frac{1}{A \frac{1}{B+t\I} + \lambda\I} \d \lambda \cdot \frac{ \sin \alpha \pi }{\pi}
			\\
			&= \frac{\d}{\d t} \int_0^\infty \lambda^{\tilde{\alpha}-1} A  \frac{1}{A  + \lambda (B+t\I)  } \d \lambda \cdot \frac{ \sin \alpha \pi }{\pi}
			\\
			&= \int_0^\infty \lambda^{\tilde{\alpha}-1} A  \frac{-\lambda}{\left(A  + \lambda (B+t\I) \right)^2 } \d \lambda \cdot \frac{ \sin \alpha \pi }{\pi}
			\\
			&= \int_0^\infty - \lambda^{\tilde{\alpha}} A \frac{1}{B+t\I} \frac{1}{A \frac{1}{B+t\I} + \lambda \I} \frac{1}{B+t\I} \frac{1}{A \frac{1}{B+t\I} + \lambda \I} \d \lambda \cdot \frac{ \sin \alpha \pi }{\pi}.
		\end{align*}
		Hence, the trace of the last term in \eqref{eq:order_above_1:derivative} is
		\begin{align}
			&\Tr\left[ \left( A \frac{1}{B+t\I} \right)^n \frac{\d}{\d t} \left( A \frac{1}{B+t\I} \right)^{\tilde{\alpha}} \right]
			\notag
			\\
			&= \Tr\left[ \left( A \frac{1}{B+t\I} \right)^n \int_0^\infty - \lambda^{\tilde{\alpha}} A \frac{1}{B+t\I} \frac{1}{A \frac{1}{B+t\I} + \lambda \I} \frac{1}{B+t\I} \frac{1}{A \frac{1}{B+t\I} + \lambda \I} \d \lambda \right]\cdot \frac{ \sin \alpha \pi }{\pi}
			\notag
			\\
			&=\Tr\left[ \left( A \frac{1}{B+t\I} \right)^{n+1} \int_0^\infty - \lambda^{\tilde{\alpha}}  \frac{1}{A \frac{1}{B+t\I} + \lambda \I} \frac{1}{B+t\I} \frac{1}{A \frac{1}{B+t\I} + \lambda \I} \d \lambda \right]\cdot \frac{ \sin \alpha \pi }{\pi}
			\notag
			\\
			&=\Tr\left[ \left( A \frac{1}{B+t\I} \right)^{n+1} \int_0^\infty - \lambda^{\tilde{\alpha}}  \frac{1}{\left(A \frac{1}{B+t\I} + \lambda \I\right)^2}  \d \lambda \cdot \frac{ \sin \alpha \pi }{\pi} \frac{1}{B+t\I}\right]
			\notag
			\\
			&=\Tr\left[ \left( A \frac{1}{B+t\I} \right)^{n+1} \left\{ - \tilde{\alpha} \left( A \frac{1}{B+t\I} \right)^{\tilde{\alpha}-1} \right\} \frac{1}{B+t\I}\right]
			\notag
			\\
			&= - \tilde{\alpha} \Tr\left[ \frac{1}{B+t\I} \left( A \frac{1}{B+t\I} \right)^{\alpha} \right].
			\label{eq:order_above_1:derivative1}
		\end{align}
		
		Using \eqref{eq:order_above_1:derivative}, we now calculate
		\begin{align}
			\frac{\d}{\d t}\left[ t \left( A \frac{1}{B+t\I} \right)^{\alpha} \right]
			\notag
			&= \left( A \frac{1}{B+t\I} \right)^{\alpha}
			+ t \frac{\d}{\d t}\left[\left( A \frac{1}{B+t\I} \right)^{\alpha} \right]
			\notag
			\\
			\begin{split} \label{eq:order_above_1:derivative2}
				&= \left( A \frac{1}{B+t\I} \right)^{\alpha} + t A \frac{-1}{(B+t\I)^2} A \frac{1}{B+t\I} A \frac{1}{B+t\I} \cdots \left( A \frac{1}{B+t\I} \right)^{\tilde{\alpha}}
				\\
				&\qquad\qquad\qquad\quad\;\, + t A \frac{1}{B+t\I} \frac{-1}{(B+t\I)^2} A \frac{1}{B+t\I}  \cdots \left( A \frac{1}{B+t\I} \right)^{\tilde{\alpha}}
				\\
				&\qquad\qquad\qquad\quad\quad \vdots
				\\
				&\qquad\qquad\qquad\quad\;\, + tA \frac{1}{B+t\I} A \frac{1}{B+t\I} \cdots \frac{-1}{(B+t\I)^2} \cdot \left( A \frac{1}{B+t\I} \right)^{\tilde{\alpha}}
				\\
				&\qquad\qquad\qquad\quad\;\, + t\left(A \frac{1}{B+t\I} \right)^n  \cdot \frac{\d}{\d t} \left( A \frac{1}{B+t\I} \right)^{\tilde{\alpha}}.
			\end{split}
		\end{align}
		
		Using \eqref{eq:order_above_1:derivative1}, we take $\int_0^\infty \Tr[ \,\cdot\, ]\d t $ on \eqref{eq:order_above_1:derivative2} to obtain
		\begin{align}
			\int_0^\infty \Tr \left[ \frac{\d}{\d t} t \left( A \frac{1}{B+t\I} \right)^{\alpha} \right]
			&=
			\int_0^\infty \Tr\left[ \left( A \frac{1}{B+t\I} \right)^{\alpha} \right] \d t
			\notag
			\\
			&\quad - n \cdot \int_0^\infty \Tr\left[ \frac{t}{B+t\I} \left( A \frac{1}{B+t\I} \right)^{\alpha} \right] - \tilde{\alpha} \cdot \int_0^\infty \Tr\left[ \frac{t}{B+t\I} \left( A \frac{1}{B+t\I} \right)^{\alpha} \right]
			\notag
			\\
			&= \int_0^\infty \Tr\left[ \left( A \frac{1}{B+t\I} \right)^{\alpha} \right] \d t - \alpha \int_0^\infty \Tr\left[ \frac{t}{B+t\I} \left( A \frac{1}{B+t\I} \right)^{\alpha}\right].
			\label{eq:order_above_1:derivative3}
		\end{align}
		
		On the other hand, for $\alpha>1$ and bounded $A$, we have
		\begin{align}
			\int_0^\infty \Tr \left[ \frac{\d}{\d t} t \left( A \frac{1}{B+t\I} \right)^{\alpha} \right]
			&=  \Tr \left[ \int_0^\infty\frac{\d}{\d t} t \left( A \frac{1}{B+t\I} \right)^{\alpha} \right]
			\notag
			\\
			&= \left. \Tr \left[  t \left( A \frac{1}{B+t\I} \right)^{\alpha} \right] \right|_{0}^\infty
			\notag
			\\
			&= 0
			\label{eq:order_above_1:derivative4}
		\end{align}
		Combining \eqref{eq:order_above_1:derivative3} and \eqref{eq:order_above_1:derivative4}, we get
		\begin{align}
			\int_0^\infty \Tr\left[ \left( A \frac{1}{B+t\I}\right)^{\alpha} \right]
			= \alpha \int_0^\infty \Tr \left[ \frac{t}{B+t\I} \left( A \frac{1}{B+t\I}\right)^{\alpha} \right], 
			\notag
		\end{align}
		concluding the proof.
	\end{proof}
	
	\medskip 
	By applying a change-of-measure argument (Lemma~\ref{lemm:change-of-measure}) that will be introduced later in Section~\ref{sec:RS-integral},
	Lemma~\ref{lemm:order_above_1} has a simpler proof, which we provide below.
	\begin{proof}[An alternative proof of Lemma~\ref{lemm:order_above_1}]
		Recall \eqref{eq:change-of-measure_combined} of Lemma~\ref{lemm:change-of-measure}:
		\begin{align*}
			\int_0^\infty g'(\gamma) \Tr[A\{A>\gamma B\}] \, \d \gamma=\int_0^\infty \left(g(\gamma)+\gamma g'(\gamma)\right) \Tr[A\{A>\gamma B\}] \, \d \gamma.
		\end{align*}
		
		Let $g(\gamma)=\gamma^{\alpha-1}$, $g'(\gamma)=(\alpha-1)\gamma^{\alpha-2}$, then
		\[
		(\alpha-1)\int_0^\infty \gamma^{\alpha-2} \Tr[A\{A>\gamma B\}] \, \d \gamma
		=\alpha \int_0^\infty \gamma^{\alpha-1} \Tr[B\{A>\gamma B\}] \, \d \gamma
		\]
		
		By Theorem~\ref{theo:change-of-variable},
		\begin{align}
			\int_0^\infty \Tr\left[ B\frac{1}{B+t\I}\left(A\frac{1}{B+t\I}\right)^\alpha \right]\d t 
			&=\Tr \left[ B \int_0^\infty \gamma^{\alpha-1} \{A>\gamma B\} \right] \, \d \gamma\\
			&=\frac{\alpha-1}{\alpha} \Tr \left[ A \int_0^\infty \gamma^{\alpha-2} \{A>\gamma B\} \right] \, \d \gamma\\
			&=\frac{\alpha-1}{\alpha}\int_0^\infty \Tr\left[ \left(A\frac{1}{B+t\I}\right)^\alpha \right] \, \d t,
		\end{align}
		concluding the proof.
	\end{proof}

	\subsection{General \(f\)-divergences}
	The results in the previous section raise the question whether the same holds true for more general functions $f$. In that direction, we state the following observation. 
	\begin{prop} \label{prop:genearl_f}
		For any quantum states $\rho$, $\sigma$ and any convex $f$ such that $f'(\gamma)$ is Lebesgue-integrable on $[0,\e^{D_{\max}(\rho\Vert\sigma)}]$, 
		\begin{align}
			D_f(\rho\Vert\sigma) &=
			f(0) + \int_0^\infty \Tr\left[ \sigma \frac{1}{\sigma+t\I} \rho \frac{1}{\sigma+t\I} f'\left( \rho \frac{1}{\sigma+t\I} \right) \right] \, \d t .
		\end{align}
	\end{prop}
	\begin{proof}
		Again, we start from Equation~\eqref{Eq:Def-1-dif}
		\begin{align}
			D_f(\rho\Vert\sigma) =
			f(0) + \int_0^\infty f'(\gamma) \Tr\left[ \sigma \proj{ \rho > \gamma \sigma}\right] \, \d \gamma.
		\end{align}
		Using  Theorem~\ref{theo:change-of-variable} the claim follows directly. 
	\end{proof}
	Note that in the commuting setting, it can easily be checked that this corresponds to the usual definition of the classical $f$-divergence. 
	
	%\CH{Can't we prove something like, 
		%\begin{align}
		%    \frac{\d}{\d t}f\left( \rho \frac{1}{\sigma+t\I} \right) = -\frac{1}{\sigma+t\I} \rho \frac{1}{\sigma+t\I} f'\left( \rho \frac{1}{\sigma+t\I} \right)?
		%\end{align}
		%Classically, this is enough to bring the previous definition into the shape of the usual f-divergence.}
	%\HC{The formula in Proposition 3 is correct. What you suggested could also be true (not sure yet); if so, then there might be other proofs for Theorem 2.}
	%\CH{Given that previously we got such different representations for \Renyi $\alpha>1$ and $\alpha<1$ should suggest that this can only hold for certain $f$, right?}

	%%%%%%%%%%%%%%%%%%%%%%%%%%%%%%%%%%%%%%%%%%%%%%%%%%%%%%%%%%%%%%%%%%%
	
	\section{Application to error exponents} \label{sec:exponents}
	
	We are interested in two quantities, the \textbf{error exponent},
	\begin{align}
		B_e(r) = \sup\{ -\limsup_{n\rightarrow\infty} \frac1n \log\alpha_n(T_n) \mid \limsup_{n\rightarrow\infty} \frac1n \log\beta_n(T_n) \leq -r \}, 
	\end{align}
	and the \textbf{strong converse exponent},
	\begin{align}
		B^*_e(r) = \sup\{ -\liminf_{n\rightarrow\infty} \frac1n \log(1-\alpha_n(T_n)) \mid \limsup_{n\rightarrow\infty} \frac1n \log\beta_n(T_n) \leq -r \}.
	\end{align}
	We denote by $S_n(a):=\{\rho^{\otimes n}-\e^{na}\sigma^{\otimes n} > 0\}$ the quantum Neyman--Pearson tests. 
	\begin{prop}[Error probabilities via quantum Neyman--Pearson tests] \label{prop:threshold_test}
		Let $n\in\mathds{N}$ to be an arbitrary integer.
		For all $\alpha\in(0,1)\cup(1,\infty)$, 
		\begin{align} \label{eq:type-II_error_upper_bound}
			\log\Tr\left[ \sigma^{\otimes n}  S_n(a) \right] 
			&\leq -n a - n(\alpha-1)\left( a - \frac1n D_\alpha(\rho^{\otimes n}\|\sigma^{\otimes n})\right).
		\end{align}
		For all $\alpha>1$, 
		\begin{align} \label{eq:type-I_success_upper_bound}
			\log\tr[\rho^{\otimes n}S_n(a)] 
			&\leq  -n(\alpha-1)\left(a-\frac1n D_\alpha(\rho^{\otimes n}\|\sigma^{\otimes n})\right).
		\end{align}
		For all $0<\alpha<1$,
		\begin{align} \label{eq:type-I_error_upper_bound}
			\log\left[1-\tr[\rho^{\otimes n}S_n(a)] \right]
			\leq -n(\alpha-1)\left(a - \frac1n D_\alpha(\rho^{\otimes n}\|\sigma^{\otimes n})\right).
		\end{align}
	\end{prop}

	\begin{proof}
		Starting from Equation~\eqref{Eq:Q-layerCake}, we have for the \Renyi divergence,
		\begin{align}
			D_\alpha(\rho\|\sigma) = \frac{1}{\alpha-1}\log\left( \alpha \int_{0}^{\infty} \gamma^{\alpha-1}  \Tr\left[ \sigma  \left\{  \rho > \gamma \sigma \right\} \right] \d\gamma\right). 
		\end{align}
		Observe the following for $0<c<e^{D_{\max}(\rho\|\sigma)}$, 
		\begin{DispWithArrows}[displaystyle]
			e^{(\alpha-1)D_\alpha(\rho\|\sigma)} 
			&=  \alpha \int_{0}^{\infty} \gamma^{\alpha-1}  \Tr\left[ \sigma  \left\{  \rho > \gamma \sigma \right\} \right] \d\gamma \\
			&\geq \alpha \int_{0}^{c} \gamma^{\alpha-1}  \Tr\left[ \sigma  \left\{  \rho > \gamma \sigma \right\} \right] \d\gamma 
			\Arrow{\textnormal{by Lemma~\ref{lemm:monotone_in_gamma}}}   
			\\
			&\geq \alpha \int_{0}^{c} \gamma^{\alpha-1}  \Tr\left[ \sigma  \left\{  \rho > c \sigma \right\} \right] \d\gamma \\
			&= c^\alpha \Tr\left[ \sigma  \left\{  \rho > c \sigma \right\} \right].
		\end{DispWithArrows}
		Now, substitute $\rho\rightarrow\rho^{\otimes n}$, $\sigma\rightarrow\sigma^{\otimes n}$, $c=e^{n a}$, 
		\begin{align}
			\Tr\left[ \sigma^{\otimes n}  \left\{  \rho^{\otimes n} > e^{n a} \sigma^{\otimes n} \right\} \right] 
			&\leq e^{-\alpha n a} e^{(\alpha-1)D_\alpha(\rho^{\otimes n}\|\sigma^{\otimes n})} \\
			&= e^{-n a -\frac{(\alpha-1)}{\alpha}\left(\alpha n a - \alpha D_\alpha(\rho^{\otimes n}\|\sigma^{\otimes n})\right)}.
		\end{align}
		Similarly, for $\alpha>1$, from Equation~\eqref{eq:layer-cake_alpha>1-0}, 
		\begin{DispWithArrows}[displaystyle]
			e^{(\alpha-1)D_\alpha(\rho\|\sigma)} 
			&= (\alpha-1) \int_0^\infty \gamma^{\alpha-2} \tr[\rho\{\rho>\gamma\sigma\}] \d\gamma 
			\Arrow{\textnormal{by Lemma~\ref{lemm:monotone_in_gamma}}}   
			\\
			&\geq (\alpha-1) \int_0^c \gamma^{\alpha-2} \tr[\rho\{\rho>c\sigma\}] \d\gamma \\
			&= c^{\alpha-1} \tr[\rho\{\rho>c\sigma\}], 
		\end{DispWithArrows}
		implying
		\begin{align}
			\tr[\rho^{\otimes n}\{\rho^{\otimes n}>e^{na}\sigma^{\otimes n}\}] 
			&\leq e^{(1-\alpha)na}e^{(\alpha-1)D_\alpha(\rho^{\otimes n}\|\sigma^{\otimes n})} \\
			&= e^{-(\alpha-1)(na-D_\alpha(\rho^{\otimes n}\|\sigma^{\otimes n})}.
		\end{align}
		Finally, for $0\leq\alpha\leq1$, from Equation~\eqref{eq:layer-cake_alpha<1-0},  
		\begin{DispWithArrows}[displaystyle]
			e^{(\alpha-1)D_\alpha(\rho\|\sigma)} 
			&= (1-\alpha) \int_0^\infty \gamma^{\alpha-2} \tr[\rho\{\rho\leq\gamma\sigma\}] \d\gamma 
			\Arrow{\textnormal{by Lemma~\ref{lemm:monotone_in_gamma}}}   
			\\
			&\geq (1-\alpha) \int_c^\infty \gamma^{\alpha-2} \tr[\rho\{\rho\leq c\sigma\}] \d\gamma \\
			&= c^{\alpha-1} \tr[\rho\{\rho\leq c\sigma\}], 
		\end{DispWithArrows}
		which implies, 
		\begin{align}
			1-\tr[\rho^{\otimes n}\{\rho^{\otimes n} > e^{na}\sigma^{\otimes n}\}] = \tr[\rho^{\otimes n}\{\rho^{\otimes n}\leq e^{na}\sigma^{\otimes n}\}] &\leq e^{(1-\alpha)na + (\alpha-1)D_\alpha(\rho^{\otimes n}\|\sigma^{\otimes n})} \\
			&= e^{-n(\alpha-1)\left(a - \frac1n D_\alpha(\rho^{\otimes n}\|\sigma^{\otimes n})\right)}
		\end{align}
	\end{proof}

	Now, set $\alpha=s+1$ for some $s>0$, then
	\begin{align}
		\Tr\left[ \sigma^{\otimes n}  \left\{  \rho^{\otimes n} > e^{n a} \sigma^{\otimes n} \right\} \right] \leq e^{- n a - \left(n s a - s D_{s+1}(\rho^{\otimes n}\|\sigma^{\otimes n})\right)}, 
	\end{align}
	which implies, because it holds for all $s>0$,
	\begin{align}
		\Tr\left[ \sigma^{\otimes n}  \left\{  \rho^{\otimes n} > e^{n a} \sigma^{\otimes n} \right\} \right] \leq e^{- n a - \sup_{s>0}\left(n s a - s D_{s+1}(\rho^{\otimes n}\|\sigma^{\otimes n})\right)}, 
	\end{align}
	Note that this is tighter that the inequality used in~\cite{mosonyi2015quantum}. To see that, observe, 
	\begin{align}
		\Tr\left[ \sigma^{\otimes n}  \left\{  \rho^{\otimes n} > e^{n a} \sigma^{\otimes n} \right\} \right] 
		&\leq e^{- n a - \sup_{s>0}\left(n s a - s D_{s+1}(\rho^{\otimes n}\|\sigma^{\otimes n})\right)} \\
		&\leq e^{- n a - \sup_{s>0}\left(n s a - s \widetilde D_{s+1}(\rho^{\otimes n}\|\sigma^{\otimes n})\right)}\\
		&= e^{- n a - n\sup_{s>0}\left( s a - s \widetilde D_{s+1}(\rho\|\sigma)\right)}\\
		&= e^{- n (a+\phi(a))}, 
	\end{align}
	which is~\cite[Eq.~(55)]{mosonyi2015quantum}.

	\begin{prop}[Quantum Markov inequality]
		For some nondecreasing function $f\in \mathcal{F}$ satisfying $f(0)\geq 0$ and for all $0<c<\e^{D_{\max}(\rho \Vert \sigma)}$,
		\begin{align}
			\Tr\left[ \sigma \left\{ \rho > c \sigma \right\} \right] \leq \frac{ D_{f}(\rho \Vert \sigma )}{f(c)}.
		\end{align}
	\end{prop}
	\begin{proof}
		We start from Equation~\eqref{Eq:Def-1-dif}, 
		\begin{align}
			D_f(\rho\Vert\sigma) &=
			f(0) + \int_0^\infty f'(\gamma) \Tr\left[ \sigma \proj{ \rho > \gamma \sigma}\right] \, \d \gamma \\
			&\geq
			f(0) + \int_0^c f'(\gamma) \Tr\left[ \sigma \proj{ \rho > \gamma \sigma}\right] \, \d \gamma \\
			&\geq
			f(0) + \int_0^c f'(\gamma) \Tr\left[ \sigma \proj{ \rho > c \sigma}\right] \, \d \gamma \\
			&=
			f(0) + f(c) \Tr\left[ \sigma \proj{ \rho > c \sigma}\right] - f(0) \Tr\left[ \sigma \proj{ \rho > c \sigma}\right] \\
			&\geq f(c) \Tr\left[ \sigma \proj{ \rho > c \sigma}\right] 
		\end{align}
		where the first inequality follows because the integrand is always positive, the second because the trace function $\Tr\left[ \sigma \proj{ \rho > \gamma \sigma}\right]$ is non-increasing and the final one because its upper bounded by $1$. The claim then follows by rearranging. 
	\end{proof}
	
	\begin{theo}[Error exponent for asymmetric setting] \label{theorem:Hoeffding}
		Let $n\in\mathds{N}$ to be an arbitrary integer.
		For all $r>0$ and all $0<\alpha<1$,
		\begin{align} \label{eq:Hoeffding}
			\begin{dcases}
				\log \Tr\left[ 1 - \rho^{\otimes n} S_n(a) \right] \leq - n \frac{\alpha-1}{\alpha} \left( \frac{1}{n} D_{\alpha}\left(\rho^{\otimes n}\Vert \sigma^{\otimes n} \right) - r \right)
				% \leq  n \frac{\alpha-1}{\alpha} \left( \frac{1}{n} \bar{D}_{\alpha}\left(\rho^{\otimes n}\Vert \sigma^{\otimes n} \right) - r \right)
				\\
				\log \Tr\left[ \sigma^{\otimes n} S_n(a) \right] \leq - n r
			\end{dcases},
		\end{align}
		where the threshold in the quantum Neyman--Pearson test is chosen to be
		\begin{align} \label{eq:threshould_Hoeffding}
			a = \frac{r+(\alpha-1) \frac{1}{n}D_{\alpha} \left(\rho^{\otimes n}\Vert \sigma^{\otimes n}\right)}{\alpha}.
		\end{align}
	\end{theo}
	\begin{proof}
		Our claim follows by applying \eqref{eq:type-II_error_upper_bound} and \eqref{eq:type-I_error_upper_bound} in Proposition~\ref{prop:threshold_test} with \eqref{eq:threshould_Hoeffding}.
	\end{proof}
	
	\begin{theo}[Error exponent for symmetric setting] \label{theorem:Chernoff}
		For all $0<p<1$ and all $0<\alpha<1$,
		\begin{align}
			\begin{split} \label{eq:Chernoff}
				p \Tr\left[ \rho^{\otimes n} \left( \I - S_n(a) \right) \right]
				+ (1-p) \Tr\left[ \sigma^{\otimes n} S_n(a) \right]
				&\leq 2 p^{\alpha}(1-p)^{1-\alpha}\cdot \e^{- (1-\alpha) D_{\alpha}\left( \rho^{\otimes n} \Vert \sigma^{\otimes n} \right) },
				% \\
				% &\leq 2 p^{\alpha}(1-p)^{1-\alpha}\cdot \e^{- (1-\alpha) \bar{D}_{\alpha}\left( \rho^{\otimes n} \Vert \sigma^{\otimes n} \right) },
			\end{split}
		\end{align}
		where the threshold in the quantum Neyman--Pearson test is chosen to be
		\begin{align} \label{eq:threshould_Chernoff}
			a = \log \frac{1-p}{p}.
		\end{align}
	\end{theo}
	\begin{proof}
		The exponential term $\e^{- (1-\alpha) {D}_{\alpha}\left( \rho^{\otimes n} \Vert \sigma^{\otimes n} \right) }$ follows by applying \eqref{eq:type-II_error_upper_bound} and \eqref{eq:type-I_error_upper_bound} in Proposition~\ref{prop:threshold_test} with \eqref{eq:threshould_Chernoff}.
		We only need to calculate the prefactor as follows:
		\begin{align}
			p \e^{-(\alpha-1) a} + (1-p) \e^{-\alpha a}
			= p\left( \frac{p}{1-p} \right)^{\alpha-1}
			+ (1-p) \left( \frac{p}{1-p} \right)^{\alpha}
			= p^{\alpha} (1-p)^{1-\alpha} + p^{\alpha} (1-p)^{1-\alpha}.
		\end{align}
	\end{proof}
	
	\begin{remark}
		Note that both Theorems~\ref{theorem:Hoeffding} and \ref{theorem:Chernoff} hold for all $0<\alpha<1$.
		One can maximize over $0<\alpha<1$ to achieve the largest error exponents, respectively.
		Moreover, we can relax the established error exponents to the ones induced by the Petz--\Renyi divergence by the relation shown in Ref.~\cite{beigi2025some}:
		\begin{align}
			D_{\alpha} \left( \rho^{\otimes n} \Vert \sigma^{\otimes n} \right)
			\geq \bar{D}_{\alpha} \left( \rho^{\otimes n} \Vert \sigma^{\otimes n} \right)
			= n \bar{D}_{\alpha} \left( \rho \Vert \sigma \right).
		\end{align}
	\end{remark}

	%%%%%%%%%%%%%%%%%%%%%%%%%%%%%%%%%%%%%

	%%%%%%%%%%%%%%%%%%%%%%%%%%%%%%%%%%%%%%%%%%%%%%
	
	\section{Riemann–Stieltjes Integral Representations} \label{sec:RS-integral}
	
	In this section, we show that the quantum divergences introduced in Section~\ref{sec:relation} admit more intuitive Riemann--Stieltjes integral representations.
	As will be shown shortly in Proposition~\ref{prop:f-divergence_RS}, the conditions for the Riemann--Stieltjes integral representation of the quantum $f$-divergence only need to satisfy:
	(i) $D_{\max}(\rho\Vert\sigma) < \infty$,
	(ii) $f$ defined on $[0,\infty)$ is convex with $f(1) = 0$, and
	(iii) $f(x) < \infty$ for all $x>0$.
	Here, $f(0):= \lim_{x\searrow 0} f(x)$ could even be infinite.
	Later in Section~\ref{sec:relative_entropy_RS}, we show that the established Riemann–Stieltjes integral representations can be expressed in terms of some induced \emph{Riemann–Stieltjes distributions}.

	\begin{prop} \label{prop:f-divergence_RS}
		For quantum states $\rho$ and $\sigma$ satisfying $D_{\max}(\rho\Vert\sigma) < \infty$ and any convex function $f$ on $[0,\infty)$ such that $f(x)<\infty$ for all $x \geq 0$ and $f(1) = 0$,
		\begin{align}
			D_f(\rho\|\sigma) 
			&= -\int_{0}^{\infty} f(\gamma) \, \d\Tr\left[ \sigma  \left\{  \rho > \gamma \sigma \right\} \right]  \label{eq:f-divergence_RS_1}
			\\
			&= \int_{0}^{\infty} f(\gamma) \, \d\Tr\left[ \sigma  \left\{  \rho \leq \gamma \sigma \right\} \right]. \label{eq:f-divergence_RS_2}
		\end{align}
	\end{prop}
	\begin{proof}
		Using integration by parts,
		% \begin{DispWithArrows}[displaystyle]
			\begin{align}
				D_f(\rho\|\sigma) 
				&=   f(0) +   \int_{0}^{\infty}  \Tr\left[ \sigma  \left\{  \rho > \gamma \sigma \right\} \right] \d f(\gamma)
				\\
				&=  f(0) +  \left. f(\gamma) \Tr\left[ \sigma  \left\{  \rho > \gamma \sigma \right\} \right]
				\right|_0^\infty -\int_{0}^{\infty} f(\gamma)  \d\Tr\left[ \sigma  \left\{  \rho > \gamma \sigma \right\} \right] \d\gamma
				% \Arrow{$\left\{  \rho > \gamma \sigma \right\} = 0$ for all $\gamma > \e^{D_{\max}(\rho\Vert\sigma)}$}
				\\
				&\overset{\textnormal{(a)}}{=} f(0) - f(0) -\int_{0}^{\infty} f(\gamma)  \d\Tr\left[ \sigma  \left\{  \rho > \gamma \sigma \right\} \right] \d\gamma
				\\
				&= -\int_{0}^{\infty} f(\gamma)  \d\Tr\left[ \sigma  \left( \I - \left\{  \rho \leq \gamma \sigma \right\} \right) \right] \d\gamma
				\\
				&= \int_{0}^{\infty} f(\gamma)  \d\Tr\left[ \sigma \left\{  \rho \leq \gamma \sigma \right\} \right] \d\gamma,
			\end{align}
			% \end{DispWithArrows}
		where (a) is because $\left\{  \rho > \gamma \sigma \right\} = 0$ for all $\gamma > \e^{D_{\max}(\rho\Vert\sigma)}$ and $f(\gamma) < \infty$ for all finite $\gamma$.
	\end{proof}
	
	To see why the Riemann--Stieltjes integral representation given in Proposition~\ref{prop:f-divergence_RS} is a more intuitive expression even in the classical case (for a discrete probablity space), let us elaborate on \eqref{eq:f-divergence_RS_2} in the commuting case.
	
	Denote the spectral decompositions for the commuting states $\rho$ and $\sigma$, $\rho\ll\sigma$, as
	\begin{align*}
		\rho = \sum\nolimits_i \lambda_i \Pi_i,
		\quad
		\sigma = \sum\nolimits_i \mu_i \Pi_i,
	\end{align*}
	where $(\lambda_i)_i$ and $(\mu_i)_i$ are eigenvalues, and $(\Pi_i)_i$ are eigen-projections.
	We write the projection as
	\begin{align*}
		\proj{ \rho \leq \gamma \sigma}
		= \sum\nolimits_i \proj{ \nicefrac{\lambda_i}{\mu_i} \leq \gamma } \Pi_i.
	\end{align*}
	Note that for each $i$, the indicator function $\proj{ \nicefrac{\lambda_i}{\mu_i} \leq \gamma }$ is a
	right-continuous monotone increasing cumulative distribution function in $\gamma \geq 0$.
	Then,
	\begin{align*}
		\int_0^\infty f(\gamma) \, \d \proj{ \rho \leq \gamma \sigma}
		&= \sum\nolimits_i \int_0^\infty f(\gamma) \, \d \proj{ \nicefrac{\lambda_i}{\mu_i} \leq \gamma } \Pi_i
		\\
		&= \sum\nolimits_i \int_0^\infty f(\gamma) \delta\left( \nicefrac{\lambda_i}{\mu_i} - \gamma \right) \Pi_i \, \d \gamma
		\\
		&= \sum\nolimits_i f\left( \nicefrac{\lambda_i}{\mu_i} \right) \Pi_i,
	\end{align*}
	where we have used the Dirac delta function $\delta$ for simplicity.
	Taking $\Tr\left[ \sigma (\,\cdot\,)\right]$ on both sides (i.e., taking the expectation under the state $\sigma$), we recover the classical $f$-divergence between the eigenvalues $(\lambda_i)_i$ and $(\sigma_i)_i$:
	\begin{align}
		\int_0^\infty f(\gamma) \, \d \Tr[\sigma \proj{ \rho \leq \gamma \sigma}]
		=  
		\sum\nolimits_i \mu_i f\left( \nicefrac{\lambda_i}{\mu_i} \right)
		= D_f\left( (\lambda_i)_i \Vert (\mu_i)_i \right).
	\end{align}
	
	\medskip
	Returning to the general setting, Proposition~\ref{prop:f-divergence_RS} then also provides the Riemann--Stieltjes integral representations for the quasi \Renyi divergence as well, i.e.,
	\begin{tcolorbox}[size = small, colback=orange!1.5!white, colframe=orange,  boxrule=0.5pt, width = \linewidth]
		\begin{align}
			Q_{\alpha} (\rho\Vert\sigma)
			= - \int_{0}^\infty \gamma^{\alpha} \, \d \Tr\left[ \sigma \proj{ \rho > \gamma \sigma} \right]
			= \int_{0}^\infty \gamma^{\alpha} \, \d \Tr\left[ \sigma \proj{ \rho \leq \gamma \sigma} \right], \quad \alpha > 0
		\end{align}
	\end{tcolorbox}
	\noindent which we consider one of the most intuitive expressions for the quasi \Renyi divergence with $\rho\ll\sigma$.
	%\HC{
		%\begin{align}
		%    Q_{\alpha} (\rho\Vert\sigma)
		%    & = \alpha \int_0^\infty \gamma^{\alpha-1} \Tr\left[ \sigma \proj{\rho > \gamma \sigma} \right] \d \gamma
		%    \\
		%    &= \int_0^\infty \Tr\left[ \sigma \proj{\rho > \gamma \sigma} \right] \d \gamma^{\alpha}
		%    \\
		%    &=   \gamma^{\alpha} \Tr\left[ \sigma \proj{\rho>\gamma \sigma} \right]\Big|_0^\infty
		%    - \int_0^\infty \gamma^{\alpha} \, \d \Tr\left[ \sigma \proj{\rho > \gamma \sigma } \right]
		%    \\
		%    &=- \int_0^\infty \gamma^{\alpha} \, \d \Tr\left[ \sigma \proj{\rho > \gamma \sigma } \right]
		%\end{align}
		%}

	\subsection{The relative entropy} \label{sec:relative_entropy_RS}
	
	Recall that for probability distributions $P$ and $Q$ with $P\ll Q$, there are two ways of expressing the Kullback--Leibler divergence:
	\begin{align}
		D(P\Vert Q)
		= \mathds{E}_{Q} \left[ \frac{\d P}{\d Q} \log \frac{\d P}{\d Q} \right]
		= \mathds{E}_{P} \left[ \log \frac{\d P}{\d Q} \right],
	\end{align}
	which can be considered as a change-of-measure argument (so as for the \Renyi divergence $D_{\alpha}(P\Vert Q)$, $\alpha > 1$).
	
	Using the Riemann--Stieltjes integral representation in Proposition~\ref{prop:f-divergence_RS}, we will show that the quantum relative entropy also inherits a similar flavor.
	For quantum states $\rho$ and $\sigma $ satisfying $\rho\ll\sigma$, we define the following induced functions:
	\begin{tcolorbox}[size = small, colback=orange!1.5!white, colframe=orange,  boxrule=0.5pt, width = \linewidth]
		\begin{align} \label{eq:RS-distribution}
			P_{\rho,\sigma}^{\texttt{RS}}(\gamma) := \Tr[\rho\{\rho\leq\gamma\sigma\}],
			\quad
			Q_{\rho,\sigma}^{\texttt{RS}}(\gamma) := \Tr[\sigma\{\rho\leq\gamma\sigma\}], \quad \gamma \geq 0.
		\end{align}
	\end{tcolorbox}
	\noindent 
	Note that by Lemma~\ref{lemm:differentiability_projection}, $P_{\rho,\sigma}^{\texttt{RS}}(\gamma)$
	and 
	$Q_{\rho,\sigma}^{\texttt{RS}}(\gamma)$
	are right-continuous, monotone increasing functions from $0$ to $1$ (see e.g., Figures~\ref{figure:RS-integral_quantum} and \ref{figure:RS-integral_quantum_increasing}).
	Hence, we call them \emph{Riemann--Stieltjes distributions} (as opposed to the Nussbaum--Szko{\l}a distributions \cite{NS09}) for states $(\rho,\sigma)$.
	We may use them to express the quantum relative entropy as follows (so as for the quantum $f$-divergence in Proposition~\ref{prop:f-divergence_RS}).
	We refer the readers to Figures~\ref{figure:decreasing} and \ref{figure:increasing} for illustrations of the commuting case.

	\begin{prop}[Riemann--Stieltjes integral representations for Umegaki's relative entropy] \label{prop:relative_entropy_formula2}
		For quantum states $\rho$ and $\sigma$ satisfying $D_{\max}(\rho \Vert \sigma) < \infty $, 
		\begin{align}
			\log A - \log B
			&= - \int_{0}^{\infty} \log \gamma \, \d \proj{A> \gamma B}, \quad A,B>0,
			\\
			D(\rho\Vert\sigma) 
			&= 
			- \int_0^\infty  \log \gamma \, \d \tr\left[\rho\{\rho>\gamma\sigma\}\right]
			= \int_0^\infty  \log \gamma \, \d P_{\rho,\sigma}^{\textnormal{\texttt{RS}}}(\gamma)
			\label{eq:first-row}
			\\
			&= 
			- \int_0^\infty  \gamma \log \gamma \, \d \tr\left[\sigma\{\rho>\gamma\sigma\}\right]
			= \int_0^\infty  \gamma \log \gamma \, \d Q_{\rho,\sigma}^{\textnormal{\texttt{RS}}}(\gamma).
			\label{eq:second-row}
		\end{align}
	\end{prop}
	\begin{proof}
		The equalities in the second row \eqref{eq:second-row} directly follow from Proposition~\ref{prop:f-divergence_RS} with $f(x) = x \log x$, and then the first row \eqref{eq:first-row} follows from a change-of-measure argument in Lemma~\ref{lemm:change-of-measure} below.
		In the following, we provide an alternative proof to show the first row \eqref{eq:first-row} directly, and then recover the second row \eqref{eq:second-row}.
		
		We first recall \eqref{eq:log-difference1} to obtain
		% \begin{DispWithArrows}[displaystyle]
			\begin{align*}
				\log A - \log B
				&= \int_1^\infty \frac{1}{\gamma} \proj{A>\gamma B} \d \alpha
				- \int_1^\infty \frac{1}{\alpha} \proj{B > \alpha A} \d \alpha
				% \Arrow{$\alpha \colon$ $\{\alpha A = B\}$ \\ \textnormal{is measure} $0$}
				\notag
				\\
				&\overset{\textnormal{(a)}}{=} \int_1^\infty \frac{1}{\gamma} \proj{A>\gamma B} \d \gamma
				- \int_1^\infty \frac{1}{\alpha} \proj{B \geq \alpha A} \d \alpha
				\\
				&\overset{\textnormal{(b)}}{=} 
				\int_1^\infty \frac{1}{\gamma} \proj{A>\gamma B} \d \gamma
				- \int_0^1 \frac{1}{\gamma} \proj{A \leq \gamma B} \d \gamma
				\\
				&= \int_1^\infty \frac{1}{\gamma} \proj{A>\gamma B} \d \gamma
				- \lim_{\epsilon \searrow 0} \left\{
				\int_{\epsilon}^1 \frac{1}{\gamma } \I \, \d \gamma - \int_{\epsilon}^1 \frac{1}{\gamma} \proj{A>\gamma B} \, \d \gamma 
				\right\}
				\\
				&= \lim_{\epsilon \searrow 0} \left\{
				\int_{\epsilon}^\infty \frac{1}{\gamma} \proj{A>\gamma B} \, \d \gamma + \log \epsilon \cdot \I 
				\right\}
				\\
				&\overset{\textnormal{(c)}}{=}
				\lim_{\epsilon\searrow 0} \left\{
				\log \gamma \proj{A>\gamma B} \big\vert_{\epsilon}^\infty - \int_{\epsilon}^\infty \log \gamma \, \d\proj{A>\gamma B}
				+ \log \epsilon \cdot \I
				\right\}
				\\
				&\overset{\textnormal{(d)}}{=}
				\lim_{\epsilon\searrow 0} \left\{
				0 - \log \epsilon \cdot \proj{A>\epsilon B} + \log \epsilon \cdot \I - \int_{\epsilon}^{\infty} \log \gamma \, \d \proj{A> \gamma B}
				\right\}
				\\
				&= \lim_{\epsilon\searrow 0} \left\{
				\log \epsilon \cdot \proj{A\leq \epsilon B} - \int_{\epsilon}^{\infty} \log \gamma \, \d \proj{A> \gamma B}
				\right\}
				\\
				&\overset{\textnormal{(e)}}{=}
				\lim_{\epsilon\searrow 0} \left\{
				- \int_{\epsilon}^{\infty} \log \gamma \, \d \proj{A> \gamma B}
				\right\}
				\\
				&=
				- \int_{0}^{\infty} \log \gamma \, \d \proj{A> \gamma B}.
			\end{align*}
			% \end{DispWithArrows}
		In (a), the set $\{ \alpha : \proj{B=\alpha A} \neq 0\}$ is measure zero.
		In (b), we change variable $\gamma \leftarrow \nicefrac{1}{\alpha}$ in the second term.
		In (c), we use integration by parts.
		In (d), we have $\proj{A>\gamma B} = 0$ for all $\gamma > \left\| B^{-\nicefrac{1}{2}} A B^{-\nicefrac{1}{2}} \right\|_{\infty}$.
		In (e), we have $\proj{A\leq \epsilon B} = 0$ for all $\epsilon$ strictly less than the smallest eigenvalue of $B^{-\nicefrac{1}{2}} A B^{-\nicefrac{1}{2}}$.
		
		Finally, take $\Tr[A(\cdot)]$ on both sides,
		substitute $B \leftarrow \sigma $ and $A \leftarrow \rho + \epsilon \I$. 
		By the continuity of the quantum relative entropy in its first argument and and letting $\epsilon \to 0$, we conclude the proof.
	\end{proof}

	\begin{lemm}[Change of measure] \label{lemm:change-of-measure}
		For all $A\geq 0$, $B>0$, and measurable functions $g$,
		\begin{align}
			\int_0^\infty g(\gamma) \, \d \Tr\left[ A \proj{ A > \gamma B } \right]
			&= \int_0^\infty \gamma g(\gamma) \, \d \Tr\left[ B \proj{ A > \gamma B } \right],
			\label{eq:change-of-measure1}
			\\
			\int_0^\infty g(\gamma) \, \d \Tr\left[ A \proj{ A \leq \gamma B  } \right]
			&= \int_0^\infty \gamma g(\gamma) \, \d \Tr\left[ B \proj{A \leq \gamma B } \right].
		\end{align}
		
		Moreover, for differentiable function $g$, such that $g(0)=0$, we have
		\begin{align} \label{eq:change-of-measure_combined}
			\int_0^\infty g'(\gamma) \Tr[A\{A>\gamma B\}] \, \d \gamma=\int_0^\infty \left(g(\gamma)+\gamma g'(\gamma)\right) \Tr[A\{A>\gamma B\}] \, \d \gamma.
		\end{align}
	\end{lemm}
	
	\begin{proof}
		We denote the spectrum $\spec\left\{B^{-1/2} A B^{-1/2} \right\}$ by $\mathcal{S}$.
		Recall that the map $\gamma \mapsto \proj{A > \gamma B}$ is piecewise analytical except on $\mathcal{S}$ (Lemma~\ref{lemm:differentiability_projection}).
		For $\gamma \not\in \mathcal{S}$, we calculate 
		% \begin{align}
			%     \frac{\d}{\d \gamma} \proj{A - \gamma B}
			%     &= \frac{\d}{\d \gamma} \left( \proj{A - \gamma B} \proj{A - \gamma B} \right)
			%     \\
			%     &= \frac{\d}{\d \gamma} \left( \proj{A - \gamma B}  \right) \proj{A - \gamma B} + \proj{A - \gamma B} \frac{\d}{\d \gamma} \left( \proj{A - \gamma B} \right).
			% \end{align}
		\begin{align}
			\frac{\d}{\d \gamma} \Tr\left[ \left( A - \gamma B \right)_+ \right] 
			&= \Tr\left[ \frac{\d}{\d \gamma} (A - \gamma B ) \cdot \{ A > \gamma B \} \right] + \Tr\left[ (A - \gamma B) \frac{\d}{\d \gamma } \{ A > \gamma B \}  \right]
			\\
			&= -\Tr\left[ B \{ A > \gamma B \} \right] + \Tr\left[ (A - \gamma B) \frac{\d}{\d \gamma } \{ A > \gamma B \}  \right].
		\end{align}
		On the other hand, Lemma~\ref{lemm:HP14} shows that
		\begin{align}
			\frac{\d}{\d \gamma} \Tr\left[ (A - \gamma B)_+ \right]
			= - \Tr[B \{ A > \gamma B \} ].
		\end{align}
		This gives
		\begin{align}
			\Tr\left[ (A - \gamma B ) \frac{\d}{\d \gamma} \{A > \gamma B \} \right] = 0,
		\end{align}
		or equivalently,
		\begin{align}
			\frac{\d}{\d \gamma} \Tr\left[ A \{ A > \gamma B\} \right]
			&= \gamma \frac{\d}{\d \gamma} \Tr\left[ B \{ A > \gamma B \right].
		\end{align}
		Then,
		\begin{align}
			\int_{\gamma \not\in \mathcal{S}} g(\gamma) \, \d \Tr\left[ A \{ A > \gamma B \} \right]
			&= \int_{\gamma \not\in \mathcal{S}} g(\gamma) \frac{\d}{\d \gamma} \Tr\left[ A \{ A > \gamma B \} \right]  \d \gamma
			\\
			&= \int_{\gamma \not\in \mathcal{S}} \gamma g(\gamma) \frac{\d}{\d \gamma} \Tr\left[ B \{ A > \gamma B \} \right]  \d \gamma
			\\
			&= \int_{\gamma \not\in \mathcal{S}} \gamma g(\gamma) \, \d \Tr\left[ B \{ A > \gamma B \}\right] \label{eq:change-of-measure1}
		\end{align}
		by noting that both the maps $\gamma \mapsto \Tr\left[ A \{ A > \gamma B \}\right]$ and $\gamma \mapsto \Tr\left[ B \{ A > \gamma B \} \right]$ are absolutely continuous on $[0,\infty)\backslash \mathcal{S}$ as they are differentiable with uniformly bounded right derivatives on $[0,\infty)\backslash \mathcal{S}$.
		
		Next, we move on to the case of $\gamma \in \mathcal{S}$.
		For a map $\gamma \mapsto h(\gamma)$, we denote the jump at $\gamma$ by
		\begin{align}
			\delta(h(\gamma))
			:= h(\gamma^+) - h(\gamma^-).
		\end{align}
		for notational simplicity.
		We claim that
		\begin{align}
			\Tr\left[ (A - \gamma B) \cdot \delta\left( \Tr[A\{A>\gamma B\}] \right) \right],
		\end{align}
		which is equivalent to
		\begin{align} \label{eq:change-of-measure-jump}
			\Tr\left[ A \cdot \delta\left( \Tr[A\{A>\gamma B\}] \right) \right]
			= \gamma \Tr\left[ B \cdot \delta\left( \Tr[A\{A>\gamma B\}] \right) \right].
		\end{align}
		Indeed,
		\begin{DispWithArrows}[displaystyle]
			\delta\left( \Tr[A\{A>\gamma B\}] \right)
			&= \Tr\left[ A \left( \{A>\gamma^+ B \} - \{A > \gamma^- B \} \right) \right]
			\Arrow{\textnormal{by Lemma~\ref{lemm:differentiability_projection}}}
			\\
			&= \Tr\left[ A \left( \{A>\gamma B \} - \{A \geq \gamma B \} \right) \right]
			\\
			&= -\Tr\left[ A \{A - \gamma B = 0\} \right]
			\\ 
			&= -\gamma \Tr\left[ B \{A - \gamma B = 0\} \right]
			\\
			&= \gamma \left( \Tr\left[ B \left( \{A>\gamma B \} - \{A \geq \gamma B \} \right) \right] \right)
			\Arrow{\textnormal{by Lemma~\ref{lemm:differentiability_projection}}}
			\\
			&= \gamma \left( \Tr\left[ B \left( \{A>\gamma^+ B \} - \{A \geq \gamma^- B \} \right) \right] \right)
			\\
			&= \gamma \cdot \delta \left( \Tr[B\{A>\gamma B\}] \right),
		\end{DispWithArrows}
		which proves \eqref{eq:change-of-measure-jump}.
		
		To show the claim of the lemma, we let
		\begin{align}
			h_1(\gamma) &:= \Tr\left[ A \{A>\gamma B \} \right],
			\\
			h_2(\gamma) &:= \Tr\left[ B \{A>\gamma B \} \right],
		\end{align}
		and let $\bar{h}_1(\gamma)$ be a simple function (i.e., a linear combination of linear functions) such that $h_1 - \bar{h}_1$ is continuous.
		In other words,
		\begin{align} \label{eq:change_jump}
			\delta(h_1(\gamma) ) = \delta (\bar{h}_1(\gamma)), \quad \forall \gamma \in \mathcal{S}.
		\end{align}
		We define $\bar{h}_2$ similarly.
		Finally, we have
		\begin{DispWithArrows}[displaystyle]
			\int_0^\infty g(\gamma) \, \d h_1(\gamma) 
			&= \int_0^\infty g(\gamma) \, \d \bar{h}_1(\gamma)  + \int_0^\infty g(\gamma) \, \d \left( h_1(\gamma) - \bar{h}_1(\gamma)\right)
			\notag
			\\
			&= \sum_{\gamma \in \mathcal{S}} g(\gamma) \delta(\bar{h}_1(\gamma)) + \int_0^\infty g(\gamma) \frac{\d \left( h_1(\gamma) - \bar{h}_1(\gamma)\right) }{\d \gamma} \, \d \gamma
			\Arrow{\textnormal{by} \eqref{eq:change_jump}}
			\notag
			\\
			&= \sum_{\gamma \in \mathcal{S}} g(\gamma) \delta({h}_1(\gamma)) + \int_{0}^\infty g(\gamma) \frac{\d \left( h_1(\gamma) - \bar{h}_1(\gamma)\right) }{\d \gamma} \, \d \gamma
			\Arrow{\textnormal{by ignoring} $\gamma \in \mathcal{S}$}
			\notag
			\\
			&= \sum_{\gamma \in \mathcal{S}} g(\gamma) \delta({h}_1(\gamma)) + \int_{\gamma \not\in \mathcal{S}} g(\gamma) \frac{\d h_1(\gamma)  }{\d \gamma} \, \d \gamma
			\Arrow{\textnormal{by} \eqref{eq:change-of-measure1}}
			\notag
			\\
			&= \sum_{\gamma \in \mathcal{S}} g(\gamma) \delta({h}_1(\gamma)) + \int_{\gamma \not\in \mathcal{S}} \gamma g(\gamma) \frac{\d h_2(\gamma)  }{\d \gamma} \, \d \gamma
			\Arrow{\textnormal{by ignoring} $\gamma \in \mathcal{S}$}
			\notag
			\\
			&= \sum_{\gamma \in \mathcal{S}} g(\gamma) \delta({h}_1(\gamma)) + \int_{0}^\infty \gamma g(\gamma) \frac{\d \left( h_2(\gamma) - \bar{h}_2(\gamma) \right) }{\d \gamma} \, \d \gamma
			\Arrow{\textnormal{by} \eqref{eq:change-of-measure-jump}}
			\notag
			\\
			&= \sum_{\gamma \in \mathcal{S}} \gamma g(\gamma) \delta({h}_2(\gamma)) + \int_{0}^\infty \gamma g(\gamma) \frac{\d \left( h_2(\gamma) - \bar{h}_2(\gamma) \right)  }{\d \gamma} \, \d \gamma
			\notag
			\\
			&= \int_0^\infty \gamma g(\gamma) \, \d h_2(\gamma)
			\notag
		\end{DispWithArrows}
		proving our first claim.
		
		Note that 
		\begin{align}
			\d \Tr\left[ A \proj{ A > \gamma B } \right]
			&= - \d \Tr\left[ A \proj{ A \leq \gamma B } \right],
			\\
			\d \Tr\left[ B \proj{ A > \gamma B } \right]
			&= - \d \Tr\left[ B \proj{ A \leq \gamma B } \right].
		\end{align}
		The second claim follows.
		
		For the last claim,
		we apply integration by parts on both sides of \eqref{eq:change-of-measure1}
		to obtain
		\begin{align}
			&g(\gamma)\Tr[A\{A>\gamma B\}]\biggr\vert_0^\infty -\int_0^\infty g'(\gamma) \Tr[A\{A>\gamma B\}] \, \d \gamma\\
			&=\gamma g(\gamma)\Tr[A\{A>\gamma B\}]\biggr\vert_0^\infty -\int_0^\infty \left(g(\gamma)+\gamma g'(\gamma)\right) \Tr[A\{A>\gamma B\}] \, \d \gamma,
		\end{align}
		or
		\begin{align}
			\int_0^\infty g'(\gamma) \Tr[A\{A>\gamma B\}] \, \d \gamma=\int_0^\infty \left(g(\gamma)+\gamma g'(\gamma)\right) \Tr[A\{A>\gamma B\}] \, \d \gamma,
		\end{align}
		which concludes the proof.
	\end{proof}

	\begin{figure}[ht]
		\centering
		\includegraphics[width=0.5\linewidth]{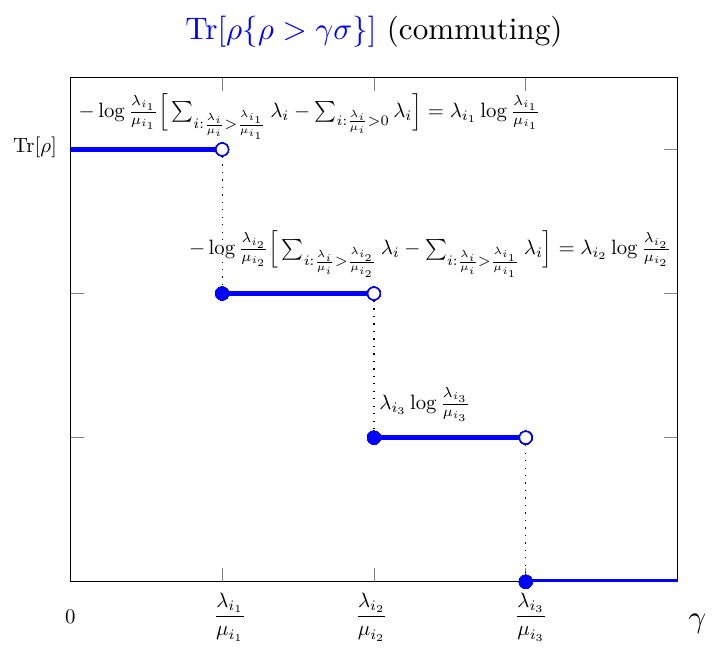}
		\caption{An illustration of the Riemann--Stieltjes integral representation for $D(\rho\Vert\sigma)= -\int_0^\infty \log \gamma \,\mathrm{d} \text{Tr}\left[\rho\{ \rho > \gamma \sigma \}\right]$ (Proposition~\ref{prop:relative_entropy_formula2}) for the commuting case: $\rho = \sum_i \lambda_i |i\rangle \langle i|$ and $ \sigma = \sum_i \mu_i |i\rangle \langle i|$.
			The plotted function in blue is $\gamma\mapsto \Tr[\rho\{\rho>\gamma \sigma\}]$.
			The indices $i_1, i_2,\ldots$ are chosen such that 
			${\lambda_{i_1}}/{\mu_{i_1}} \leq {\lambda_{i_2}}/{\mu_{i_2}} \leq \cdots$.
			The equation above each strip (partition) is $- \log \gamma_{i_k} \left( \text{Tr}\left[\rho\{ \rho > \gamma_{i_k} \sigma \}\right] - \text{Tr}\left[\rho\{ \rho > \gamma_{i_{k-1}} \sigma \}\right] \right) = \lambda_{i_k} \log \gamma_{i_k}$, where $\gamma_{i_k} = \lambda_{i_k}/\mu_{i_k}$ for some $k$.
			The overall sum over each partition is $\sum_k \lambda_{i_k} \log ({\lambda_{i_k}}/{\mu_{i_k}}) = D(\rho\Vert \sigma)$.
		} \label{figure:decreasing}
	\end{figure}
	
	\begin{figure}[ht]
		\centering
		\includegraphics[width=0.5\linewidth]{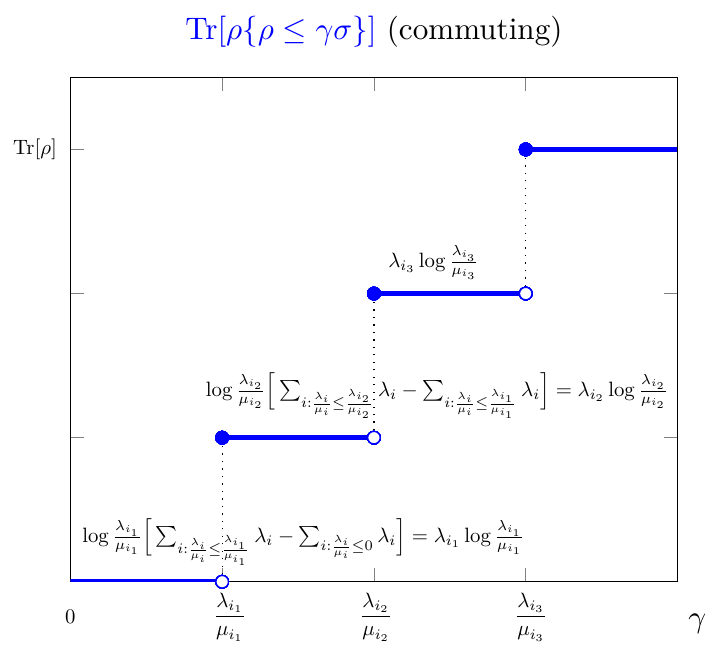}
		\caption{An illustration of the Riemann--Stieltjes integral representation for $D(\rho\Vert\sigma)= \int_0^\infty \log \gamma \,\mathrm{d} \text{Tr}\left[\rho\{ \rho \leq \gamma \sigma \}\right]$ (Proposition~\ref{prop:relative_entropy_formula2}) for the commuting case: $\rho = \sum_i \lambda_i |i\rangle \langle i|$ and $ \sigma = \sum_i \mu_i |i\rangle \langle i|$.
			The plotted function in blue is $\gamma\mapsto \Tr[\rho\{\rho<\gamma \sigma\}]$.
			The indices $i_1, i_2,\ldots$ are chosen such that 
			${\lambda_{i_1}}/{\mu_{i_1}} \leq {\lambda_{i_2}}/{\mu_{i_2}} \leq \cdots$.
			The equation above each strip (partition) is $\log \gamma_{i_k} \left( \text{Tr}\left[\rho\{ \rho \leq \gamma_{i_{k}} \sigma \}\right] - \text{Tr}\left[\rho\{ \rho \leq \gamma_{i_{k-1}} \sigma \}\right] \right) = \lambda_{i_k} \log \gamma_{i_k}$, where $\gamma_{i_k} = \lambda_{i_k}/\mu_{i_k}$ for some $k$.
			The overall sum over each partition is $\sum_k \lambda_{i_k} \log ({\lambda_{i_k}}/{\mu_{i_k}}) = D(\rho\Vert \sigma)$.
		}
		\label{figure:increasing}
	\end{figure}

	%%%%%%%%%%%%%%%%%%%%%%%%%%%%%%%%%%%%%%%%%%%%%%
	%%%%%%%%%%%%%%%%%%%%%%%%%%%%%%%%%%%%%%%%%%%%%%
	\newpage
	\section{Quantum \(f\)-divergence duality} \label{sec:f-duality}
	
	The classical $f$-divergence $D_f(P \Vert Q) = \mathds{E}_{Q}\left[ f\left(\frac{\d P}{\d Q}\right) \right]$ admits a variational formula via convex duality (see e.g.~\cite[\S 7.14]{Polyanskiy_book2024}):
	\begin{align}
		D_f(P \Vert Q ) 
		% &= \sup_{g: \mathcal{X}\to \mathds{R}} \left\{  \int_{\mathcal{X}} g(x) \, \d P(x) - \int_{\mathcal{X}} f^\star \circ g(x) \, \d Q(x) \right\}
		% \\
		&= \sup_{g: \mathcal{X}\to \mathds{R}} \left\{  \mathds{E}_{P}[g(X)] - \mathds{E}_{Q}\left[ f^\star ( g(X) ) \right] \right\}, \label{eq:convex_conjugate}
	\end{align}
	where the real-valued function $f^\star$ is the \emph{convex conjugate} (or \emph{Legendre--Fenchel transform}) of the convex function $f$ defined as
	\begin{align}
		f^\star(y) := \sup_{x \in \texttt{dom}(f)} \left\{ y x - f(x) \right\}.
	\end{align}
	It is well known that $f^\star$ is convex and continuous on its domain, and a biconjugation relation holds if $f$ is convex, i.e.,~$(f^\star)^\star = f$.
	Note that if $f$ is continuous and convex, then for fixed $y$ the optimizer is $x={f'}^{-1}(y)$.

	Via the Riemann--Stieltjes integral representation in Section~\ref{sec:relative_entropy_RS}, we show below that the quantum $f$-divergence naturally inherits such a variational representation
	via the Riemann--Stieltjes integral formula given in Proposition~\ref{prop:f-divergence_RS}.
	
	\begin{theo}[Quantum $f$-divergence duality] \label{theorem:f-duality}
		For any convex differentiable function $f$ on $[0,\infty)$,
		\begin{align}
			D_f(\rho \Vert \sigma)
			&= \sup_{g: [0,\infty)\to \mathds{R}}  \left\{
			\int_0^\infty g(\gamma) \, \d \Tr[\rho\{\rho\leq\gamma \sigma] - \int_0^\infty f^\star(g(\gamma)) \, \d \Tr[\sigma \{\rho\leq\gamma\sigma\}] \right\}. 
		\end{align}
	\end{theo}
	\begin{proof}
		Using the biconjugation $(f^\star)^\star = f$, the Riemann--Stieltjes integral representation given in \eqref{eq:f-divergence_RS_1} implies that
		\begin{align}
			D_f(\rho \Vert \sigma)
			&=   \int_{0}^{\infty} f(\gamma)\,  \d\Tr\left[ \sigma  \left\{  \rho \leq \gamma \sigma \right\} \right] 
			\\
			&= \int_{0}^{\infty} \sup_{y\in \texttt{dom}(f^\star)} \left\{ y \gamma - f^\star(y) \right\}  \d\Tr\left[ \sigma  \left\{  \rho \leq \gamma \sigma \right\} \right] 
			\\
			&\geq \int_{0}^{\infty}  g(\gamma) \gamma - f^\star(g(\gamma))\, \d\Tr\left[ \sigma  \left\{  \rho \leq \gamma \sigma \right\} \right] 
		\end{align}
		for all measurable functions $g: [0,\infty)\to \mathds{R}$. 
		Here, the inequality ``$\geq$'' because the map $\gamma \mapsto \Tr\left[ \sigma  \left\{  \rho \leq \gamma \sigma \right\} \right]$ is nondecreasing by Lemma~\ref{lemm:monotone_in_gamma}. 
		
		By a change of measure in Lemma~\ref{lemm:change-of-measure}, i.e.,
		\begin{align}
			\int_0^\infty g(\gamma) \gamma \d\Tr\left[ \sigma  \left\{  \rho \leq \gamma \sigma \right\} \right] 
			= \int_0^\infty g(\gamma) \d\Tr\left[ \rho  \left\{  \rho \leq \gamma \sigma \right\} \right],
		\end{align}
		and taking the supremum over all measurable functions $g: [0,\infty)\to \mathds{R}$, we prove one direction.
		
		Note that a convex conjugate in \eqref{eq:convex_conjugate} can be achieved by the derivative of $f$, i.e.,~$y = f'(x)$.
		Hence, the lower bound is achieved by $g(\gamma) = f'(\gamma)$.
	\end{proof}

	It is well known that the convex conjugate of $f(x) = x\log x$ is $f^\star(y) = \e^{y-1}$.
	Theorem~\ref{theorem:f-duality} implies the following variational formula.
	\begin{coro}[Variational representation for Umegaki's relative entropy via the Riemann--Stieltjes integral]
		For finite-dimensional quantum states $\rho$ and $\sigma$ satisfying $D_{\max}(\rho \Vert \sigma) < \infty $, 
		\begin{align} \label{eq:variational-RS}
			D(\rho\Vert\sigma)
			&=  \sup_{g: [0,\infty)\to \mathds{R}}  \left\{
			\int_0^\infty g(\gamma) \, \d \Tr[\rho\{\rho\leq\gamma \sigma] - \int_0^\infty \e^{g(\gamma)-1} \, \d \Tr[\sigma \{\rho \leq \gamma\sigma\}] \right\}.
		\end{align}
	\end{coro}
	
	\begin{remark}
		The variational formula in \eqref{eq:variational-RS} is not the same as the Donsker--Varadhan expression \cite{Pet88}:
		\begin{align}
			D(\rho\Vert\sigma)
			= \sup_{\omega >0} \left\{  \Tr[ \rho \log \omega ] - \log \Tr\left[ \e^{\log \sigma + \log \omega }\right] \right\}.
		\end{align}
		(which is also mentioned even in the classical case \cite[Example 7.6]{Polyanskiy_book2024}),
		in the sense that the former is formed by the scalar inner product while the latter is formed by the Hilbert--Schmidt inner product.
	\end{remark}
	
	The convex conjugate of $f(x)=\frac12 |x-1|$ is $f^\star(y) = y$ for $|y|\leq \frac12$ and $f^\star(y) = \infty$, otherwise.
	\begin{coro}
		\begin{align}
			\frac{1}{2}\left\| \rho - \sigma \right\|_1
			&= \sup_{ \|g\|_\infty \leq 1} \frac12 \int_0^\infty g(\gamma) \, \d \Tr[(\sigma-\rho) \{\rho\leq\gamma\sigma\}]. 
			% - \d \Tr[\rho\{\rho>\gamma\sigma\}].
		\end{align}
	\end{coro}

	The convex conjugate of $f(x) = (x-1)^2$ is $f^\star(y) = y + \frac{y^2}{4}$.
	\begin{coro}
		\begin{align}
			\chi^2(\rho\Vert\sigma)
			&= \sup_{g:[0,\infty)\to\mathds{R}} \left\{ \int_0^\infty g(\gamma) \, \d\Tr[\rho\proj{\rho\leq\gamma \sigma}]
			-  \int_0^\infty \left( g(\gamma) + \frac{g(\gamma)^2}{4} \right) \, \d \Tr[\sigma \proj{\rho\leq\gamma \sigma}]
			\right\} \\
			&= \sup_{g:[0,\infty)\to\mathds{R}} \left\{  \int_0^\infty g(\gamma) \, \d\Tr[\rho\proj{\rho\leq\gamma \sigma}]
			-  \int_0^\infty \left( 1 + \frac{g(\gamma)^2}{4} \right) \, \d \Tr[\sigma \proj{\rho\leq\gamma \sigma}]
			\right\}.
		\end{align}
	\end{coro}
	\begin{proof}
		Here, the second line follows either from $f(x)=x^2-1$ with $f^\star(y) = 1 + \frac{y^2}{4}$, or from the first line by substituting $g\rightarrow g-2$. 
	\end{proof}
	Further examples can be derived, e.g. for the Hellinger divergence we have $f^\star(y)=(\frac{\alpha-1}{\alpha}y)^{\frac{\alpha}{\alpha-1}}+\frac{1}{\alpha-1}$, which generalizes the result above for the $\chi^2$ divergence.

	%\CH{About duality:}
	%\HC{Can we remove the condition $f(0)=0$?} {\color{red}Now I think we can remove the condition $f(0)=0$ in this section}\CH{ Agreed, I removed it.}\\
	%\HC{Can we use integration by parts to change the Riemann--Stieltjes integral back to the normal Riemann integral?} {\color{red}Yes, they are equivalent expressions as indicated in Figures 1, 2.}\\
	%\HC{From the proof, maybe it's better to use $\Tr[\sigma \{\rho \leq \gamma\sigma\}]$ instead of $-\Tr[\sigma \{\rho>\gamma\sigma\}]$?} 
	%Yeah, indeed, struggled.\\
	%\CH{Maybe the comments below we delete and keep for the future?}
	%\HC{Dose Petz's $f$-divergence satisfy similar duality?}
	
	%\HC{However, in the classical case \cite[Theorem 7.27]{Polyanskiy_book2024}, there is way to twist the function $f$ to get Donsker's expression.
		%Can we do that for quantum $f$-divergence as well?}
	
	%\HC{In the classical setting, the order-$1$ Wasserstein distance can be defined as
		%\begin{align}
		%W_1(P,Q) := \sup_{ \|f\|_{\texttt{Lip}} \leq 1 } \mathds{E}_P[f(X)] - \mathds{E}_Q[f(X)],
		%\end{align}
		%where $\|f\|_{\texttt{Lip}} := \sup_{x,y\in\mathds{R}^d} \frac{|f(x)-f(y)|}{\|x-y\|}$.\\
		%Can we do the same:
		%\begin{align}
		%    \sup_{ \|g\|_{\texttt{Lip}} \leq 1} \int_0^\infty g(\gamma) \, \d \Tr[(\sigma-\rho) \{\rho>\gamma\sigma\}]?
		%\end{align}
		%Is it the same as the quantum order-$1$ Wasserstein distance?
		%}
	%\HC{Does this hold \url{https://mathoverflow.net/questions/194803/1-wasserstein-distance-v-s-total-variation-distance}?}
	
	%%%%%%%%%%%%%%%%%%%%%%%%%%%%%%%%%%%%%%%%%%%%%%%%%%%%%%%%%%%%%%%%%%%%%%%%
	
	\section{Conclusions} \label{sec:conclusions}
	
	In this work, we took a new approach towards defining quantum generalizations of classical divergences. These lead us to the recently introduced $f$-divergences defined via integral representations. The new approach allows us to find several new representations and resulting properties. In particular, we prove a trace representation for the \Renyi divergence with $\alpha>1$ conjectured in~\cite{beigi2025some}. As an information theoretic application we give bounds on error exponents in hypothesis testing, unifying and simplifying earlier approaches. 
	
	Different representations of quantum divergences have proven useful at many points in the literature, and we also expect our results to be useful beyond the results given here. Especially, the close connection to hypothesis testing suggest a fundamental nature of the here explored \Renyi divergences. Finding further applications in tighter bounds on information theoretic problems would be interesting.

	%%%%%%%%%%%%%%%%%%%%%%%%%%%%%%%%%%%%%%%%%%%%%%%%%%%%%%%%%%%%%%%%%%%%%%%%%%%%%%%%%%%%%%%%%%%%%%%%%
	\section*{Acknowledgments}
	CH received funding by the Deutsche Forschungsgemeinschaft (DFG, German
	Research Foundation) – 550206990.
	PL and HC are supported under grants 113-2119-M-007-006, 113-2119-M-001-006, NSTC 114-2124-M-002-003, NTU-113V1904-5, NTU-CC-113L891605, NTU-
	113L900702, NTU-114L900702, and NTU-114L895005.

	%%%%%%%%%%%%%%%%%%%%%%%%%%%%%%%%%%%%%%%%%%%%%%
	
	% \newpage
	\appendix
	% \section{Proof of Theorem~\ref{Theo:trace-renyi-PO}}
	\section{Proof of Theorem~\ref{Theo:trace-renyi-PO}} \label{sec:appendix_proof}
	\begin{proof}
		Let $A$ and $B$ be positive semi-definite operators.
		Recall,  
		\begin{align}
			H_{\alpha}(A\Vert B)
			&:= \Tr[A-B] + \alpha \int_1^\infty \gamma^{\alpha-2} E_{\gamma}(A\Vert B) + \gamma^{-\alpha -1} E_{\gamma}(B\Vert A) \d \gamma,
		\end{align}
		where $E_{\gamma}(A\Vert B) := \Tr\left[ (A-B)_+\right]$.
		
		We first simplify the expression for $H_{\alpha}(A\Vert B)$ as follows
		\begin{align}
			H_{\alpha}(A\Vert B) - \Tr[A-B]
			&= \alpha \int_1^\infty \Tr\left[ \gamma^{\alpha-2} (A - \gamma B) \proj{ A > \gamma B} \d \gamma \right] + \gamma^{-\alpha-1} \Tr\left[ (B - \gamma A) \proj{B> \gamma A} \right] \d \gamma.
			\notag
			\\
			\begin{split} \label{eq:order_above_1_H4}
				&= \alpha \int_1^\infty \gamma^{\alpha-2} \Tr\left[ A \proj{A > \gamma B } \right] \d \gamma  - \alpha \int_1^\infty \gamma^{\alpha-1} \Tr\left[ B \proj{A > \gamma B} \right] \d \gamma
				\\ 
				&+ \alpha \int_1^\infty \gamma^{-\alpha-1} \Tr\left[ B \proj{B>\gamma A} \right] \d \gamma - \alpha \int_1^\infty \gamma^{-\alpha} \Tr\left[ A \proj{B>\gamma A} \right] \d \gamma.
			\end{split}
		\end{align}
		
		We simplify the last two terms of \eqref{eq:order_above_1_H4} as follows
		\begin{align}
			\alpha \int_1^\infty \gamma^{-\alpha-1} \Tr\left[ B \proj{B>\gamma A} \right] \d \gamma 
			&=\alpha \int_0^1 \gamma^{\alpha-1} \Tr\left[ B \proj{ \gamma B> A} \right] \d \gamma
			&& ( \gamma \leftarrow \gamma^{-1} )
			\notag
			\\
			&= 
			\alpha \int_0^1 \gamma^{\alpha-1} \Tr[B] \d \gamma
			- \alpha \int_0^1 \gamma^{\alpha-1} \Tr\left[ B \proj{ \gamma B \leq A} \right] \d \gamma
			\notag
			\\
			&=\Tr[B]
			- \alpha \int_0^1 \gamma^{\alpha-1} \Tr\left[ B \proj{ \gamma B < A} \right] \d \gamma
			\label{eq:order_above_1_H4-(3)}
		\end{align}
		and
		\begin{align}
			\alpha \int_1^\infty \gamma^{-\alpha} \Tr\left[ A \proj{B>\gamma A} \right] \d \gamma 
			&=\alpha \int_0^1 \gamma^{\alpha-2} \Tr\left[ A \proj{ \gamma B> A} \right] \d \gamma
			&& ( \gamma \leftarrow \gamma^{-1} )
			\notag
			\\
			&= 
			\alpha \int_0^1 \gamma^{\alpha-2} \Tr[A] \d \gamma
			- \alpha \int_0^1 \gamma^{\alpha-2} \Tr\left[ A \proj{ \gamma B \leq A} \right] \d \gamma
			\notag
			\\
			&=
			\frac{\alpha}{\alpha-1} \Tr[A] - \alpha \int_0^1 \gamma^{\alpha-2} \Tr\left[ A \proj{ \gamma B < A} \right] \d \gamma,
			\label{eq:order_above_1_H4-(4)}
		\end{align}
		
		Substituting \eqref{eq:order_above_1_H4-(3)} and \eqref{eq:order_above_1_H4-(4)} into \eqref{eq:order_above_1_H4}, we obtain
		\begin{align}
			H_{\alpha}(A\Vert B)
			&= \frac{-1}{\alpha-1}\Tr[A]
			+ \alpha \int_0^\infty \gamma^{\alpha-2} \Tr\left[ A \proj{A > \gamma B } \right] \d \gamma  - \alpha \int_0^\infty \gamma^{\alpha-1} \Tr\left[ B \proj{A > \gamma B} \right] \d \gamma.
			\label{eq:order_above_1_H_formula}
		\end{align}
		We rewrite the last two terms as follows.
		Applying operator change of variables (Theorem~\ref{theo:change-of-variable}), we get
		\begin{align}
			\int_0^\infty \Tr\left[ A \proj{ A > \gamma B} \right] \gamma^{\alpha-2} \d \gamma
			&= \int_0^\infty \Tr\left[ A \frac{1}{B+t\I} A \frac{1}{B+t\I} \left( A \frac{1}{B+t\I} \right)^{\alpha-2} 
			\right] \d t
			\notag
			\\
			&= \int_0^\infty \Tr\left[\left( A \frac{1}{B+t\I} \right)^{\alpha} 
			\right] \d t,
			\label{eq:order_above_1_H4-(1)}
		\end{align}
		and
		\begin{align}
			\int_0^\infty \Tr\left[ B \proj{ A > \gamma B} \right] \gamma^{\alpha-1} \d \gamma
			&= \int_0^\infty \Tr\left[ B \frac{1}{B+t\I} A \frac{1}{B+t\I} \left( A \frac{1}{B+t\I} \right)^{\alpha-1} 
			\right] \d t
			\notag
			\\
			&= \int_0^\infty \Tr\left[ B \frac{1}{B+t\I} \left( A \frac{1}{B+t\I} \right)^{\alpha} 
			\right] \d t
			\notag
			\\
			&= \frac{\alpha-1}{\alpha} \int_0^\infty \left( A \frac{1}{B + t\I} \right)^\alpha \d t,
			\label{eq:order_above_1_H4-(2)}
		\end{align}
		where the last line follows from Lemma~\ref{lemm:order_above_1}.
		
		Combining \eqref{eq:order_above_1_H4-(1)} and \eqref{eq:order_above_1_H4-(2)}, we obtain
		\begin{align}
			H_{\alpha}(A\Vert B)
			&= \frac{-1}{\alpha-1}\Tr[A]
			+ \alpha \int_0^\infty \left( A \frac{1}{B + t\I} \right)^\alpha \d t
			- (\alpha-1) \int_0^\infty \left( A \frac{1}{B + t\I} \right)^\alpha \d t
			\notag
			\\
			&= \frac{-1}{\alpha-1}\Tr[A]
			+ \int_0^\infty \left( A \frac{1}{B + t\I} \right)^\alpha \d t.\notag
		\end{align} 
	\end{proof}

	%%%%%%%%%%%%%%%%%%%%%%%%%%%%%%%%%%%%%%%%%%%%%%%%%%%%%%%%%%%%%%%%%%%%%%%%%%%%%%%%%%%%%%%%%%%%%%%%%%%%

	\newpage
	{\larger
		\bibliographystyle{myIEEEtran}
		\bibliography{operator2.bib}

\begin{thebibliography}{10}
\providecommand{\url}[1]{#1}
\csname url@samestyle\endcsname
\providecommand{\newblock}{\relax}
\providecommand{\bibinfo}[2]{#2}
\providecommand{\BIBentrySTDinterwordspacing}{\spaceskip=0pt\relax}
\providecommand{\BIBentryALTinterwordstretchfactor}{4}
\providecommand{\BIBentryALTinterwordspacing}{\spaceskip=\fontdimen2\font plus
\BIBentryALTinterwordstretchfactor\fontdimen3\font minus
  \fontdimen4\font\relax}
\providecommand{\BIBforeignlanguage}[2]{{%
\expandafter\ifx\csname l@#1\endcsname\relax
\typeout{** WARNING: IEEEtran.bst: No hyphenation pattern has been}%
\typeout{** loaded for the language `#1'. Using the pattern for}%
\typeout{** the default language instead.}%
\else
\language=\csname l@#1\endcsname
\fi
#2}}
\providecommand{\BIBdecl}{\relax}
\BIBdecl

\bibitem{AD15}
K.~M.~R. Audenaert and N.~Datta, ``{$\alpha$}-{$z$}-{R{\'{e}}nyi} relative
  entropies,'' \href{http://dx.doi.org/10.1063/1.4906367}{\emph{Journal of
  Mathematical Physics}}, \href{http://dx.doi.org/10.1063/1.4906367}{vol.~56},
  \href{http://dx.doi.org/10.1063/1.4906367}{no.~2},
  \href{http://dx.doi.org/10.1063/1.4906367}{p. 022202},
  \href{http://dx.doi.org/10.1063/1.4906367}{Feb 2015}.

\bibitem{Tom16}
M.~Tomamichel, \emph{Quantum Information Processing with Finite
  Resources}.\hskip 1em plus 0.5em minus 0.4em\relax Springer International
  Publishing, 2016.

\bibitem{Ren61}
A.~R{\'e}nyi, ``On measures of entropy and information,'' \emph{Proc. 4th
  Berkeley Symp. on Math. Statist. Probability}, vol.~1, pp. 547--561, 1961.

\bibitem{Don86}
M.~J. Donald, ``On the relative entropy,''
  \href{http://dx.doi.org/10.1007/bf01212339}{\emph{Communications in
  Mathematical Physics}}, \href{http://dx.doi.org/10.1007/bf01212339}{vol.
  105}, \href{http://dx.doi.org/10.1007/bf01212339}{no.~1},
  \href{http://dx.doi.org/10.1007/bf01212339}{pp. 13--34},
  \href{http://dx.doi.org/10.1007/bf01212339}{Mar 1986}.

\bibitem{HP91}
F.~Hiai and D.~Petz, ``The proper formula for relative entropy and its
  asymptotics in quantum probability,''
  \href{http://dx.doi.org/10.1007/bf02100287}{\emph{Communications in
  Mathematical Physics}}, \href{http://dx.doi.org/10.1007/bf02100287}{vol.
  143}, \href{http://dx.doi.org/10.1007/bf02100287}{no.~1},
  \href{http://dx.doi.org/10.1007/bf02100287}{pp. 99--114},
  \href{http://dx.doi.org/10.1007/bf02100287}{Dec 1991}.

\bibitem{Matsumoto}
K.~Matsumoto, \emph{A New Quantum Version of {$f$}-Divergence}.\hskip 1em plus
  0.5em minus 0.4em\relax Singapore: Springer Singapore, 2018, pp. 229--273.

\bibitem{Hia19}
F.~Hiai, ``Quantum {$f$}-divergences in {von Neumann} algebras. ii. maximal
  {$f$}-divergences,'' \href{http://dx.doi.org/10.1063/1.5051427}{\emph{Journal
  of Mathematical Physics}},
  \href{http://dx.doi.org/10.1063/1.5051427}{vol.~60},
  \href{http://dx.doi.org/10.1063/1.5051427}{no.~1},
  \href{http://dx.doi.org/10.1063/1.5051427}{January 2019}.

\bibitem{Ara76}
H.~Araki, ``Relative entropy of states of {von Neumann} algebras,''
  \href{http://dx.doi.org/10.2977/prims/1195191148}{\emph{Publications of the
  Research Institute for Mathematical Sciences}},
  \href{http://dx.doi.org/10.2977/prims/1195191148}{vol.~11},
  \href{http://dx.doi.org/10.2977/prims/1195191148}{no.~3},
  \href{http://dx.doi.org/10.2977/prims/1195191148}{pp. 809--833},
  \href{http://dx.doi.org/10.2977/prims/1195191148}{December 1976}.

\bibitem{Ara77}
------, ``Relative entropy of states of {von Neumann} algebras. ii,''
  \href{http://dx.doi.org/10.2977/PRIMS/1195190105}{\emph{Publications of the
  Research Institute for Mathematical Sciences}},
  \href{http://dx.doi.org/10.2977/PRIMS/1195190105}{vol.~13},
  \href{http://dx.doi.org/10.2977/PRIMS/1195190105}{no.~1},
  \href{http://dx.doi.org/10.2977/PRIMS/1195190105}{pp. 173--192},
  \href{http://dx.doi.org/10.2977/PRIMS/1195190105}{April 1977}.

\bibitem{Pet85}
``Quasi-entropies for states of a {von Neumann} algebra,''
  \href{http://dx.doi.org/10.2977/PRIMS/1195178929}{\emph{Publ. Res. Inst.
  Math. Sci.}}, \href{http://dx.doi.org/10.2977/PRIMS/1195178929}{pp.
  787--800}, \href{http://dx.doi.org/10.2977/PRIMS/1195178929}{1985}.

\bibitem{Pet86}
D.~Petz, ``Quasi-entropies for finite quantum systems,''
  \href{http://dx.doi.org/10.1016/0034-4877(86)90067-4}{\emph{Reports on
  Mathematical Physics}},
  \href{http://dx.doi.org/10.1016/0034-4877(86)90067-4}{vol.~23},
  \href{http://dx.doi.org/10.1016/0034-4877(86)90067-4}{no.~1},
  \href{http://dx.doi.org/10.1016/0034-4877(86)90067-4}{pp. 57--65},
  \href{http://dx.doi.org/10.1016/0034-4877(86)90067-4}{Feb 1986}.

\bibitem{Petz_book_1993}
M.~Ohya and D.~Petz, \emph{Quantum Entropy and Its Use}.\hskip 1em plus 0.5em
  minus 0.4em\relax Springer Berlin, Heidelberg, 1993.

\bibitem{HMP+11}
F.~Hiai, M.~Mosonyi, D.~Petz, and C.~B{\'{e}}ny, ``Quantum {$f$}-divergences
  and error correction,''
  \href{http://dx.doi.org/10.1142/s0129055x11004412}{\emph{Reviews in
  Mathematical Physics}},
  \href{http://dx.doi.org/10.1142/s0129055x11004412}{vol.~23},
  \href{http://dx.doi.org/10.1142/s0129055x11004412}{no.~07},
  \href{http://dx.doi.org/10.1142/s0129055x11004412}{pp. 691--747},
  \href{http://dx.doi.org/10.1142/s0129055x11004412}{Aug 2011}.

\bibitem{Hia18}
F.~Hiai, ``Quantum {$f$}-divergences in {von Neumann} algebras. i. standard
  {$f$}-divergences,'' \href{http://dx.doi.org/10.1063/1.5039973}{\emph{Journal
  of Mathematical Physics}},
  \href{http://dx.doi.org/10.1063/1.5039973}{vol.~59},
  \href{http://dx.doi.org/10.1063/1.5039973}{no.~10},
  \href{http://dx.doi.org/10.1063/1.5039973}{September 2018}.

\bibitem{PW75}
W.~Pusz and S.~Woronowicz, ``Functional calculus for sesquilinear forms and the
  purification map,''
  \href{http://dx.doi.org/10.1016/0034-4877(75)90061-0}{\emph{Reports on
  Mathematical Physics}},
  \href{http://dx.doi.org/10.1016/0034-4877(75)90061-0}{vol.~8},
  \href{http://dx.doi.org/10.1016/0034-4877(75)90061-0}{no.~2},
  \href{http://dx.doi.org/10.1016/0034-4877(75)90061-0}{pp. 159--170},
  \href{http://dx.doi.org/10.1016/0034-4877(75)90061-0}{October 1975}.

\bibitem{KA80}
F.~Kubo and T.~Ando, ``Means of positive linear operators,''
  \href{http://dx.doi.org/10.1007/bf01371042}{\emph{Mathematische Annalen}},
  \href{http://dx.doi.org/10.1007/bf01371042}{vol. 246},
  \href{http://dx.doi.org/10.1007/bf01371042}{no.~3},
  \href{http://dx.doi.org/10.1007/bf01371042}{pp. 205--224},
  \href{http://dx.doi.org/10.1007/bf01371042}{October 1980}.

\bibitem{BS82}
\BIBentryALTinterwordspacing
V.~P. Belavkin and P.~Staszewski, ``\BIBforeignlanguage{eng}{{$C^*$}-algebraic
  generalization of relative entropy and entropy},''
  \emph{\BIBforeignlanguage{eng}{Annales de l'I.H.P. Physique théorique}},
  vol.~37, no.~1, pp. 51--58, 1982. [Online]. Available:
  \url{http://eudml.org/doc/76163}
\BIBentrySTDinterwordspacing

\bibitem{MDS+13}
M.~M{\"u}ller-Lennert, F.~Dupuis, O.~Szehr, S.~Fehr, and M.~Tomamichel, ``On
  quantum {R{\'e}nyi} entropies: A new generalization and some properties,''
  \href{http://dx.doi.org/10.1063/1.4838856}{\emph{Journal of Mathematical
  Physics}}, \href{http://dx.doi.org/10.1063/1.4838856}{vol.~54},
  \href{http://dx.doi.org/10.1063/1.4838856}{no.~12},
  \href{http://dx.doi.org/10.1063/1.4838856}{p. 122203},
  \href{http://dx.doi.org/10.1063/1.4838856}{2013}.

\bibitem{WWY14}
M.~M. Wilde, A.~Winter, and D.~Yang, ``Strong converse for the classical
  capacity of entanglement-breaking and {Hadamard} channels via a sandwiched
  {R{\'{e}}nyi} relative entropy,''
  \href{http://dx.doi.org/10.1007/s00220-014-2122-x}{\emph{Communications in
  Mathematical Physics}},
  \href{http://dx.doi.org/10.1007/s00220-014-2122-x}{vol. 331},
  \href{http://dx.doi.org/10.1007/s00220-014-2122-x}{no.~2},
  \href{http://dx.doi.org/10.1007/s00220-014-2122-x}{pp. 593--622},
  \href{http://dx.doi.org/10.1007/s00220-014-2122-x}{Jul 2014}.

\bibitem{CG22}
H.-C. Cheng and L.~Gao, ``Error exponent and strong converse for quantum soft
  covering,'' \href{http://dx.doi.org/10.1109/tit.2023.3307437}{\emph{IEEE
  Transactions on Information Theory}},
  \href{http://dx.doi.org/10.1109/tit.2023.3307437}{vol.~70},
  \href{http://dx.doi.org/10.1109/tit.2023.3307437}{no.~5},
  \href{http://dx.doi.org/10.1109/tit.2023.3307437}{pp. 3499--3511},
  \href{http://dx.doi.org/10.1109/tit.2023.3307437}{May 2024}.

\bibitem{CG22b}
\BIBentryALTinterwordspacing
------, ``Tight one-shot analysis for convex splitting with applications in
  quantum information theory,'' 2023. [Online]. Available:
  \url{https://arxiv.org/abs/2304.12055}
\BIBentrySTDinterwordspacing

\bibitem{Jen18}
A.~Jenčová, ``Rényi relative entropies and noncommutative {$L_p$}-spaces,''
  \href{http://dx.doi.org/10.1007/s00023-018-0683-5}{\emph{Annales Henri
  Poincaré}}, \href{http://dx.doi.org/10.1007/s00023-018-0683-5}{vol.~19},
  \href{http://dx.doi.org/10.1007/s00023-018-0683-5}{no.~8},
  \href{http://dx.doi.org/10.1007/s00023-018-0683-5}{pp. 2513--2542},
  \href{http://dx.doi.org/10.1007/s00023-018-0683-5}{June 2018}.

\bibitem{BTS18}
M.~Berta, V.~B. Scholz, and M.~Tomamichel, ``Rényi divergences as weighted
  non-commutative vector-valued {$L_p$}-spaces,''
  \href{http://dx.doi.org/10.1007/s00023-018-0670-x}{\emph{Annales Henri
  Poincaré}}, \href{http://dx.doi.org/10.1007/s00023-018-0670-x}{vol.~19},
  \href{http://dx.doi.org/10.1007/s00023-018-0670-x}{no.~6},
  \href{http://dx.doi.org/10.1007/s00023-018-0670-x}{pp. 1843--1867},
  \href{http://dx.doi.org/10.1007/s00023-018-0670-x}{March 2018}.

\bibitem{Hia21}
F.~Hiai, \emph{Quantum {$f$}-Divergences in {von Neumann} Algebras:
  Reversibility of Quantum Operations}.\hskip 1em plus 0.5em minus 0.4em\relax
  Springer Singapore, 2021.

\bibitem{Kos84}
H.~Kosaki, ``Applications of the complex interpolation method to a {von
  Neumann} algebra: Non-commutative {$L_p$}-spaces,''
  \href{http://dx.doi.org/10.1016/0022-1236(84)90025-9}{\emph{Journal of
  Functional Analysis}},
  \href{http://dx.doi.org/10.1016/0022-1236(84)90025-9}{vol.~56},
  \href{http://dx.doi.org/10.1016/0022-1236(84)90025-9}{no.~1},
  \href{http://dx.doi.org/10.1016/0022-1236(84)90025-9}{pp. 29--78},
  \href{http://dx.doi.org/10.1016/0022-1236(84)90025-9}{March 1984}.

\bibitem{hirche2023quantum}
C.~Hirche and M.~Tomamichel, ``Quantum r\'enyi and $f$-divergences from
  integral representations,''
  \href{http://dx.doi.org/10.1007/s00220-024-05087-3}{\emph{Communications in
  Mathematical Physics}},
  \href{http://dx.doi.org/10.1007/s00220-024-05087-3}{vol. 405},
  \href{http://dx.doi.org/10.1007/s00220-024-05087-3}{no.~9},
  \href{http://dx.doi.org/10.1007/s00220-024-05087-3}{p. 208},
  \href{http://dx.doi.org/10.1007/s00220-024-05087-3}{2024}.

\bibitem{sason2016f}
I.~Sason and S.~Verd{\'u}, ``$ f $-divergence inequalities,''
  \href{http://dx.doi.org/10.1109/TIT.2016.2603151}{\emph{IEEE Transactions on
  Information Theory}},
  \href{http://dx.doi.org/10.1109/TIT.2016.2603151}{vol.~62},
  \href{http://dx.doi.org/10.1109/TIT.2016.2603151}{no.~11},
  \href{http://dx.doi.org/10.1109/TIT.2016.2603151}{pp. 5973--6006},
  \href{http://dx.doi.org/10.1109/TIT.2016.2603151}{2016}.

\bibitem{LM01}
E.~Lieb and M.~Loss, \emph{Analysis}.\hskip 1em plus 0.5em minus 0.4em\relax
  American Mathematical Society, March 2001.

\bibitem{beigi2025some}
S.~Beigi, C.~Hirche, and M.~Tomamichel, ``Some properties and applications of
  the new quantum $ f $-divergences,'' \emph{arXiv preprint arXiv:2501.03799},
  2025.

\bibitem{preparation}
\BIBentryALTinterwordspacing
H.-C. Cheng and P.-C. Liu, ``Error exponents for quantum packing problems via
  an operator layer cake theorem,'' 2025, arXiv:2507.06232 [quant-ph].
  [Online]. Available: \url{http://arxiv.org/abs/2507.06232}
\BIBentrySTDinterwordspacing

\bibitem{Ume62}
H.~Umegaki, ``Conditional expectation in an operator algebra. {IV}. entropy and
  information,'' \href{http://dx.doi.org/10.2996/kmj/1138844604}{\emph{Kodai
  Mathematical Seminar Reports}},
  \href{http://dx.doi.org/10.2996/kmj/1138844604}{vol.~14},
  \href{http://dx.doi.org/10.2996/kmj/1138844604}{no.~2},
  \href{http://dx.doi.org/10.2996/kmj/1138844604}{pp. 59--85},
  \href{http://dx.doi.org/10.2996/kmj/1138844604}{1962}.

\bibitem{frenkel2022integral}
P.~E. Frenkel, ``Integral formula for quantum relative entropy implies data
  processing inequality,'' 2022.

\bibitem{ACM+07}
K.~M.~R. Audenaert, J.~Calsamiglia, R.~Mu{\~{n}}oz-Tapia, E.~Bagan, L.~Masanes,
  A.~Acin, and F.~Verstraete, ``Discriminating states: The quantum {Chernoff}
  bound,''
  \href{http://dx.doi.org/10.1103/physrevlett.98.160501}{\emph{Physical Review
  Letters}}, \href{http://dx.doi.org/10.1103/physrevlett.98.160501}{vol.~98},
  \href{http://dx.doi.org/10.1103/physrevlett.98.160501}{p. 160501},
  \href{http://dx.doi.org/10.1103/physrevlett.98.160501}{Apr 2007}.

\bibitem{ANS+08}
K.~M.~R. Audenaert, M.~Nussbaum, A.~Szko{\l}a, and F.~Verstraete, ``Asymptotic
  error rates in quantum hypothesis testing,''
  \href{http://dx.doi.org/10.1007/s00220-008-0417-5}{\emph{Communications in
  Mathematical Physics}},
  \href{http://dx.doi.org/10.1007/s00220-008-0417-5}{vol. 279},
  \href{http://dx.doi.org/10.1007/s00220-008-0417-5}{no.~1},
  \href{http://dx.doi.org/10.1007/s00220-008-0417-5}{pp. 251--283},
  \href{http://dx.doi.org/10.1007/s00220-008-0417-5}{Feb 2008}.

\bibitem{Lie73}
E.~H. Lieb, ``Convex trace functions and the {Wigner-Yanase-Dyson}
  conjecture,''
  \href{http://dx.doi.org/10.1016/0001-8708(73)90011-x}{\emph{Advances in
  Mathematics}},
  \href{http://dx.doi.org/10.1016/0001-8708(73)90011-x}{vol.~11},
  \href{http://dx.doi.org/10.1016/0001-8708(73)90011-x}{no.~3},
  \href{http://dx.doi.org/10.1016/0001-8708(73)90011-x}{pp. 267--288},
  \href{http://dx.doi.org/10.1016/0001-8708(73)90011-x}{Dec 1973}.

\bibitem{SBT16}
D.~Sutter, M.~Berta, and M.~Tomamichel, ``Multivariate trace inequalities,''
  \href{http://dx.doi.org/10.1007/s00220-016-2778-5}{\emph{Communications in
  Mathematical Physics}},
  \href{http://dx.doi.org/10.1007/s00220-016-2778-5}{vol. 352},
  \href{http://dx.doi.org/10.1007/s00220-016-2778-5}{no.~1},
  \href{http://dx.doi.org/10.1007/s00220-016-2778-5}{pp. 37--58},
  \href{http://dx.doi.org/10.1007/s00220-016-2778-5}{May 2017}.

\bibitem{sharma2012strong}
N.~Sharma and N.~A. Warsi, ``On the strong converses for the quantum channel
  capacity theorems,'' 2012.

\bibitem{hirche2023QDP}
C.~Hirche, C.~Rouz{\'e}, and D.~S. Fran{\c{c}}a, ``Quantum differential
  privacy: An information theory perspective,'' \emph{IEEE Transactions on
  Information Theory}, vol.~69, no.~9, pp. 5771--5787, 2023.

\bibitem{cheng2024sample}
H.-C. Cheng, C.~Hirche, and C.~Rouz{\'e}, ``Sample complexity of locally
  differentially private quantum hypothesis testing,'' in \emph{2024 IEEE
  International Symposium on Information Theory (ISIT)}.\hskip 1em plus 0.5em
  minus 0.4em\relax IEEE, 2024, pp. 2921--2926.

\bibitem{Kat95}
T.~Kato, \emph{Perturbation Theory for Linear Operators}.\hskip 1em plus 0.5em
  minus 0.4em\relax Springer Berlin, Heidelberg, 1995.

\bibitem{HP14}
F.~Hiai and D.~Petz, \emph{Introduction to Matrix Analysis and
  Applications}.\hskip 1em plus 0.5em minus 0.4em\relax Springer International
  Publishing, 2014.

\bibitem{polyanskiy2010channel}
Y.~Polyanskiy, \emph{Channel coding: Non-asymptotic fundamental limits}.\hskip
  1em plus 0.5em minus 0.4em\relax Princeton University, 2010.

\bibitem{mosonyi2015quantum}
M.~Mosonyi and T.~Ogawa, ``Quantum hypothesis testing and the operational
  interpretation of the quantum r{\'e}nyi relative entropies,''
  \emph{Communications in Mathematical Physics}, vol. 334, pp. 1617--1648,
  2015.

\bibitem{NS09}
M.~Nussbaum and A.~Szko{\l}a, ``The {Chernoff} lower bound for symmetric
  quantum hypothesis testing,''
  \href{http://dx.doi.org/10.1214/08-aos593}{\emph{Annals of Statistics}},
  \href{http://dx.doi.org/10.1214/08-aos593}{vol.~37},
  \href{http://dx.doi.org/10.1214/08-aos593}{no.~2},
  \href{http://dx.doi.org/10.1214/08-aos593}{pp. 1040--1057},
  \href{http://dx.doi.org/10.1214/08-aos593}{Apr 2009}.

\bibitem{Polyanskiy_book2024}
Y.~Polyanskiy and Y.~Wu, \emph{Information Theory: From Coding to
  Learning}.\hskip 1em plus 0.5em minus 0.4em\relax Cambridge University Press,
  December 2024.

\bibitem{Pet88}
D.~Petz, ``A variational expression for the relative entropy,''
  \href{http://dx.doi.org/10.1007/bf01225040}{\emph{Communications in
  Mathematical Physics}}, \href{http://dx.doi.org/10.1007/bf01225040}{vol.
  114}, \href{http://dx.doi.org/10.1007/bf01225040}{no.~2},
  \href{http://dx.doi.org/10.1007/bf01225040}{pp. 345--349},
  \href{http://dx.doi.org/10.1007/bf01225040}{June 1988}.

\end{thebibliography}
	}
	
\end{document}